\documentclass[a4paper]{easychair}

\usepackage{listings} 
\usepackage{subfig}
\usepackage{pstricks} 
\usepackage{pst-node}

\usepackage{amssymb}
\usepackage{amsmath}
\usepackage{mathtools}
\usepackage{ifthen} 
\usepackage{xspace} 
\usepackage{wasysym} 

\newcommand\etal{~\textit{et al.}\xspace}

\newcommand\setcomp[2]{\left\{{#1}\ \left|\ {#2}\right.\right\}}
\newcommand\set[1]{\left\{#1\right\}}
\newcommand\tup[1]{\brac{#1}}
\newcommand\nats{\mathbb{N}}
\newcommand\poss{\nats_{+}}
\newcommand\brac[1]{\left({#1}\right)}
\newcommand\ap[2]{{#1}\mathord{\brac{#2}}}
\newcommand\idxi{i}
\newcommand\idxj{j}
\newcommand\idxe{e}
\newcommand\num{m} 
\newcommand\varnum{n} 
\newcommand\vecnums{\vec{v}}
\newcommand\varvecnums{\vec{u}}
\newcommand\eword{\varepsilon}
\newcommand\lang{\mathcal{L}}
\newcommand\cat{\cdot}
\newcommand\word{w}
\newcommand\bit{b}
\newcommand\zerovec{\onevec{0}}
\newcommand\onevec[1]{\vec{#1}}
\newcommand\sizeof[1]{\left|{#1}\right|}
\newcommand\minimumof[1]{\ap{min}{#1}}

\newcommand\acker{\mathbf{F}_{\omega}}

\newcommand\fmla{\varphi}

\newcommand\rega{\mathcal{A}}
\newcommand\parikh[1]{\ap{\textsc{Parikh}}{#1}}
\newcommand\countof[2]{\left|{#1}\right|_{#2}}
\newcommand\nfastates{\mathcal{Q}}
\newcommand\nfatrans{\Delta}

\newcommand\nfastate{q}
\newcommand\nfatran[1]{\xrightarrow{#1}}
\newcommand\nfafins{\mathcal{F}}

\newcommand\tree{T}
\newcommand\cha{a}
\newcommand\chb{b}
\newcommand\chc{c}
\newcommand\chd{d}
\newcommand\alphabet{\Sigma}
\newcommand\rank{\mathtt{rank}}
\newcommand\tdom{D}
\newcommand\tnode{n}
\newcommand\tlab{\lambda}
\newcommand\trees[1]{\ap{\text{\sc Trees}}{#1}}
\newcommand\tap[2]{\ap{#1}{#2}}
\newcommand\emptree{\mathcal{E}}

\newcommand\context{C}
\newcommand\cvar{x}
\newcommand\csub[2]{{#1}\mathord{\left[{#2}\right]}}

\newcommand{\ta}{\mathcal{T}}
\newcommand\tastates{\mathcal{Q}}
\newcommand\tarules{\Delta}
\newcommand\tafinals{\mathcal{F}}
\newcommand\tarule[3]{\ifthenelse{\equal{#1}{}}{}{\tup{#1}} \xrightarrow{#2} {#3}}
\newcommand\tast{q}
\newcommand\tarun{\rho}
\newcommand\tasingle[1]{\ta_{#1}}

\newcommand\pn{\mathcal{N}}
\newcommand\counters{X}
\newcommand\ctr{x}
\newcommand\incr[1]{\ap{incr}{#1}}
\newcommand\decr[1]{\ap{decr}{#1}}
\newcommand\reset[1]{\ap{res}{#1}}
\newcommand\pnstates{\nfastates}
\newcommand\cops[1]{{\textsc{Op}_{#1}}}
\newcommand\pnmarking{\pi}
\newcommand\pnrules{\Delta}
\newcommand\pnstate{\nfastate}

\newcommand\pnops{\widetilde{o}}
\newcommand\pnrule[3]{{#1} \pnlabtran{#2} {#3}}
\newcommand\pntran{\longrightarrow}
\newcommand\pnlabtran[1]{\xrightarrow[#1]{}}
\newcommand\pnrun{\longrightarrow^\ast}
\newcommand\pncovers{\leq}
\newcommand\pnzeromarking{\pnmarking_0}

\newcommand\gtrs{G}
\newcommand\controls{\mathcal{P}}
\newcommand\rules{\mathcal{R}}
\newcommand\control{p}
\newcommand\rrule[4]{\tup{{#1}, {#2}} \longrightarrow \tup{{#3}, {#4}}}
\newcommand\orrule[5]{\tup{{#1}, {#2}} \xrightarrow{#3} \tup{{#4}, {#5}}}
\newcommand\rruler{r}
\newcommand\config[2]{\tup{{#1}, {#2}}}
\newcommand\tran{\longrightarrow}
\newcommand\run{\tran^\ast}
\newcommand\lifespan{k}
\newcommand\numrules{\ell}
\newcommand\oalphabet{\Gamma}
\newcommand\osym{\gamma}
\newcommand\otran[1]{\xrightarrow{#1}}
\newcommand\empsym{\varepsilon}
\newcommand\orun[1]{\xrightarrow{#1}}
\newcommand\runsym{\rho}

\newcommand\treeout{\tree^{\text{out}}}
\newcommand\treein{\tree^{\text{in}}}

\newcommand\hasid[1]{\ap{has}{#1}}
\newcommand\spawnnode{\ast}
\newcommand\splitnode{\bullet} 
\newcommand\errorcontrol{\control_{E}}
\newcommand\spawncontrol[3]{\tup{{#1},\tup{{#2},{#3}}}}
\newcommand\spawncase[4]{\tup{{#1},{#2},{#3},{#4}}_1}
\newcommand\spawncasenext[4]{\tup{{#1},{#2},{#3},{#4}}_2}
\newcommand\threadstate[2]{\tup{{#1},{#2}}}

\newcommand\controlinit{\control_{src}} 
\newcommand\controldest{\control_{snk}}
\newcommand\treeinit{\tree_{src}}

\newcommand\istiseq{\alpha}
\newcommand\istitrees{\vec{\istigen}}
\newcommand\istigen{\eta}
\newcommand\exttree{\bigtriangleup}
\newcommand\istigtrs{\gtrs_I}
\newcommand\istigtrscontrols{\controls_I}
\newcommand\istigtrsalphabet{\alphabet_I}
\newcommand\istigtrsrules{\rules_I}
\newcommand\istigtrsinitlab{\Diamond}
\newcommand\istigtrsoalphabet{\oalphabet_I}
\newcommand\istirega{\rega_I} 
\newcommand\intadd[2]{{#1}\ \mathop{\mathord{+}\mathord{+}}\ {#2}}
\newcommand\intres[1]{\ap{\textsc{Res}}{#1}}
\newcommand\istictr[2]{\ctr^{#1}_{#2}} 
\newcommand\pnreach[1]{\pn_{#1}}
\newcommand\istifulla[2]{\rega^{#1}_{#2}}
\newcommand\istiinitst[2]{\nfastate^{#1}_{#2}}
\newcommand\istifinst[2]{f^{#1}_{#2}}
\newcommand\istitrans[2]{\nfatrans^{#1}_{#2}}
\newcommand\ististates[2]{\nfastates^{#1}_{#2}}
\newcommand\pnreachctrs[1]{\counters_{\gtrs}}
\newcommand\pnreachstates[1]{\pnstates_{\gtrs}}
\newcommand\pnreachrules[1]{\pnrules_{\gtrs}}
\newcommand\pnsumseqs{\mathcal{S}}
\newcommand\pnshiftst[2]{\lhd^{#1}_{#2}}
\newcommand\pnsumseq{\beta} 
\newcommand\pnrulesadd{\pnrules_{\textsc{Add}}} 
\newcommand\pnrulesres{\pnrules_{\textsc{Res}}}
\newcommand\pninitmarking{\pnmarking_{src}}
\newcommand\runlen{h}
\newcommand\istid[1]{id_{#1}}
\newcommand\istrun[1]{run_{#1}}
\newcommand\istint[1]{int_{#1}}
\newcommand\tnodeparent{\tnode_p}
\newcommand\idxjnew{\idxj_{\ast}}

\newcommand\submarking[2]{{#1} - {#2}} 
\newcommand\treeseq[1]{\vec{\tree}_{#1}}
\newcommand\pntree[3]{\tree^{#3}_{{#1}, {#2}}}
\newcommand\pntreealt[3]{Y^{#3}_{{#1}, {#2}}}

\newcommand\gtrscover[1]{\gtrs_{#1}} 
\newcommand\gtrsreach[1]{\gtrs'_{#1}} 
\newcommand\killst[2]{\times^{#1}_{#2}}
\newcommand\deadnode{\rightmoon}
\newcommand\targetctr[1]{\overline{#1}}
\newcommand\varcontroldest{\control'_{snk}}

\newcommand\pdsalphabet{\alphabet}
\newcommand\numstacks{z} 
\newcommand\scopebound{k}
\newcommand\pds{\mathbb{P}}
\newcommand\pdscontrols{\controls}
\newcommand\pdsrules[1]{\rules_{#1}}
\newcommand\pdspushrules[1]{\rules^{push}_{#1}}
\newcommand\pdsintrules[1]{\rules^{int}_{#1}}
\newcommand\pdspoprules[1]{\rules^{pop}_{#1}}

\newcommand\pdscontrol{\control}
\newcommand\pdscontrolinit{\controlinit}
\newcommand\pdscontroldest{\controldest}
\newcommand\sbot{\bot}
\newcommand\pdschainit{\sbot}
\newcommand\pdsitran[1]{\pdstran_{#1}}
\newcommand\pdstran{\longrightarrow}
\newcommand\pdsirun[1]{\pdsitran{#1}^\ast}
\newcommand\pdsround{\longrightarrow_{R}}
\newcommand\pdsstack{\word}
\newcommand\pdspushrule[3]{\tup{{#1}, {#2}, {#3}}}
\newcommand\pdsintrule[2]{\tup{{#1}, {#2}}}
\newcommand\pdspoprule[3]{\tup{{#1}, {#2}, {#3}}}

\newcommand\pdsconfsym{\sigma}

\newcommand\gtrsscoped[1]{\gtrs_{#1}}
\newcommand\gtrsscopedcontrols{\controls_S}
\newcommand\gtrsscopedalphabet{\alphabet_S}
\newcommand\gtrsscopedrules{\rules_S}
\newcommand\stoppednode[1]{\Box_{#1}}
\newcommand\gtrsscopedrulesbeginend{\rules_{src/snk}}
\newcommand\gtrsscopedrulessim{\rules_{PDS}}
\newcommand\gtrsscopedrulesswitch{\rules_{Switch}}

\newtheorem{theorem}{Theorem}[section]
\newtheorem{definition}{Definition}[section]

\newtheorem{lemma}{Lemma}[section] 
\newtheorem{property}{Property}[section] 

\newenvironment{namedtheorem}[2]{%
    \expandafter\gdef\csname reftheorem#1\endcsname{%
        Theorem~\ref{#1} (#2)%
    }%
    \begin{theorem}[{#2}] \label{#1}%
}{%
    \end{theorem}%
}
\newcommand\reftheorem[1]{\expandafter\csname reftheorem#1\endcsname}

\newenvironment{namedlemma}[2]{%
    \expandafter\gdef\csname reflemma#1\endcsname{%
        Lemma~\ref{#1} (#2)%
    }%
    \begin{lemma}[{#2}] \label{#1}%
}{%
    \end{lemma}%
}
\newcommand\reflemma[1]{\expandafter\csname reflemma#1\endcsname}

\newcommand\refproperty[1]{\expandafter\csname refproperty#1\endcsname}

\newenvironment{nameddefinition}[2]{%
    \expandafter\gdef\csname refdefinition#1\endcsname{%
        Definition~\ref{#1} (#2)%
    }%
    \begin{definition}[{#2}] \label{#1}%
}{%
    \end{definition}%
}
\newcommand\refdefinition[1]{\expandafter\csname refdefinition#1\endcsname}

\newif\ifdraft\drafttrue

\usepackage{pstricks}

\ifdraft

\newcommand\todo[1]{{\color{purple} [\textbf{To do:} #1]}}

\newcommand\al[1]{{\color{blue} [#1 - \textbf{Olivier}]}}
\newcommand\mh[1]{{\color{orange}
[#1 - \textbf{Matt}]}}

\else

\newcommand\todo[1]{}
\newcommand\al[1]{}
\newcommand\mh[1]{}

\fi

\title{Senescent Ground Tree Rewrite Systems} 
\titlerunning{Senescent Ground Tree Rewrite Systems} 
\author{M. Hague}

\institute{
  Royal Holloway University of London \\
  \email{matthew.hague@rhul.ac.uk} 
}

\authorrunning{M. Hague}

\begin{document}

\maketitle

\begin{abstract}
    
Ground Tree Rewrite Systems with State are known to have an undecidable control
state reachability problem.  Taking inspiration from the recent introduction of
scope-bounded multi-stack pushdown systems, we define \emph{Senescent Ground
Tree Rewrite Systems}.  These are a restriction of ground tree rewrite systems
with state such that nodes of the tree may no longer be rewritten after having
witnessed an \textit{a priori} fixed number of control state changes.  As well
as generalising scope-bounded multi-stack pushdown systems, we show --- via
reductions to and from reset Petri-nets --- that these systems have an
Ackermann-complete control state reachability problem.  However, reachability of
a regular set of trees remains undecidable.

\end{abstract}

\tableofcontents

\section{Introduction}

The study of reachability problems for infinite state systems, such as Turing
machines, has often used strings to represent system states.  In seminal work,
B\"uchi showed the decidability of reachability for pushdown systems~\cite{B64}.
A state (or configuration) of a pushdown system is represented by a control
state (from a finite set) and a stack over a given finite alphabet.  In this
case, a stack can be considered a word and one stack is obtained from another by
replacing a prefix $\word$ of the stack with another word $\word'$.  In fact,
the reachability problem is in P-time~\cite{BEM97,FWW97}.  

Pushdown systems allow the control-flow of first-order programs to be accurately
modelled~\cite{JM77} and have been well-studied in the automata-theoretic
approach to software model checking (E.g.~\cite{BEM97,FWW97,EKS03,RSJM05}).
Many scalable model checkers for pushdown systems have been implemented, and
these tools (e.g.  Bebop~\cite{BR00} and Moped~\cite{S02b}) are an essential
back-end component of celebrated model checkers such as SLAM~\cite{BLR11}.

A natural and well-studied generalisation of these ideas is to use a tree
representation of system states.  This approach was first considered by
Brainerd, who, generalising B\"uchi's result, showed decidability of
reachability for Ground Tree Rewrite Systems\footnote{Also known as Ground Term
Rewrite Systems} (GTRS)~\cite{B69}.  In these systems, each transition replaces
a complete subtree of the state with another.  Thus, these systems generalise
pushdown systems and allow the analysis of tree manipulating programs.  As in
the pushdown case, reachability is solvable in P-time~\cite{L03}.

Unfortunately, these tree generalisations of pushdown automata do not allow a
control state in their configurations.  Instead, the pushdown system's control
state must be encoded as the leaf of the tree.  This is for good reason: when a
control state external to the tree is permitted, reachability is immediately
undecidable.  This is because one can easily simulate a two-stack pushdown
system with a tree that contains a branch for each stack.  It is well known that
a two-stack pushdown system can simulate a Turing machine, and thus reachability
is undecidable.  

However, due to the increasing importance of concurrent systems --- where each
thread requires its own stack --- there has been renewed interest in identifying
classes of multi-stack pushdown systems for which reachability becomes
decidable.  A seminal notion in this regard is that of \emph{context-bounding}.
This underapproximates a concurrent system by bounding the number of context
switches that may occur~\cite{QR05}.  It is based on the observation that most
real-world bugs require only a small number of thread interactions~\cite{Q08}.
By considering only context-bounded runs of a multi-stack pushdown system, the
reachability problem becomes NP-complete.

In recent work~\cite{T10}, Lin (formerly known as To) observed that a GTRS
modelling a context-bounded multi-stack pushdown system has an underlying
control state graph that is $1$-weak.  A $1$-weak automaton is an automaton
whose control state graph contains no cycles except for self-loops~\cite{M00}.
Intuitively, the control state is only used to manage context-switches, and
hence a $1$-weak control state graph suffices.  Moreover, the reachability
problem for GTRS with control states is decidable with such a control state
graph~\cite{T10}.  Indeed, the problem remains NP-complete~\cite{L12}.

As well as the notion of context-bounding there are many more relaxed
restrictions on multi-stack pushdown behaviours for which reachability also
remains decidable.  Of particular interest is
\emph{scope-bounding}~\cite{lTN11}.  In this setting, we fix a bound $\lifespan$
and insist that an item may only be removed from the stack if it was pushed at
most $\lifespan$ context switches earlier.  Thus, an arbitrary number of context
switches may occur, and the underlying control state graph is no longer
$1$-weak.  In this case, by relaxing the restriction on the control state
behaviours, the complexity of reachability is increased from NP to PSPACE.

In this work we study how to generalise scope-bounding to GTRS with control
states.  We obtain a model of computation reminiscent of a tree growing in
nature: it begins with a green shoot, which may grow and change.  As this shoot
ages, it becomes hardened and forms the trunk of the tree.  From this trunk, new
green shoots grow, and --- via leaves that may fall and grow again --- remain
changeable.  If a shoot lives long enough,  it hardens and forms a new (fixed)
branch of the tree.

Thus, we define \emph{senescent ground tree rewrite systems}.  The passage of
time is marked by changes to the control state.  If a node remains unchanged for
a fixed number $\lifespan$ of changes, it becomes unchangeable -- that is, part
of a hardened branch of the tree.  

These systems naturally generalise scope-bounded pushdown systems while also
allowing for additional features such as dynamic thread creation.  To our
knowledge, they also provide the most precise under-approximation of GTRS with
control states currently known to have a decidable control state reachability
problem, and thus may be used in the analysis of tree-manipulating programs.

We show, via inter-reductions with reset Petri-nets, that the control state
reachability problem for senescent GTRS is Ackermann-complete while the
reachability of a regular set of trees is undecidable.  This increase in
modelling power is in sharp contrast to the analogous restrictions for pushdown
systems, where the increase is much more modest.  That is, from The Faithful
Gardner~\cite{E95}: 
\begin{quote}
    \textit{To be poor and be without trees, is to be the most starved human
    being in the world. To be poor and have trees, is to be completely rich in
    ways that money can never buy}.
\end{quote}

\section{Related Work}

Abdulla\etal~\cite{AJMdO02} define a regular model checking algorithm for tree
automatic structures.  That is, the transition relation is given by a regular
tree transducer.  They give a reachability algorithm that is complete when there
is a fixed bound $\lifespan$ such that, during any run of the system, a node is
changed at most $\lifespan$ times.  This has flavours of the systems we define
here.  However, in our model, a node may be changed an arbitrary number of
times.  In fact, we impose a limit on the extent to which a node may remain
\emph{unchanged}.  It is not clear how the two models compare, and such a
comparison is an interesting avenue of future work.

Atig\etal~\cite{ABQ11} consider a model of multi-stack pushdown systems with
dynamic thread creation.  Decidability of the reachability problem is obtained
by allowing each thread to be active at most $\lifespan$ times during a run.  As
we show in Section~\ref{sec:threads-are-branches}, we can consider each thread
to be a branch of the tree, and context switches correspond to control state
changes.  However, while in Atig\etal's model a thread may be active for
\emph{any} $\lifespan$ context switches (as long as it's inactive for the
others), in our model a thread will begin to suffer restrictions after the
$\lifespan$ next context switches (though may be active for an arbitrary
number).  Perhaps counter-intuitively, our restriction actually increases
expressivity: Atig\etal show inter-reducibility between reachability in their
model and Petri-net coverability, while in our model the more severe restriction
allows us to inter-reduce with coverability of \emph{reset} Petri-nets.

The scope-bounded restriction has recently been relaxed for multi-stack pushdown
systems by La Torre and Napoli~\cite{lTN12}.  In their setting, a character that
is popped from a particular stack may must have been pushed within $\lifespan$
active contexts of that stack.  In particular, this allows for an unbounded
number of context switches to occur between a push and a pop, as long as the
stack involved is only active in up to $\lifespan$ of those contexts.  However,
it is unclear what such a relaxation would mean in the context of senescent
GTRS.

GTRS have been extensively studied as generators of graphs (e.g.~\cite{L03}) and
and are known to have decidable verification problems for repeated
reachability~\cite{L03}, first-order logic~\cite{DT90},
confluence~\cite{DHLT90}, \&c.  However, LTL and CTL model checking are
undecidable~\cite{BKRS09,GL11,L03}.  They also have intimate connections
(e.g.~\cite{L03,GL11}) with the Process Rewrite Systems Hierarchy~\cite{M98}.  

There are several differing restrictions to multi-stack pushdown systems with
decidable verification problems.  Amongst these are phase-bounded~\cite{lTMP07}
and ordered~\cite{BCCC96} (corrected in~\cite{ABH08}) pushdown systems.  There
are also generic frameworks --- that bound the tree-~\cite{MP11} or
split-width~\cite{CGK12} of the interactions between communication and storage
--- that give decidability for all communication architectures that can be
defined within them.

\section{Preliminaries}

We write $\nats$ to denote the set of natural numbers and $\poss$ to denote the
set of strictly positive natural numbers.  Given a word language $\lang
\subseteq \alphabet^\ast$ for some alphabet $\alphabet$ and a word $\word \in
\alphabet^\ast$, let $\word \cat \lang = \setcomp{\word \word'}{\word' \in
\lang}$.  For a given set $S$, let $\sizeof{S}$ denote the cardinality of the
set.

In the cases when the dimension is clear, we will write $\zerovec$ to denote the
tuple $\tup{0, \ldots, 0}$ and $\onevec{\idxi}$ to denote the tuple
$\tup{\varnum_1, \ldots, \varnum_\num}$ where all $\varnum_\idxj= 0$ for all
$\idxj \neq \idxi$ and $\varnum_\idxi = 1$.  

We will denote by $\acker$ both the Ackermann function and the class of problems
solvable in $\acker$-time.  Following Schmitz and Schnoebelen~\cite{SS12}, we
have the class $\acker$ of problems computable in Ackermannian time.  This class
is closed under primitive-recursive reductions.

\subsection{Trees and Automata}

\subsubsection{Regular Automata and Parikh Images}

A \emph{regular automaton} is a tuple $\rega = \tup{\nfastates, \oalphabet,
\nfatrans, \nfastate_0, \nfafins}$ where $\nfastates$ is a finite set of states,
$\oalphabet$ is a finite output alphabet, $\nfatrans \subseteq \nfastates \times
\oalphabet \times \nfastates$ is a transition relation, $\nfastate_0 \in
\nfastates$ is an initial state and $\nfafins \subseteq \nfastates$ is a set of
final states.

We write $\nfastate \nfatran{\cha} \nfastate'$ to denote a transition
$\tup{\nfastate, \cha, \nfastate'} \in \nfatrans$.  A run from $\nfastate_1 \in
\nfastates$ over a word $\word = \cha_1 \ldots \cha_\runlen$ is a sequence 
\[
    \nfastate_1 \nfatran{\cha_1} \cdots \nfatran{\cha_\runlen}
    \nfastate_{\runlen+1} \ .
\]
A run is \emph{accepting} whenever $\nfastate_{\runlen+1} \in \nfafins$.  The
language $\ap{\lang}{\rega}$ of $\rega$ is the set of words $\word \in
\oalphabet^\ast$ such that there is an accepting run of $\rega$ over $\word$
from $\nfastate_0$.

For a word $\word \in \oalphabet^\ast$ for some alphabet $\oalphabet$, we define
$\countof{\word}{\osym}$ to be the number of occurrences of $\osym$ in $\word$.
Given a fixed linear ordering $\osym_1, \ldots, \osym_\num$ over $\oalphabet =
\set{\osym_1, \ldots, \osym_\num}$ and a word $\word \in \oalphabet^\ast$, we
define $\parikh{\word} = \tup{\countof{\word}{\osym_1}, \ldots,
\countof{\word}{\osym_\num}}$.  Given a language $\lang \subseteq
\oalphabet^\ast$, we define $\parikh{\lang} = \setcomp{\parikh{\word}}{\word \in
\lang}$.  Finally, given a regular automaton $\rega$, we define $\parikh{\rega}
= \parikh{\ap{\lang}{\rega}}$.

\subsubsection{Trees}

A \emph{ranked alphabet} is a finite set of characters $\alphabet$ together with
a rank function $\rank : \alphabet \mapsto \nats$.  A \emph{tree domain} $\tdom
\subset \poss^\ast$ is a nonempty finite subset of $\poss^\ast$ that is both
\emph{prefix-closed} and \emph{younger-sibling-closed}.  That is, if $\tnode
\idxi \in \tdom$, then we also have $\tnode \in \tdom$ and, for all $1 \leq
\idxj \leq \idxi$, $\tnode \idxj \in \tdom$ (respectively).  A \emph{tree} over
a ranked alphabet $\alphabet$ is a pair $\tree = \tup{\tdom, \tlab}$ where
$\tdom$ is a tree domain and $\tlab : \tdom \mapsto \alphabet$ such that for all
$\tnode \in \tdom$, if $\ap{\tlab}{\tnode} = \cha$ and $\ap{\rank}{\cha} = \num$
then $\tnode$ has exactly $\num$ children (i.e. $\tnode \num \in \tdom$ and
$\tnode (\num + 1) \notin \tdom$).  Let $\trees{\alphabet}$ denote the set of
trees over $\alphabet$.

Given a node $\tnode$ and trees $\tree_1, \ldots, \tree_\num$, we will often
write $\tap{\tnode}{\tree_1, \ldots, \tree_\num}$ to denote the tree with root
node $\tnode$ and left-to-right child sub-trees $\tree_1, \ldots, \tree_\num$.
When $\tnode$ is labelled $\cha$, we may also write $\tap{\cha}{\tree_1,
\ldots, \tree_\num}$ to denote the same tree.  We will often simply write $\cha$
to denote the tree with a single node labelled $\cha$.  Finally, let $\emptree$
denote the empty tree.

\subsubsection{Context Trees}

A \emph{context tree} over the alphabet $\alphabet$ with context variables
$\cvar_1, \ldots, \cvar_\num$ is a tree $\context = \tup{\tdom, \tlab}$ over
$\alphabet \uplus \set{\cvar_1, \ldots, \cvar_\num}$ such that for each $1 \leq
\idxi \leq \num$ we have $\ap{\rank}{\cvar_\idxi} = 0$ and there exists a unique
\emph{context node} $\tnode_\idxi$ such that $\ap{\tlab}{\tnode_\idxi} =
\cvar_\idxi$.  We will denote such a tree $\csub{\context}{\cvar_1, \ldots,
\cvar_\num}$.  

Given trees $\tree_\idxi = \tup{\tdom_\idxi, \tlab_\idxi}$ for each $1 \leq
\idxi \leq \num$, we denote by $\csub{\context}{\tree_1, \ldots, \tree_\num}$
the tree $\tree'$ obtained by filling each variable $\cvar_\idxi$ with the tree
$\tree_\idxi$.  That is, $\tree' = \tup{\tdom', \tlab'}$ where $\tdom' = \tdom
\cup \tnode_1 \cat \tdom_1 \cup \cdots \cup \tnode_\num \cat \tdom_\num$ and 
\[
    \ap{\tlab'}{\tnode} = 
    \begin{cases}
        \ap{\tlab}{\tnode} & \tnode \in \tdom \land \forall \idxi . \tnode \neq
        \tnode_\idxi \\

        \ap{\tlab_\idxi}{\tnode'} & \tnode = \tnode_\idxi \tnode' \ .
    \end{cases} 
\]

\subsubsection{Tree Automata}

A \emph{bottom-up nondeterministic tree automaton} (NTA) over a ranked alphabet
$\alphabet$ is a tuple $\ta = \tup{\tastates, \tarules, \tafinals}$ where
$\tastates$ is a finite set of states, $\tafinals \subseteq \tastates$ is a set
of final (accepting) states, and $\tarules$ is a finite set of rules of the form
$\tarule{\tast_1, \ldots, \tast_\num}{\cha}{\tast}$ where $\tast_1, \ldots,
\tast_\num, \tast \in \tastates$, $\cha \in \alphabet$ and $\ap{\rank}{\cha} =
\num$.  A \emph{run} of $\ta$ on a tree $\tree = \tup{\tdom, \tlab}$ is a
mapping $\tarun : \tdom \mapsto \tastates$ such that for all $\tnode \in \tdom$
labelled $\ap{\tlab}{\tnode} = \cha$ with $\ap{\rank}{\cha} = \num$ we have
$\tarule{\ap{\tarun}{\tnode 1}, \ldots, \ap{\tarun}{\tnode
\num}}{\cha}{\ap{\tarun}{\tnode}}$.  It is accepting if $\ap{\tarun}{\eword} \in
\tafinals$.  The \emph{language} defined by a tree automaton $\ta$ over alphabet
$\alphabet$ is a set $\ap{\lang}{\ta} \subseteq \trees{\alphabet}$ over which
there exists an accepting run of $\ta$.  A set of trees $\lang$ is
\emph{regular} iff there is a tree automaton $\ta$ such that $\ap{\lang}{\ta} =
\lang$.

For any tree $\tree$, let $\tasingle{\tree}$ be a tree automaton accepting only
the tree $\tree$.

\subsection{Reset Petri-Nets}

We give a simplified presentation of reset Petri-nets as counter machines with
increment, decrement and reset operations.  This can easily been seen to be
equivalent to the standard definition~\cite{AK77}.

Given a set $\counters = \set{\ctr_1, \ldots, \ctr_\num}$ of counter variables,
we define the set $\cops{\counters}$ of counter operations to be
$\setcomp{\incr{\ctr}, \decr{\ctr}, \reset{\ctr}}{\ctr \in \counters}$.

\begin{definition}[Reset Petri Nets]
    A \emph{reset Petri net} is a tuple $\pn = \tup{\pnstates, \counters,
    \pnrules}$ where $\pnstates$ is a finite set of control states, $\counters$
    is a finite set of counter variables, and $\pnrules \subseteq \pnstates
    \times 2^\cops{\counters} \times \pnstates$ is a transition relation.
\end{definition}

A configuration of a reset Petri net is a pair $\config{\pnstate}{\pnmarking}$
where $\pnstate \in \pnstates$ is a control state and $\pnmarking : \counters
\rightarrow \nats$ is a marking assigning values to counter variables.  We write
$\pnrule{\control}{\pnops}{\control'}$ to denote a rule $\tup{\control, \pnops,
\control'} \in \pnrules$ and omit the set notation when the set of counter
operations is a singleton.  There is a transition $\config{\control}{\pnmarking}
\pntran \config{\control'}{\pnmarking'}$ whenever we have
$\pnrule{\control}{\pnops}{\control'} \in \pnrules$ and there are markings
$\pnmarking_1, \pnmarking_2$ such that 
\begin{itemize}
    \item 
        we have 
        \[
            \ap{\pnmarking_1}{\ctr} = 
            \begin{cases}
                \ap{\pnmarking}{\ctr} - 1 & \text{if } \decr{\ctr} \in \pnops
                \text{ and } \ap{\pnmarking_1}{\ctr} > 0 \\
                \ap{\pnmarking}{\ctr} & \text{if } \decr{\ctr} \notin \pnops
            \end{cases}
        \]
        and

    \item 
        we have 
        \[
            \ap{\pnmarking_2}{\ctr} = 
            \begin{cases}
                0 & \text{if } \reset{\ctr} \in \pnops \\
                \ap{\pnmarking_1}{\ctr} & \text{if } \reset{\ctr} \notin \pnops
            \end{cases}
        \]
        and

    \item
        we have
        \[
            \ap{\pnmarking'}{\ctr} = 
            \begin{cases}
                \ap{\pnmarking_2}{\ctr} + 1 & \text{if } \incr{\ctr} \in \pnops
                \\
                \ap{\pnmarking_2}{\ctr} & \text{if } \incr{\ctr} \notin \pnops
            \end{cases}
        \]
\end{itemize}
Note, in particular, operations are applied in the order $\decr{\ctr},
\reset{\ctr}, \incr{\ctr}$, and if $\ctr$ is $0$, attempting to apply
$\decr{\ctr}$ will cause the Petri net to become stuck.  We write
$\config{\pnstate}{\pnmarking} \pnrun \config{\pnstate'}{\pnmarking'}$ whenever
there is a run $\config{\pnstate}{\pnmarking} \pntran \cdots \pntran
\config{\pnstate'}{\pnmarking'}$ of $\pn$.

Given two markings $\pnmarking$ and $\pnmarking'$, we say $\pnmarking$
\emph{covers} $\pnmarking'$, written $\pnmarking' \pncovers \pnmarking$
whenever, for all $\ctr$ we have $\ap{\pnmarking'}{\ctr} \leq
\ap{\pnmarking}{\ctr}$.

\begin{definition}[Coverability Problem]
    Given a reset Petri-net $\pn$, and configurations
    $\config{\pnstate}{\pnmarking}$ and $\config{\pnstate'}{\pnmarking'}$ the
    \emph{coverability problem} is to decide whether there exists a run
    $\config{\pnstate}{\pnmarking} \pnrun \config{\pnstate'}{\pnmarking''}$ of
    $\pn$ such that $\pnmarking' \pncovers \pnmarking''$.
\end{definition}

The coverability problem for reset Petri nets is decidable via the Karp-Miller
algorithm~\cite{KM67} whose complexity is bounded by $\acker$~\cite{M84}.  In
fact, the coverability problem for reset Petri nets is
$\acker$-complete~\cite{S02,S10}.  In contrast, the reachability problem
(defined below) is undecidable~\cite{AK77}.

\begin{definition}[Reachability Problem]
    Given a reset Petri-net $\pn$, and configurations
    $\config{\pnstate}{\pnmarking}$ and $\config{\pnstate'}{\pnmarking'}$ the
    \emph{reachability problem} is to decide whether there exists a run
    $\config{\pnstate}{\pnmarking} \pnrun \config{\pnstate'}{\pnmarking'}$ of
    $\pn$.
\end{definition}

In the following, we will write $\pnzeromarking$ for the marking assigning zero
to all counters.

\subsection{Ground Tree Rewrite Systems with State}

In this work, we will actually consider a generalisation of GTRS where regular
automata appear in the rewrite rules.  In this way, a single rewrite rule may
correspond to an infinite number of rewrite rules containing concrete trees on
their left- and right-hand sides.  Such an extension is frequently considered
(e.g.~\cite{DHLT90,L03,L12}).  It can be noted that our lower bound results only
use tree automata that accept a singleton set of trees, and thus we do not
increase our lower bounds due to this generalisation.

\subsubsection{Basic Model}

A Ground Tree Rewrite System with State maintains a tree over a given alphabet
$\alphabet$ and a control state from a finite set.  Each transition may update
the control state and rewrite a part of the tree.  Rewriting a tree involves
matching a sub-tree of the current tree and replacing it with a new tree.  Note,
that since we are considering ranked trees, a sub-tree cannot be erased by a
rewrite rule, since this would make the tree inconsistent w.r.t the ranks of the
tree labels.

\begin{definition}[Ground Tree Rewrite System with State]
    A \emph{ground tree rewrite system with state} (sGTRS) is a tuple $\gtrs =
    \tup{\controls, \alphabet, \rules}$ where $\controls$ is a finite set of
    control states, $\alphabet$ is a finite ranked alphabet, and $\rules$ is a
    finite set of rules of the form
    $\rrule{\control_1}{\ta_1}{\control_2}{\ta_2}$ where $\control_1, \control_2
    \in \controls$ and $\ta_1, \ta_2$ are NTAs over $\alphabet$ such that
    $\emptree \notin \ap{\lang}{\ta_1} \cup \ap{\lang}{\ta_2}$.
\end{definition}

A \emph{configuration} of a sGTRS is a pair $\config{\control}{\tree} \in
\controls \times \trees{\alphabet}$.  We have a \emph{transition}
$\config{\control_1}{\tree_1} \tran \config{\control_2}{\tree_2}$ whenever there
is a rule $\rrule{\control_1}{\ta_1}{\control_2}{\ta_2} \in \rules$ such that
$\tree_1 = \csub{\context}{\tree'_1}$ for some context $\context$ and tree
$\tree'_1 \in \ap{\lang}{\ta_1}$ and $\tree_2 = \csub{\context}{\tree'_2}$ for
some tree $\tree'_2 \in \ap{\lang}{\ta_2}$. 

A \emph{run} of an sGTRS is a sequence $\config{\control_1}{\tree_1} \tran
\cdots \tran \config{\control_\runlen}{\tree_\runlen}$ such that for all $1 \leq
\idxi < \runlen$ we have $\config{\control_\idxi}{\tree_\idxi} \tran
\config{\control_{\idxi+1}}{\tree_{\idxi+1}}$ is a transition of $\gtrs$.  We
write $\config{\control}{\tree} \run \config{\control'}{\tree'}$ whenever there
is a run from $\config{\control}{\tree}$ to $\config{\control'}{\tree'}$.

We are interested in both the control state reachability problem and the regular
reachability problem.

\begin{definition}[Control State Reachability Problem]
    Given an sGTRS $\gtrs$, an initial configuration
    $\config{\controlinit}{\treeinit}$ of $\gtrs$ and a target control state
    $\controldest$, the \emph{control state reachability problem} asks whether
    there is a run $\config{\controlinit}{\treeinit} \run
    \config{\controldest}{\tree}$ of $\gtrs$ for some tree $\tree$.
\end{definition}

\begin{definition}[Regular Reachability Problem]
    Given an sGTRS $\gtrs$, an initial configuration
    $\config{\controlinit}{\treeinit}$ of $\gtrs$, a target control state
    $\controldest$, and tree automaton $\ta$, the \emph{regular reachability
    problem} is to decide whether there exists a run
    $\config{\controlinit}{\treeinit} \run \config{\controldest}{\tree}$ for
    some $\tree \in \ap{\lang}{\ta}$.
\end{definition}

\subsubsection{Output Symbols}

We may also consider sGTRSs with output symbols.

\begin{definition}[Ground Tree Rewrite System with State and Outputs]
    A \emph{ground tree rewrite system with state and outputs} is a tuple $\gtrs
    = \tup{\controls, \alphabet, \oalphabet, \rules}$ where $\controls$ is a
    finite set of control states, $\alphabet$ is a finite ranked alphabet,
    $\oalphabet$ is a finite alphabet of output symbols, and $\rules$ is a
    finite set of rules of the form
    $\orrule{\control_1}{\ta_1}{\osym}{\control_2}{\ta_2}$ where $\control_1,
    \control_2 \in \controls$, $\osym \in \oalphabet$, and $\ta_1, \ta_2$ are
    NTAs over $\alphabet$ such that $\emptree \notin \ap{\lang}{\ta_1} \cup
    \ap{\lang}{\ta_2}$.
\end{definition}

As before, a \emph{configuration} of an sGTRS is a pair
$\config{\control}{\tree} \in \controls \times \trees{\alphabet}$.  We have a
\emph{transition} $\config{\control_1}{\tree_1} \otran{\osym}
\config{\control_2}{\tree_2}$ whenever there is a rule
$\orrule{\control_1}{\ta_1}{\osym}{\control_2}{\ta_2} \in \rules$ such that
$\tree_1 = \csub{\context}{\tree'_1}$ for some context $\context$ and tree
$\tree'_1 \in \ap{\lang}{\ta_1}$ and $\tree_2 = \csub{\context}{\tree'_2}$ for
some tree $\tree'_2 \in \ap{\lang}{\ta_2}$.  A \emph{run} over
$\osym_1\ldots\osym_{\runlen-1}$ is a sequence $\config{\control_1}{\tree_1}
\otran{\osym_1} \cdots \otran{\osym_{\runlen-1}}
\config{\control_\runlen}{\tree_\runlen}$ such that for all $1 \leq \idxi <
\runlen$ we have $\config{\control_\idxi}{\tree_\idxi} \otran{\osym_\idxi}
\config{\control_{\idxi+1}}{\tree_{\idxi+1}}$ is a transition of $\gtrs$.  We
write $\config{\control}{\tree} \orun{\osym_1\ldots\osym_\runlen}
\config{\control'}{\tree'}$ whenever there is a run from
$\config{\control}{\tree}$ to $\config{\control'}{\tree'}$ over the sequence of
output symbols $\osym_1\ldots\osym_\runlen$.  Let $\empsym$ denote the empty
output symbol.

\subsubsection{Weakly Extended Ground Tree Rewrite Systems}

The control state and regular reachability problems for sGTRS are known to be
undecidable~\cite{BKRS09,GL11}.  The problems become NP-complete for
\emph{weakly-synchronised} sGTRS~\cite{L12}, where the underlying control state
graph (where there is an edge between $\control_1$ and $\control_2$ whenever
there is a transition $\rrule{\control_1}{\ta_1}{\control_2}{\ta_2}$) may only
have cycles of length $1$ (i.e. self-loops).

More formally, we define the \emph{underlying control graph} of a sGTRS $\gtrs =
\tup{\controls, \alphabet, \oalphabet, \rules}$ as a tuple $\tup{\controls,
\nfatrans}$ where $\nfatrans = \setcomp{\tup{\control,
\control'}}{\orrule{\control}{\ta}{\osym}{\control'}{\ta'} \in \rules}$.  Note,
the underlying control graph of a sGTRS without output symbols can be defined by
simply omitting $\oalphabet$ and $\osym$.

\begin{definition}[Weakly Extended GTRS~\cite{L12}] 
    An sGTRS (with or without output symbols) is \emph{weakly extended} if its
    underlying control graph $\tup{\controls, \nfatrans}$ is such that all paths
    \[
        \tup{\control_1, \control_2}\tup{\control_2,
        \control_3}\ldots\tup{\control_{\runlen-2},
        \control_{\runlen-1}}\tup{\control_{\runlen-1}, \control_\runlen} \in
        \nfatrans^\ast
    \]
    with $\control_1 = \control_\runlen$ satisfy $\control_\idxi = \control_1$
    for all $1 \leq \idxi \leq \runlen$.
\end{definition}

A key result of Lin is that the Parikh image of a weakly extended sGTRS with
output symbols can be represented by an existential Presburger formula that is
constructible in polynomial time.   We use this result to obtain regular
automata representing the possible outputs of weakly extended sGTRSs, which will
be used later in our decidability proofs.  In the following lemma, fix an
arbitrary linear ordering over the output alphabet of $\gtrs$.

\begin{namedlemma}{lem:wgtrsrega}{Parikh Image of Weakly Extended sGTRS}
    Given a weakly extended sGTRS $\gtrs$ with outputs $\oalphabet$, control
    states $\control_1$ and $\control_2$ and tree automata $\ta_1$ and $\ta_2$,
    we can construct a regular automaton $\rega$ with outputs $\oalphabet$ such
    that we have some trees $\tree_1 \in \ap{\lang}{\ta_1}$ and $\tree_2 \in
    \ap{\lang}{\ta_2}$ and a run 
    \[
        \config{\control_1}{\tree_1} \orun{\osym_1\ldots\osym_\runlen}
        \config{\control_2}{\tree_2}
    \] 
    with $\parikh{\osym_1\ldots\osym_\runlen} = \vecnums$ iff $\vecnums \in
    \parikh{\rega}$.  Moreover, the size of $\rega$ is at most triply
    exponential in the size of $\gtrs$.
\end{namedlemma}
\begin{proof}
    Directly from Lemma~2 in Lin 2012~\cite{L12} we can construct in polynomial
    time an existential Presburger formula $\fmla$ with $\sizeof{\oalphabet}$
    free variables such that we have some trees $\tree_1 \in \ap{\lang}{\ta_1}$
    and $\tree_2 \in \ap{\lang}{\ta_2}$ and a run 
    \[
        \config{\control_1}{\tree_1} \orun{\osym_1\ldots\osym_\runlen}
        \config{\control_2}{\tree_2}
    \] 
    with $\parikh{\osym_1\ldots\osym_\runlen} = \vecnums$ iff $\vecnums$
    satisfies $\fmla$.  We know from Ginsburg and Spanier~\cite{GS66} that the
    set of satisfying assignments to an existential Presburger formula can be
    described by a semilinear set (and vice-versa).  That is, there exists some
    $\num$ and for all $1 \leq \idxi \leq \num$ there are vectors of natural
    numbers $\varvecnums_\idxi$ and $\varvecnums^1_\idxi, \ldots,
    \varvecnums^{\varnum_\idxi}_\idxi$ such that for all $\vecnums$, $\vecnums$
    satisfies $\fmla$ iff there exists some $1 \leq \idxi \leq \num$ and
    constants $\mu_1, \ldots, \mu_{\varnum_\idxi} \in \nats$ such that $\vecnums
    = \varvecnums_\idxi + \mu_1 \cdot \varvecnums^1_\idxi + \cdots +
    \mu_{\varnum_\idxi} \cdot \varvecnums^{\varnum_{\idxi}}_\idxi$.  In fact,
    via algorithms of Pottier~\cite{P91}, it is possible to obtain vectors such
    that the size of the values appearing in $\varvecnums_\idxi$ and
    $\varvecnums^1_\idxi, \ldots, \varvecnums^{\varnum_\idxi}_\idxi$ are at most
    doubly exponential in the size of $\fmla$ (via an exponential translation of
    $\fmla$ into DNF, then applying Pottier to gain a bound exponential in the
    size of the DNF --- see Haase~\cite{H13} or Piskac~\cite{P11}).  Since
    $\fmla$ is polynomial in the size of $\gtrs$, we know that that each vector
    $\varvecnums_\idxi$ or $\varvecnums^\idxj_\idxi$ has elements at most doubly
    exponential in size, and since there are at most a triply exponential number
    of such sets of vectors, we know that $\num$ is at most triply exponential
    in the size of $\gtrs$.

    It is straightforward to build, from a semilinear set, a regular automaton
    $\rega$ such that $\parikh{\rega}$ is equivalent to the set: for each
    $\idxi$ we have a branch in $\rega$ first outputting the appropriate number
    of characters to describe $\varvecnums_\idxi$, and then passing through a
    succession of loops each outputting characters describing some
    $\varvecnums^\idxj_\idxi$.  Such an automaton will be at most triply
    exponential in the size of $\gtrs$.
\end{proof}

\section{Senescent Ground Tree Rewrite Systems with State}

In this paper we generalise weakly-synchronised sGTRS to define senescent ground
tree rewrite systems by incorporating ideas from scope-bounded multi-stack
pushdown systems~\cite{lTN11}, where stack characters may only be accessed if
they were created less than a fixed number of context switches previously.

Intuitively, during each transition of a run that changes the control state, the
nodes in the tree ``age'' by one timestep.  When the nodes reach a certain
(fixed) age, they become fossilised and may no longer be changed by any future
transitions.

\subsection{Model Definition}

Given a run $\config{\control_1}{\tree_1} \tran \cdots \tran
\config{\control_\runlen}{\tree_\runlen}$ of an sGTRS, let $\context_1, \ldots,
\context_{\runlen-1}$ be the sequence of tree contexts used in the transitions
from which the run was constructed.  That is, for all $1 \leq \idxi < \runlen$,
we have $\tree_\idxi = \csub{\context_\idxi}{\treeout_\idxi}$ and
$\tree_{\idxi+1} = \csub{\context_\idxi}{\treein_{\idxi+1}}$ where
$\rrule{\control_\idxi}{\ta_\idxi}{\control_{\idxi+1}}{\ta'_\idxi}$ was the
rewrite rule used in the transition and $\treeout_\idxi \in
\ap{\lang}{\ta_\idxi}$, $\treein_{\idxi+1} \in \ap{\lang}{\ta'_\idxi}$ were the
trees that were used in the tree update.

For a given position $\config{\control_\idxi}{\tree_\idxi}$ in the run and a
given node $\tnode$ in the domain of $\tree_\idxi$, the \emph{birthdate} of the
node is the largest $1 \leq \idxj \leq \idxi$ such that $\tnode$ is in the
domain of $\csub{\context_\idxj}{\treein_{\idxj}}$ and $\tnode$ is in the domain
of $\csub{\context_\idxj}{\cvar}$ only if its label is $\cvar$.  The \emph{age}
of a node is the cardinality of the set $\setcomp{\idxi'}{\idxj \leq \idxi' <
\idxi \land \control_{\idxi'} \neq \control_{\idxi'+1}}$.  That is, the age is
the number of times the control state changed between the $\idxj$th and the
$\idxi$th configurations in the run.  This is illustrated in
Figure~\ref{fig:trans}.

\begin{figure}
\input{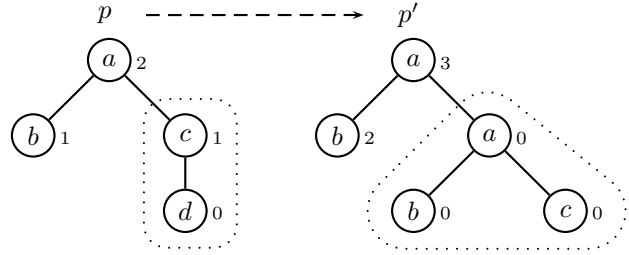}
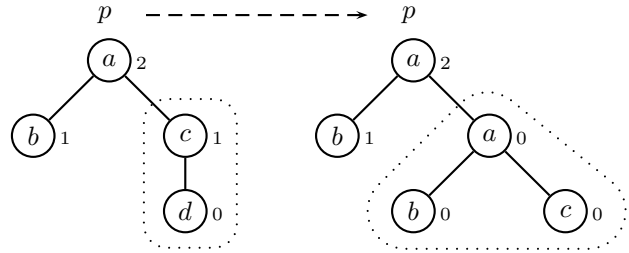
\begin{center}
    \subfloat[]
             [\label{fig:tran-change}A transition changing the control state.]{ 
        \gtrsegchange
    }
    \hspace{2ex}
    \subfloat[]
             [\label{fig:tran-nochange}A transition that does not change the control state.]{ 
        \gtrsegnochange
    }
    \caption{\label{fig:trans}Transitions of a senescent GTRS.}
\end{center}
\end{figure}

Figure~\ref{fig:trans} shows two transitions of a senescent GTRS.  A
configuration is written as its control state ($\control$ or $\control'$) with
the tree appearing below.  The label of each node appears in the centre of the
node, while the ages of each node appears to the right.  The parts of the tree
rewritten by the transition appear inside the dotted lines.
Figure~\ref{fig:tran-change} shows a transition where the control state is
changed.  This change causes the nodes that are not rewritten to increase their
age by $1$.  The rewritten nodes are given the age $0$.
Figure~\ref{fig:tran-nochange} shows a transition that does not change the
control state.  Notice that, in this case, the nodes that are not rewritten
maintain the same age.

A lifespan restricted run with a lifespan of $\lifespan$ is a run such that each
transition $\config{\control_\idxi}{\csub{\context_\idxi}{\treeout_\idxi}} \tran
\config{\control_{\idxi+1}}{\csub{\context_\idxi}{\treein_{\idxi+1}}}$ has the
property that all nodes $\tnode$ in the domain of
$\csub{\context_\idxi}{\treeout_\idxi}$ but only in the domain of
$\csub{\context_\idxi}{\cvar}$ if the label is $\cvar$ have an age of at most
$\lifespan$.  For example, the transitions in Figure~\ref{fig:trans} require a
lifespan $\geq 1$ since the oldest node that is rewritten by the transitions has
age $1$.

\begin{definition}[Senescent Ground Tree Rewrite Systems]
    A \emph{senescent ground tree rewrite system} with \emph{lifespan}
    $\lifespan$ is an sGTRS $\gtrs = \tup{\controls, \alphabet, \rules}$ where
    runs are \emph{lifespan restricted} with a lifespan of $\lifespan$.
\end{definition}

We will study the control state reachability problem and the regular
reachability problem for senescent GTRS.  We will show in
Theorem~\ref{thm:ackermanncomplete} that the control state reachability problem
is $\acker$-complete, and in Theorem~\ref{thm:gtrsreach} that the regular
reachability is undecidable.

\begin{definition}[Control State Reachability Problem]
    Given a senescent GTRS $\gtrs$ with lifespan $\lifespan$, initial
    configuration $\config{\controlinit}{\treeinit}$, and target control state
    $\controldest$, the \emph{control state reachability problem} is to decide
    whether there exists a lifespan restricted run 
    \[
        \config{\controlinit}{\treeinit} \tran \cdots \tran
        \config{\controldest}{\tree}
    \]
    for some $\tree$.
\end{definition}

\begin{definition}[Regular Reachability Problem]
    Given a senescent GTRS $\gtrs$ with lifespan $\lifespan$, initial
    configuration $\config{\controlinit}{\treeinit}$, target control state
    $\controldest$, and tree automaton $\ta$, the \emph{regular reachability
    problem} is to decide whether there exists a lifespan restricted run 
    \[
        \config{\controlinit}{\treeinit} \tran \cdots \tran
        \config{\controldest}{\tree}
    \]
    for some $\tree \in \ap{\lang}{\ta}$.
\end{definition}

One might expect that decidability of control state reachability would imply
decidability of regular reachability: one could simply encode the tree automaton
into the senescent GTRS.  However, it is not possible to enforce conditions on
the final tree (e.g. that all leaf nodes are labelled by initial states of the
tree automaton) when only the final control state can be specified.

\subsection{Example}

\subsubsection{A Simple Concurrent Program}

We present a simple example of how a senescent GTRS may be used to model a
simple concurrent system.  Consider the toy program in Figure~\ref{fig:spawner}.
This simple (outline) program has dynamic thread creation.  

Beginning from the main function, the program creates a thread using a spawn
function that takes as an argument the ID to assign to the new thread, and the
function to run in the new thread.  IDs are obtained using a getID function.
This function is declared ``critical'' to ensure that no two threads execute the
function simultaneously.  It uses a shared variable to ensure that a fresh ID is
returned after each call.

The run function describes the behaviour of each process.  It takes as an
argument its own ID and the ID of another process with which it may communicate.
The switch statement represents a non-deterministic choice between each of the
cases: the process creates two new child processes who may communicate with each
other; or the process communicates with its sibling.

\begin{figure}
    \centering
    \begin{minipage}{80ex}
    \begin{lstlisting}[language=c]
        shared int nextID = 0;

        critical int getID() { 
            return nextID++;
        }

        void run(int myID, int siblingID) {
            switch * {
                case: 
                    int id1 = getID();
                    int id2 = getID();
                    spawn(id1, run(id1, id2));
                    spawn(id2, run(id2, id1)):
                    break;
                case:
                    if (siblingID >= 0)
                        sendMessage(siblingID);
                case:
                    if (siblingID >= 0)
                        receiveMessage(siblingID);
            }
        }

        void main() {
            int id = getID();
            spawn(id, run(id, -1));
        }
    \end{lstlisting}
    \end{minipage}
    \caption{\label{fig:spawner} A simple program with thread creation.}
\end{figure}

In order to ensure reliable communication, a relevant property may be that each
thread has a unique ID.  With a standard C-like semantics, the program in
Figure~\ref{fig:spawner} does not satisfy this property: when the value of
nextID surpasses the maximum value that can be stored in an int, it will loop
back around to $0$.  Thus, two processes may share the same ID.  This is a bug
in the program.

\subsubsection{Modelling the Program as a Senescent GTRS}

We may model the above program as a senescent GTRS.  For simplicity, we will
inline all function calls (we show in Section~\ref{sec:modelling-procs} how to
model threads with stacks, and hence function calls).  We will also abstract
away the sendMessage and receiveMessage functions into silent
actions (since we are only interested in whether two threads share the same ID).

We define the senescent GTRS $\gtrs = \tup{\controls, \alphabet, \rules}$.  Let
$\num$ be the maximum value of an int.  The set of control states will be 
\[
    \begin{array}{rcl}
        \controls &=& \set{0, \ldots, \num}\ \cup \\ 
        & & \setcomp{\spawncontrol{\idxi}{\idxi_1}{\idxj_1}}{0 \leq \idxi,
        \idxi_1
        \leq \num \land -1 \leq \idxj_1 \leq \num}\ \cup \\ 
        & & \set{\hasid{0}, \ldots, \hasid{\num}}\ \cup \\ 
        & & \set{\errorcontrol} \ .
    \end{array}
\]
The control states in $\set{0, \ldots, \num}$ will track the value of the nextID
variable.  States of the form $\spawncontrol{\idxi}{\idxi_1}{\idxj_1}$ indicate
that nextID is $\idxi$ and a new thread with ID $\idxi_1$ and sibling ID
$\idxj_1$ should be started.  Control states of the form $\hasid{\idxi}$ will be
used to detect whether two threads exist with the same ID.  The final control
state $\errorcontrol$ indicates that an ID error has been detected.

The alphabet will be the set
\[
    \begin{array}{rcl}
        \alphabet &=& \setcomp{\threadstate{\idxi_1}{\idxj_1}}{-1 \leq \idxi_1,
        \idxj_1 \leq \num} \ \cup \\ 
        & & \setcomp{\spawncase{\idxi_1}{\idxj_1}{\idxi_2}{\idxj_2},
        \spawncasenext{\idxi_1}{\idxj_1}{\idxi_2}{\idxj_2}}{0 \leq \idxi_1,
        \idxi_2, \idxj_1 \leq \num \land -1 \leq \idxj_2 \leq \num} \ \cup \\ 
        & & \set{\spawnnode, \splitnode, \deadnode} \ .
    \end{array}
\]
The senescent GTRS will maintain a tree with a single leaf node per thread.
Each leaf corresponding to a thread will be of the form
$\threadstate{\idxi_1}{\idxj_1}$ indicating the values of myID and siblingID
respectively.  The labels $\spawncase{\idxi_1}{\idxj_1}{\idxi_2}{\idxj_2}$ and
$\spawncasenext{\idxi_1}{\idxj_1}{\idxi_2}{\idxj_2}$ will be used respectively
for the first and second calls to spawn: $\idxi_2$ and $\idxj_2$ will store the
values returned by getID before the spawn actions are called.

Note, if we wanted to model function calls using a pushdown system, we would use
a branch representing the call stack rather than simple leaves (as in
Section~\ref{sec:modelling-procs}).

The label $\splitnode$ will label internal nodes of the tree, while $\spawnnode$
will label a unique leaf node from which the tree may be expanded to accommodate
new threads.  Finally, $\deadnode$ indicates a sleeping thread.

The rules $\rules$ of $\gtrs$ will be of several kinds: internal thread actions,
spawn actions, and error detection actions.  
\begin{itemize}
    \item
        The internal actions (including the abstracted send and receive actions)
        will simply be rules of the form
        $\rrule{\control}{\tasingle{\threadstate{\idxi_1}{\idxj_1}}}{\control}{\tasingle{\threadstate{\idxi_1}{\idxj_1}}}$
        for all $\control \in \controls$ and $-1 \leq \idxi_1, \idxj_1, \leq
        \num$.  Note that applying these rules does not change the structure of
        given configuration, but it does reset the age of each leaf node to
        zero, preventing the node from fossilising.  Likewise, we can also have
        rules $\rrule{\control}{\tasingle{\cha}}{\control}{\tasingle{\cha}}$ for
        each $\control \in \controls$ and each $\cha \in \alphabet \setminus
        \set{\splitnode}$.

    \item
        The spawn actions will be modelled in several steps.
        \begin{itemize}
            \item
                First we get the value of id1 by
                $\rrule{\idxi}{\tasingle{\threadstate{\idxi_1}{\idxj_1}}}{\idxi'}{\tasingle{\spawncase{\idxi_1}{\idxj_1}{\idxi}{-1}}}$
                where $\idxi' = (\idxi + 1) \bmod{\num}$.  
        
            \item
                Next we get the value of id2 using
                $\rrule{\idxi}{\tasingle{\spawncase{\idxi_1}{\idxj_1}{\idxi_2}{-1}}}{\idxi'}{\tasingle{\spawncase{\idxi_1}{\idxj_1}{\idxi_2}{\idxi}}}$
                (where $\idxi' = (\idxi + 1) \bmod{\num}$).

            \item 
                We then perform the first spawn action using two rules.  First
                we use
                \[
                    \rrule{\idxi}{\tasingle{\spawncase{\idxi_1}{\idxj_1}{\idxi_2}{\idxj_2}}}
                    {\spawncontrol{\idxi}{\idxi_2}{\idxj_2}}
                    {\tasingle{\spawncasenext{\idxi_1}{\idxj_1}{\idxi_2}{\idxj_2}}}
            \]
            to spawn the first thread, and then 
            \[
                \rrule{\idxi} {\tasingle{\spawncasenext{\idxi_1} {\idxj_1} {\idxi_2}
                {\idxj_2}}} {\spawncontrol{\idxi}{\idxj_2}{\idxi_2}}
                {\tasingle{\threadstate{\idxi_1}{\idxj_1}}}
            \]
            to spawn the second.  

            Note that each spawn request changes the control state to
            $\spawncontrol{\idxi}{\idxi_1}{\idxj_1}$.  The actual creation of
            each new thread is via rules of the form
            \[
                \rrule{\spawncontrol{\idxi}{\idxi_1}{\idxj_1}}{\tasingle{\spawnnode}}{\idxi}{\tasingle{\tap{\splitnode}{\threadstate{\idxi_1}{\idxj_1},
                \spawnnode}}} \ .
            \]
         \end{itemize}

     \item
         Finally, we will detect multiple IDs via two sets of rules.  In the first
         a thread can question whether any other thread has the same ID.  That
         is, we have
         $\rrule{\idxi}{\tasingle{\threadstate{\idxi_1}{\idxj_1}}}{\hasid{\idxi_1}}{\tasingle{\deadnode}}$.
         That is, the node requests if any thread has the same ID, and then goes
         to sleep (so that it does not answer its own request).  The request can
         be answered with a rule
         $\rrule{\hasid{\idxi}}{\tasingle{\threadstate{\idxi}{\idxj_1}}}{\errorcontrol}{\tasingle{\threadstate{\idxi}{\idxj_1}}}$,
         taking the GTRS to control state $\errorcontrol$, indicating that an
         error has occurred.
\end{itemize}

The above defined system will always have a run to $\errorcontrol$ from the
initial configuration $\config{1}{\tasingle{\tap{\splitnode}{\threadstate{0,
-1}, \spawnnode}}}$ which represents the system state of the initial spawn
action.  This holds even with a lifespan of $1$: the first process simply
performs spawn actions until the nextID reaches $\num$ and wraps around to $0$;
then by starting another two new processes and moving to $\hasid{0}$ the error
can be detected.

\section{Modelling Power of Senescent GTRSs}

We show in this section that senescent GTRSs at least capture scope-bounded
multi-stack pushdown systems.  Essentially just by encoding each stack as a
different branch of the tree.  We also show that we can encode coverability and
reachability of a reset Petri net.  The former can be reduced to control state
reachability, whereas the latter is reducible to regular reachability.  Thus,
regular reachability is undecidable.

\subsection{Scope-Bounded Pushdown Systems} 
\label{sec:threads-are-branches} 
\label{sec:modelling-procs}

A senescent GTRS may quite naturally model a scope-bounded multi-stack pushdown
system.  Scope-bounded multi-pushdown systems were first introduced by La Torre
and Napoli~\cite{lTN11}, where they were shown to have a PSPACE-complete control
state reachability problem.  This is in contrast to the control state
reachability problem for senescent GTRS which we will show to be
$\acker$-complete.  Hence, senescent GTRS represent a significant increase in
modelling power.  We first define multi-stack and scope-bounded pushdown
systems, before comparing them with senescent GTRS.

\subsubsection{Model Definition}

A multi-stack pushdown system consists of, at any one moment, a control state
and a fixed number $\numstacks$ of stacks over an alphabet $\pdsalphabet$.  Runs
of a scope-bounded multi-stack pushdown system are organised into \emph{rounds},
where each round consists of $\numstacks$ \emph{phases} and during the $\idxi$th
phase, stack operations may only occur on the $\idxi$th stack.  This can be
thought of as several threads running on a round robin scheduler.  The
\emph{scope-bound} $\scopebound$ is a restriction on which characters may be
removed from a stack.  That is, a character may only be removed if it was pushed
within the previous $\scopebound$ rounds.

\begin{definition}[Multi-Stack Pushdown Systems]
    A \emph{multi-stack pushdown system} is a tuple $\pds = \tup{\pdscontrols,
    \pdsalphabet, \pdsrules{1}, \ldots, \pdsrules{\numstacks}}$ where
    $\pdscontrols$ is a finite set of control states, $\pdsalphabet$ is a finite
    set of stack characters, and for each $1 \leq \idxi \leq \numstacks$ we have
    $\pdsrules{\idxi} = \pdspushrules{\idxi} \cup \pdsintrules{\idxi} \cup
    \pdspoprules{\idxi}$ is a set of rules where $\pdspushrules{\idxi} \subseteq
    \pdscontrols \times \pdscontrols \times \pdsalphabet$ is a set of
    push rules, $\pdsintrules{\idxi} \subseteq \pdscontrols \times
    \pdscontrols$ is a set of internal rules, and $\pdspoprules{\idxi}
    \subseteq \pdscontrols \times \pdsalphabet \times \pdscontrols$ is a set of
    pop rules.
\end{definition}

A \emph{configuration} of a multi-stack pushdown system is a tuple
$\config{\pdscontrol}{\pdsstack_1, \ldots, \pdsstack_\numstacks}$ where
$\pdscontrol \in \pdscontrols$ and for all $1 \leq \idxi \leq \numstacks$ we
have $\pdsstack_1 \in \pdsalphabet^\ast$.  We have a transition on stack
$\idxi$, written 
\[
    \config{\pdscontrol}{\pdsstack_1, \ldots, \pdsstack_\numstacks}
    \pdsitran{\idxi} \config{\pdscontrol'}{\pdsstack_1, \ldots,
    \pdsstack_{\idxi-1}, \pdsstack'_\idxi, \pdsstack_{\idxi+1}, \ldots,
    \pdsstack_\numstacks}
\]
whenever we have
\begin{enumerate}
    \item
        a rule $\pdspushrule{\pdscontrol}{\pdscontrol'}{\cha} \in
        \pdspushrules{\idxi}$ and $\pdsstack'_\idxi = \cha \pdsstack_\idxi$, or

    \item
        a rule $\pdsintrule{\pdscontrol}{\pdscontrol'} \in \pdsintrules{\idxi}$
        and $\pdsstack'_\idxi = \pdsstack_\idxi$, or

    \item
        a rule $\pdspoprule{\pdscontrol}{\cha}{\pdscontrol'} \in
        \pdspoprules{\idxi}$ and $\cha \pdsstack'_\idxi = \pdsstack_\idxi$.
\end{enumerate}
We write $\pdsirun{\idxi}$ for the transitive closure of $\pdsitran{\idxi}$ and
$\pdstran$ for the union of all $\pdsitran{\idxi}$.  Take a \emph{run}
\[
    \pdsconfsym_0 \pdstran \cdots \pdstran \pdsconfsym_\runlen
\]
of a multi-stack pushdown system and suppose that the $\idxi$th configuration
$\pdsconfsym_\idxi$ is the configuration
$\config{\pdscontrol_\idxi}{\pdsstack_1, \ldots, \pdsstack_{\idxj-1}, \cha
\pdsstack_\idxj, \pdsstack_{\idxj+1}, \ldots, \pdsstack_\numstacks}$.  We say
that $\cha$ was pushed at configuration $\pdsconfsym_{\idxi'}$ whenever $\idxi'
\leq \idxi$ is the largest index such that we have 
\[
    \begin{array}{rcl}
        \pdsconfsym_{\idxi'-1} = & \config{\pdscontrol_{\idxi'-1}}{\pdsstack'_1,
        \ldots, \pdsstack'_{\idxj-1}, \pdsstack_{\idxj}, \pdsstack'_{\idxj+1},
        \ldots, \pdsstack'_\numstacks} \\
        & \pdsitran{\idxj} & \\ 
        & \config{\pdscontrol_{\idxi'}}{\pdsstack'_1, \ldots,
        \pdsstack'_{\idxj-1}, \cha \pdsstack_{\idxj}, \pdsstack'_{\idxj+1},
        \ldots, \pdsstack'_\numstacks} & = \pdsconfsym_{\idxi'} \ .
    \end{array}
\]
Note that, by convention, $\cha$ was pushed at configuration $\pdsconfsym_0$ if
no such $\idxi'$ exists.

A \emph{round} $\pdsconfsym_0 \pdsround \pdsconfsym_\numstacks$ of a multi-stack
pushdown system is a sequence 
\[
    \pdsconfsym_0 \pdsirun{1} \pdsconfsym_1 \pdsirun{2} \cdots
    \pdsirun{\numstacks} \pdsconfsym_\numstacks
\]
where each $\pdsconfsym_\idxi$ is a configuration.  

Finally, a \emph{$\scopebound$-scope-bounded} run of a multi-stack pushdown
system is a sequence of rounds
\[
    \pdsconfsym_1 \pdsround \cdots \pdsround \pdsconfsym_\runlen
\]
for some $\runlen$ such that for all $\idxj$ and all pop transitions 
\[
    \config{\pdscontrol}{\pdsstack_1, \ldots, \pdsstack_{\idxi-1}, \cha
    \pdsstack_\idxi, \pdsstack_{\idxi+1}, \ldots, \pdsstack_\numstacks}
    \pdsitran{\idxi} \config{\pdscontrol'}{\pdsstack_1, \ldots,
    \pdsstack_{\idxi-1}, \pdsstack_\idxi, \pdsstack_{\idxi+1}, \ldots,
    \pdsstack_\numstacks}
\]
occurring during the round $\pdsconfsym_\idxj \pdsround \pdsconfsym_{\idxj+1}$
we have that $\cha$ was pushed during the round $\pdsconfsym_{\idxj'} \pdsround
\pdsconfsym_{\idxj'+1}$ where $(\idxj - \scopebound) \leq \idxj' \leq \idxj$.

\begin{definition}[Scope-bounded Multi-Stack Pushdown System]
    We define a \emph{$\scopebound$-scope-bounded multi-stack pushdown system}
    to be a multi-stack pushdown system together with a scope-bound
    $\scopebound$.
\end{definition}

We thus define the control state reachability problem for scope-bounded
multi-stack pushdown systems.

\begin{definition}[Control State Reachability]
    Given a scope-bounded multi-stack pushdown system $\pds$ with scope-bound
    $\scopebound$ and $\numstacks$ stacks, control states $\pdscontrolinit$ and
    $\pdscontroldest$, and initial stack character $\pdschainit$, the
    \emph{control state reachability problem} is to decide whether there exists
    a $\scopebound$-scope-bounded run 
    \[
        \config{\pdscontrolinit}{\pdschainit, \ldots, \pdschainit} \pdstran
        \cdots \pdstran \config{\pdscontroldest}{\pdsstack_1, \ldots,
        \pdsstack_\numstacks}
    \]
    for some stacks $\pdsstack_1, \ldots, \pdsstack_\numstacks$.
\end{definition}

\subsubsection{Reduction to Senescent GTRS}

We show that the control state reachability problem for scope-bounded
multi-stack pushdown systems can be simply reduced to the control state
reachability problem for senescent GTRS.  The reduction is a straightforward
extension of the standard method for encoding a pushdown system with an sGTRS
with a single control state, which was generalised to \emph{context-bounded}
multi-stack pushdown systems by Lin~\cite{L12}.  

Without loss of generality, we will assume a stack symbol $\sbot$ that is the
bottom-of-stack symbol.  It is neither pushed onto, nor popped from the stack.
It will also be the initial stack character in the control state reachability
problem. 

Furthermore, by abuse of notation, for a stack $\pdsstack = \cha_1 \ldots
\cha_\num$ we write $\tap{\pdsstack}{\tree}$ for the tree
$\tap{\cha_\num}{\cdots\tap{\cha_1}{\tree}}$.  

A single-stack pushdown system can be modelled as follows.  A configuration
$\config{\pdscontrol}{\pdsstack}$ is encoded as a tree containing a single path.
Consider the tree $\tap{\pdsstack}{\pdscontrol}$.  Since the rules of the
pushdown system only depend on and change the control state and the top of the
stack, they can be encoded as tree rewriting operations.  For example,
the push rule $\pdspushrule{\pdscontrol}{\pdscontrol'}{\cha}$ can be modelled by
matching the subtree $\pdscontrol$ and replacing it with
$\tap{\cha}{\pdscontrol'}$.

To extend this to multi-stack pushdown systems with $\numstacks$ stacks, we
maintain a tree whose root is a node with $\numstacks$ children, where
each child encodes a stack.  However, in this case, the control state must be
stored in the control state of the sGTRS since the tree rewriting rules can only
rewrite sub-trees.

Fortunately, the structure of a scope-bounded run allows us to choose an
encoding that is more economical with the use of the senescent GTRS's control
state.  We will use the pushdown system's control state as a kind of ``token''
to indicate which stack is currently active in the round.  That is, 
it will appear as a leaf of the branch containing the currently active stack.
To move to the next stack (i.e. use $\pdsitran{\idxi+1}$ instead of
$\pdsitran{\idxi}$) the control state of the senescent GTRS will be used to
transfer the pushdown system's control state to the next branch.  Thus, a
$\scopebound$-scope-bounded multi-stack pushdown system will be modelled by a
senescent GTRS with a lifespan of $\scopebound \cdot \numstacks$.  This is
natural since a round of a scope bounded pushdown system contains
$\numstacks$ communications, and hence $\scopebound$ rounds contain $\scopebound
\cdot \numstacks$ communications.

\begin{nameddefinition}{def:gtrsscoped}{$\gtrsscoped{\pds}$}
    Given a $\scopebound$-scope-bounded multi-stack pushdown system $\pds$ as
    the tuple 
    $\tup{\pdscontrols, \pdsalphabet, \pdsrules{1}, \ldots,
    \pdsrules{\numstacks}}$ we define the senescent GTRS $\gtrsscoped{\pds} =
    \tup{\gtrsscopedcontrols, \gtrsscopedalphabet, \gtrsscopedrules}$ with
    lifespan $\scopebound \cdot \numstacks$ where 
    \[
        \begin{array}{rcl}
            \gtrsscopedcontrols & = & \brac{\controls \times \set{1, \ldots,
            \numstacks}} \cup \set{\pdscontrolinit, \pdscontroldest}
            \\
            \\
            \gtrsscopedalphabet & = & \pdsalphabet \cup \pdscontrols \cup
            \set{\splitnode, \stoppednode{1}, \ldots, \stoppednode{\numstacks}}
            \\
            \\
            \gtrsscopedrules & = & \gtrsscopedrulesbeginend \cup \gtrsscopedrulessim
            \cup \gtrsscopedrulesswitch
        \end{array}
    \]
    and $\gtrsscopedrulesbeginend$ is the set
    \[
        \set{\rrule{\pdscontrolinit}{\tasingle{\stoppednode{\numstacks}}}{\tup{\pdscontrolinit,
        1}}{\tasingle{\stoppednode{\numstacks}}}, \rrule{\tup{\pdscontroldest,
        1}}{\tasingle{\pdscontroldest}}{\pdscontroldest}{\tasingle{\pdscontroldest}}}
    \]
    and $\gtrsscopedrulessim$ is the set
    \[
        \begin{array}{l}
            \setcomp{\rrule{\tup{\pdscontrol, \idxi}}
                           {\tasingle{\pdscontrol_1}}
                           {\tup{\pdscontrol, \idxi}}
                           {\tasingle{\tap{\cha}{\pdscontrol_2}}}}
                    {
                        \begin{array}{c}
                            1 \leq \idxi \leq \numstacks \ \land \\
                            \pdscontrol \in \pdscontrols \ \land \\
                            \pdspushrule{\pdscontrol_1}
                                        {\pdscontrol_2}
                                        {\cha} \in \pdsrules{\idxi}
                        \end{array}
                    } \ \cup \\
            \setcomp{\rrule{\tup{\pdscontrol, \idxi}}
                           {\tasingle{\pdscontrol_1}}
                           {\tup{\pdscontrol, \idxi}}
                           {\tasingle{\pdscontrol_2}}}
                    {
                        \begin{array}{c}
                            1 \leq \idxi \leq \numstacks \ \land \\
                            \pdscontrol \in \pdscontrols \ \land \\
                            \pdsintrule{\pdscontrol_1}
                                       {\pdscontrol_2} \in \pdsintrules{\idxi}
                        \end{array}
                    } \ \cup \\
            \setcomp{\rrule{\tup{\pdscontrol, \idxi}}
                           {\tasingle{\tap{\cha}{\pdscontrol_1}}}
                           {\tup{\pdscontrol, \idxi}}
                           {\tasingle{\pdscontrol_2}}}
                    {
                        \begin{array}{c}
                            1 \leq \idxi \leq \numstacks \ \land \\
                            \pdscontrol \in \pdscontrols \ \land \\
                            \pdspoprule{\pdscontrol_1}
                                       {\cha}
                                       {\pdscontrol_2} \in \pdspoprules{\idxi}
                        \end{array}
                    }
        \end{array}
    \]
    and finally $\gtrsscopedrulesswitch$ is the set
    \[
        \begin{array}{l}
            \setcomp{\rrule{\tup{\pdscontrol, \idxi}}
                           {\tasingle{\pdscontrol'}}
                           {\tup{\pdscontrol', \brac{\idxi \bmod{\numstacks}} + 1}}
                           {\tasingle{\stoppednode{\idxi}}}}
                    {1 \leq \idxi \leq \numstacks} \ \cup \\
            \setcomp{\rrule{\tup{\pdscontrol, \idxi}}
                           {\tasingle{\stoppednode{\idxi}}}
                           {\tup{\pdscontrol, \idxi}}
                           {\tasingle{\pdscontrol}}}
                    {1 \leq \idxi \leq \numstacks} \ . 
        \end{array}
    \]
\end{nameddefinition}

\begin{theorem}[Scope-Bounded to Senescent GTRS]
    The control state reachability problem for scope-bounded multi-stack
    pushdown systems can be reduced to the control state reachability problem
    for senescent GTRS.
\end{theorem}
\begin{proof}
    Given a $\scopebound$-scope-bounded multi-stack pushdown system $\pds =
    \tup{\pdscontrols, \pdsalphabet, \pdsrules{1}, \ldots,
    \pdsrules{\numstacks}}$ we obtain the senescent GTRS $\gtrsscoped{\pds}$
    with lifespan $\scopebound \cdot \numstacks$ as in
    \refdefinition{def:gtrsscoped}.  It is almost direct to obtain from a run
    \[
        \config{\pdscontrolinit}{\pdschainit, \ldots, \pdschainit} \pdstran
        \cdots \pdstran \config{\pdscontroldest}{\pdsstack_1, \ldots,
        \pdsstack_\numstacks}
    \]
    of $\pds$ a run
    \[
        \config{\pdscontrolinit}{\tap{\splitnode}{\tap{\pdschainit}{\stoppednode{1}},
        \ldots, \tap{\pdschainit}{\stoppednode{\numstacks}}}} \pdstran \cdots
        \pdstran \config{\pdscontroldest}{\tap{\splitnode}{\tree_1, \ldots,
        \tree_\numstacks}}
    \]
    of $\gtrsscoped{\pds}$ and vice versa for some $\pdsstack_1, \ldots,
    \pdsstack_\numstacks$ and $\tree_1, \ldots, \tree_\numstacks$.

    To go from $\pds$ to $\gtrsscoped{\pds}$ we divide the run into rounds and
    the rounds into sub-runs 
    \[
        \config{\pdscontrol}{\pdsstack_1, \ldots, \pdsstack_{\idxi-1},
        \pdsstack_{\idxi}, \pdsstack_{\idxi+1}, \ldots, \pdsstack_\numstacks}
        \pdsirun{\idxi} \config{\pdscontrol'}{\pdsstack_1, \ldots,
        \pdsstack_{\idxi-1}, \pdsstack'_{\idxi}, \pdsstack_{\idxi+1}, \ldots,
        \pdsstack_\numstacks} \ .
    \]
    We can obtain by straightforward induction (using rules from
    $\gtrsscopedrulessim$ and the fact we never remove $\sbot$ from a stack) a
    run 
    \[
        \begin{array}{c}
            \config{\tup{\pdscontrol, \idxi}}{\tap{\splitnode}{
                \tap{\pdsstack_1}{\stoppednode{1}},
                \ldots, 
                \tap{\pdsstack_{\idxi-1}}{\stoppednode{\idxi-1}},
                \tap{\pdsstack_{\idxi}}{\pdscontrol},
                \tap{\pdsstack_{\idxi+1}}{\stoppednode{\idxi+1}}, 
                \ldots,
                \tap{\pdsstack_\numstacks}{\stoppednode{\numstacks}}
            }} \\
            \run \\
            \config{\tup{\pdscontrol, \idxi}}{\tap{\splitnode}{
                \tap{\pdsstack_1}{\stoppednode{1}},
                \ldots, 
                \tap{\pdsstack_{\idxi-1}}{\stoppednode{\idxi-1}},
                \tap{\pdsstack'_{\idxi}}{\pdscontrol'},
                \tap{\pdsstack_{\idxi+1}}{\stoppednode{\idxi+1}}, 
                \ldots,
                \tap{\pdsstack_\numstacks}{\stoppednode{\numstacks}}
            }}
        \end{array}
    \]
    of $\gtrsscoped{\pds}$.  Note that this holds true even for empty runs.
    Then by topping and tailing with transitions from $\gtrsscopedrulesswitch$
    we obtain.
    \[
        \begin{array}{c}
            \config{\tup{\pdscontrol, \idxi}}{\tap{\splitnode}{
                \tap{\pdsstack_1}{\stoppednode{1}},
                \ldots, 
                \tap{\pdsstack_{\idxi-1}}{\stoppednode{\idxi-1}},
                \tap{\pdsstack_{\idxi}}{\stoppednode{\idxi}},
                \tap{\pdsstack_{\idxi+1}}{\stoppednode{\idxi+1}}, 
                \ldots,
                \tap{\pdsstack_\numstacks}{\stoppednode{\numstacks}}
            }} \\
            \run \\
            \config{\tup{\pdscontrol', \idxi'}}{\tap{\splitnode}{
                \tap{\pdsstack_1}{\stoppednode{1}},
                \ldots, 
                \tap{\pdsstack_{\idxi-1}}{\stoppednode{\idxi-1}},
                \tap{\pdsstack'_{\idxi}}{\stoppednode{\idxi}},
                \tap{\pdsstack_{\idxi+1}}{\stoppednode{\idxi+1}}, 
                \ldots,
                \tap{\pdsstack_\numstacks}{\stoppednode{\numstacks}}
            }}
        \end{array}
    \]
    where $\idxi' = \brac{\idxi \bmod{\numstacks}} + 1$.

    Thus, combining these runs, from a round 
    \[
        \config{\pdscontrol}{\pdsstack_1, \ldots, \pdsstack_\numstacks}
        \pdsirun{1} \cdots \pdsirun{\numstacks} 
        \config{\pdscontrol'}{\pdsstack'_1, \ldots, \pdsstack'_\numstacks}
    \]
    of $\pds$ we obtain a run
    \[
        \config{\tup{\pdscontrol, 1}}{\tap{\splitnode}{
            \tap{\pdsstack_1}{\stoppednode{1}},
            \ldots, 
            \tap{\pdsstack_\numstacks}{\stoppednode{\numstacks}}
        }} 
        \run 
        \config{\tup{\pdscontrol', 1}}{\tap{\splitnode}{
            \tap{\pdsstack'_1}{\stoppednode{1}},
            \ldots, 
            \tap{\pdsstack'_\numstacks}{\stoppednode{\numstacks}}
        }}
    \]
    of $\gtrsscoped{\pds}$.

    To complete the direction, we must now combine rounds of the run of $\pds$
    into a run of $\gtrsscoped{\pds}$.  To do this, we simply concatenate the
    runs obtained for each round.  We have to be careful that in doing so we
    respect the lifespan $\scopebound \cdot \numstacks$ of $\gtrsscoped{\pds}$.
    Indeed, this is a simple consequence of the scope-bound of $\pds$: since we
    never remove a character that was pushed $\scopebound$ rounds earlier, we
    know that the top character must have been pushed at most $\scopebound$
    rounds earlier, and thus there are fewer than $\scopebound \cdot \numstacks$
    control state changes of $\gtrsscoped{\pds}$ since the birth date of the
    corresponding node in the tree (note also that the leaf node is rewritten
    every $\numstacks$ control state changes, and hence does not become
    fossilised).  We thus obtain an almost complete run of $\gtrsscoped{\pds}$:
    all that remains is to append and concatenate transitions from
    $\gtrsscopedrulesbeginend$, resulting in a run of the required form.

    In the opposite direction one need only observe that all runs of
    $\gtrsscoped{\pds}$ must be of the form constructed during the proof of the
    direction above.  Hence, applying the above reasoning in reverse obtains a
    run of $\pds$ as required.
\end{proof}

\subsection{Reset Petri-Nets}

We show that the coverability and reachability problems for reset Petri-nets can
be reduced to the control state and regular reachability problems for senescent
GTRS respectively.  The idea is that the control state of the reset Petri-net
can be directly encoded by the control state of the senescent GTRS.  To keep
track of the marking for each counter $\ctr$, we maintain a tree with
$\ap{\pnmarking}{\ctr}$ leaf nodes labelled $\ctr$.  Decrementing a counter is
then a case of rewriting a leaf node $\ctr$ to the empty tree, while
incrementing the counter requires adding a new leaf node.  To avoid leaf nodes
becoming fossilised, we allow all leaf nodes to rewrite to themselves in
(almost) every control state.  We can then reset a counter by forcing the GTRS
to change control states $\lifespan$ times without allowing the counter to
refresh; thus, all $\ctr$ nodes become fossilised and the counter is effectively
set to zero.

\subsubsection{Coverability}

We will begin with coverability.  Without loss of generality, we assume that we
aim to cover the zero marking.  Moreover, we assume that in each rule
$\pnrule{\pnstate}{\pnops}{\pnstate'}$ we have at most a single counter
operation in $\pnops$.  In the following definition, we use $\spawnnode$ to
label an open node which may spawn new counter-labelled nodes, $\splitnode$ to
label internal nodes of the tree, and $\deadnode$ to label counter nodes that
have been disactivated by a decrement operation.

\begin{nameddefinition}{def:gtrscover}{$\gtrscover{\pn}$}
    Given a reset Petri-net $\pn = \tup{\pnstates, \counters, \pnrules}$ we
    define the senescent GTRS $\gtrscover{\pn} = \tup{\controls, \alphabet,
    \rules}$ with lifespan $1$ where
    \begin{align*}
        \controls = & \pnstates \cup \setcomp{\killst{\ctr}{\pnstate}}
                                             {\ctr \in \counters \land 
                                              \pnstate \in \pnstates} \\
        \alphabet = & \counters \cup \set{\spawnnode, \splitnode, \deadnode} \\
        \rules = & \setcomp{\rrule{\control}
                                  {\tasingle{\ctr}}
                                  {\control}
                                  {\tasingle{\ctr}}}
                           {\ctr \in \counters \land
                            \control \in \controls \land 
                            \forall \pnstate \in \pnstates . 
                                \control \neq \killst{\ctr}{\pnstate}} \ \cup \\
        & \setcomp{\rrule{\control}
                         {\tasingle{\spawnnode}}
                         {\control}
                         {\tasingle{\spawnnode}}}
                  {\control \in \controls} \ \cup \\
        & \setcomp{\rrule{\pnstate}
                         {\tasingle{\spawnnode}}
                         {\pnstate'}
                         {\tasingle{\tap{\splitnode}
                                        {\ctr, \spawnnode}}}}
                  {\pnrule{\pnstate}
                          {\set{\incr{\ctr}}}
                          {\pnstate'} \in \pnrules} \ \cup \\
        & \setcomp{\rrule{\pnstate}
                         {\tasingle{\ctr}}
                         {\pnstate'}
                         {\tasingle{\deadnode}}}
                  {\pnrule{\pnstate}
                          {\set{\decr{\ctr}}}
                          {\pnstate'} \in \pnrules} \ \cup \\
        & \setcomp{
            \begin{array}{l}
                \rrule{\pnstate}
                      {\tasingle{\spawnnode}}
                      {\killst{\ctr}{\pnstate'}}
                      {\tasingle{\spawnnode}}, \\
                \rrule{\killst{\ctr}{\pnstate'}}
                      {\tasingle{\spawnnode}}
                      {\pnstate'}
                      {\tasingle{\spawnnode}}
            \end{array}
          }
          {\pnrule{\pnstate}
                  {\set{\reset{\ctr}}}
                  {\pnstate'} \in \pnrules} \ \cup \\
        & \setcomp{\rrule{\pnstate}
                         {\tasingle{\spawnnode}}
                         {\pnstate'}
                         {\tasingle{\spawnnode}}}
                  {\pnrule{\pnstate}
                          {\emptyset}
                          {\pnstate'} \in \pnrules} \\
    \end{align*}
\end{nameddefinition}

\begin{namedtheorem}{thm:gtrscover}{Coverability to Control State Reachability}
    The coverability problem for reset Petri-nets can be reduced to the control
    state reachability problem for senescent GTRS.
\end{namedtheorem}
\begin{proof}
    Given a reset Petri-net $\pn = \tup{\pnstates, \counters, \pnrules}$ we
    obtain the senescent GTRS $\gtrscover{\pn} = \tup{\controls, \alphabet,
    \rules}$ with lifespan $1$ from \refdefinition{def:gtrscover}.  We show
    there exists a run
    \[
        \config{\pnstate_1}{\pnmarking_1} \pntran \cdots \pntran
        \config{\pnstate_\runlen}{\pnmarking_\runlen}
    \]
    of $\pn$ where $\pnmarking_1 = \pnzeromarking$ iff there is a run
    \[
        \config{\pnstate_1}{\tree_1} \run \config{\pnstate_2}{\tree_2} \run \cdots
        \run \config{\pnstate_\runlen}{\tree_\runlen}
    \]
    of $\gtrscover{\pn}$ where $\tree_1 = \spawnnode$.  This implies our
    theorem.

    First consider the direction from $\pn$ to $\gtrscover{\pn}$.  We induct
    from $\idxi = 1$ to $\idxi = \runlen$.  We maintain the induction invariant
    that $\tree_\idxi$ has exactly $\ap{\pnmarking_\idxi}{\ctr}$ leaf nodes of
    age $0$ labelled by $\ctr$ for each counter $\ctr \in \counters$.
    Furthermore, there is exactly one leaf node labelled $\spawnnode$, and this
    node has age $1$.  In the base case the invariant is immediate.

    Now assume the invariant for $\idxi$.  Consider the transition 
    \[
        \config{\pnstate_\idxi}{\pnmarking_\idxi} \pnlabtran{\pnops}
        \config{\pnstate_{\idxi+1}}{\pnmarking_{\idxi+1}}
    \]
    of the run of $\pn$.  We show the existence of a run
    \[
        \config{\pnstate_\idxi}{\tree_\idxi} \run
        \config{\pnstate_{\idxi+1}}{\tree_{\idxi+1}}
    \]
    satisfying the invariant.  We perform a case split on $\pnops$.  In the
    following, when we say ``refresh the leaf nodes'' we mean that we execute
    for each counter $\ctr$ and for each leaf node (with age $< 2$) labelled by
    $\ctr$ a rule
    $\rrule{\pnstate_{\idxi+1}}{\tasingle{\ctr}}{\pnstate_{\idxi+1}}{\tasingle{\ctr}}$
    to set the age of each leaf node back to $0$, and finally we fire
    $\rrule{\pnstate_{\idxi+1}}{\tasingle{\spawnnode}}{\pnstate_{\idxi+1}}{\tasingle{\spawnnode}}$
    to set the age of the node labelled $\spawnnode$ to $0$.
    \begin{enumerate}
        \item
            When $\pnops = \emptyset$ we can fire the rule
            $\rrule{\pnstate_\idxi}{\tasingle{\spawnnode}}{\pnstate_{\idxi+1}}{\tasingle{\spawnnode}}$
            and since $\pnmarking_\idxi = \pnmarking_{\idxi+1}$ we simply
            refresh the leaf nodes to obtain the induction invariant.

        \item
            When $\pnops = \set{\incr{\ctr}}$ we fire the rule
            $\rrule{\pnstate_\idxi}{\tasingle{\spawnnode}}{\pnstate_{\idxi+1}}{\tasingle{\tap{\splitnode}{\ctr,
            \spawnnode}}}$ to obtain $\tree_{\idxi+1}$ with the correct number
            of leaf nodes concordant with $\pnmarking_{\idxi+1}$ and then
            refresh the leaf nodes to obtain the invariant.

        \item
            When $\pnops = \set{\decr{\ctr}}$ we know from the invariant that
            since $\ap{\pnmarking_\idxi}{\ctr} > 0$ that there is a leaf
            labelled $\ctr$ of age $0$.  We fire the rule
            $\rrule{\pnstate_\idxi}{\tasingle{\ctr}}{\pnstate_{\idxi+1}}{\tasingle{\deadnode}}$
            to obtain $\tree_{\idxi+1}$ with the correct number of leaf nodes
            concordant with $\pnmarking_{\idxi+1}$ and then refresh the leaf
            nodes to obtain the invariant.

        \item
            When $\pnops = \set{\reset{\ctr}}$ we fire the rule
            $\rrule{\pnstate_\idxi}{\tasingle{\spawnnode}}{\killst{\ctr}{\pnstate_{\idxi+1}}}{\tasingle{\spawnnode}}$
            and then refresh all leaf nodes except those labelled $\ctr$, which
            cannot be reset.  Thus all leaf nodes have age $0$ except those
            labelled $\ctr$ which have age $1$.  Then we fire
            $\rrule{\killst{\ctr}{\pnstate_{\idxi+1}}}{\tasingle{\spawnnode}}{\pnstate_{\idxi+1}}{\tasingle{\spawnnode}}$.
            Note that all leaf nodes labelled $\ctr$ now have age $2$ and are
            fossilised.  Then we refresh all leaf nodes (that are young enough)
            to obtain $\tree_{\idxi+1}$ concordant with $\pnmarking_{\idxi+1}$
            and obtain the invariant.
    \end{enumerate}
    Thus, by induction, we obtain a run as required.

    In the other direction, we take a run 
    \[
        \config{\pnstate_1}{\tree_1} \run \cdots \run
        \config{\pnstate_\runlen}{\tree_\runlen}
    \]
    where $\tree_1 = \spawnnode$ and for all configurations between
    $\config{\pnstate_\idxi}{\tree_\idxi}$ and
    $\config{\pnstate_{\idxi+1}}{\tree_{\idxi+1}}$ (if there are any) there is
    some $\ctr$ such that such that the control state of the configurations is
    of the form $\killst{\ctr}{\pnstate_{\idxi+1}}$.  That all runs are of this
    form follows easily from the definition of $\gtrscover{\pn}$.  We build a
    run
    \[
        \config{\pnstate_1}{\pnmarking_1} \pntran \cdots \pntran
        \config{\pnstate_\runlen}{\pnmarking_\runlen}
    \]
    of $\pn$ where $\pnmarking_1 = \pnzeromarking$.  For technical convenience
    we assume that $\pn$ can perform ``no-op'' transitions that change neither
    the control state nor the marking.  One can simply remove these transitions
    from the final run to obtain a run of $\pn$.

    We induct from $\idxi = 1$ to $\idxi = \runlen$.  We maintain the induction
    invariant $\ap{\pnmarking}{\ctr}$ is greater than or equal to the number of
    leaf nodes of age $\leq 1$ labelled by $\ctr$ for each counter $\ctr \in
    \counters$.  In the base case the invariant is immediate.

    Now consider the first transition on the run
    \[
        \config{\pnstate_\idxi}{\tree_\idxi} \run
        \config{\pnstate_{\idxi+1}}{\tree_{\idxi+1}} \ .
    \]
    There are several cases depending on the rule used by the transition.
    \begin{enumerate}
        \item
            When the rule is
            $\rrule{\pnstate_\idxi}{\tasingle{\ctr}}{\pnstate_{\idxi+1}}{\tasingle{\ctr}}$
            or $\rrule{\pnstate_\idxi}{\tasingle{\spawnnode}}{\pnstate_{\idxi+1}}{\tasingle{\spawnnode}}$
            we have $\pnstate_\idxi = \pnstate_{\idxi+1}$.  Since this
            transition can only set the age of some leaf node of age $\leq 1$ to
            $0$ we extend the run of $\pn$ with a no-op transition which
            maintains the invariant.
            
        \item
            When the rule is
            $\rrule{\pnstate_\idxi}{\tasingle{\spawnnode}}{\pnstate_{\idxi+1}}{\tasingle{\tap{\splitnode}{\ctr,
            \spawnnode}}}$ we have a rule
            $\pnrule{\pnstate_\idxi}{\set{\incr{\ctr}}}{\pnstate_{\idxi+1}}$ by
            definition.  We extend the run of $\pn$ by firing this rule to
            obtain $\pnmarking_{\idxi+1}$.  For all $\ctr' \neq \ctr$ we know
            that the number of leaf nodes labelled $\ctr'$ of age $\leq 1$ can
            only be reduced, and hence $\ap{\pnmarking_{\idxi+1}}{\ctr'}$
            remains larger or equal.  For $\ctr$ the number of leaf nodes
            may increase by at most $1$, but since we fired $\incr{\ctr}$ the
            invariant is maintained.
            
        \item
            When the rule is
            $\rrule{\pnstate_\idxi}{\tasingle{\ctr}}{\pnstate_{\idxi+1}}{\tasingle{\deadnode}}$
            we have a rule
            $\pnrule{\pnstate_\idxi}{\set{\decr{\ctr}}}{\pnstate_{\idxi+1}}$ by
            definition.  We extend the run of $\pn$ by firing this rule to
            obtain $\pnmarking_{\idxi+1}$.  This is possible since we know there
            is at least one leaf node of age $\leq 1$ labelled $\ctr$ and hence
            $\ap{\pnmarking_\idxi}{\ctr} > 0$ by induction.  Then for all $\ctr'
            \neq \ctr$ we know that the number of leaf nodes labelled $\ctr'$ of
            age $\leq 1$ can only be reduced, and hence
            $\ap{\pnmarking_{\idxi+1}}{\ctr'}$ remains larger or equal.  For
            $\ctr$ the number of leaf nodes decreases by at least $1$, and hence
            the invariant is maintained.
            
        \item
            When the rule is
            $\rrule{\pnstate_\idxi}{\tasingle{\spawnnode}}{\killst{\ctr}{\pnstate_{\idxi+1}}}{\tasingle{\spawnnode}}$
            we know by definition that the transition reaching
            $\config{\pnstate_{\idxi+1}}{\tree_{\idxi+1}}$ is via the rule
            $\rrule{\killst{\ctr}{\pnstate_{\idxi+1}}}{\tasingle{\spawnnode}}{\pnstate_{\idxi+1}}{\tasingle{\spawnnode}}$
            and all transitions in between are of the form
            $\rrule{\killst{\ctr}{\pnstate_{\idxi+1}}}{\tasingle{\ctr'}}{\killst{\ctr}{\pnstate_{\idxi+1}}}{\tasingle{\ctr'}}$
            where $\ctr' \neq \ctr$ or
            $\rrule{\killst{\ctr}{\pnstate_{\idxi+1}}}{\tasingle{\spawnnode}}{\killst{\ctr}{\pnstate_{\idxi+1}}}{\tasingle{\spawnnode}}$.
            Hence the number of leaf nodes of age $\leq 1$ in $\tree_{\idxi+1}$
            will be zero with the label $\ctr$ and less than or equal to the
            number in $\tree_\idxi$ when the label is some other $\ctr'$.  Thus,
            by firing the transition
            $\pnrule{\pnstate_\idxi}{\set{\reset{\ctr}}}{\pnstate_{\idxi+1}}$ to
            obtain $\config{\pnstate_{\idxi+1}}{\pnmarking_{\idxi+1}}$ we
            maintain the invariant.

        \item
            When the rule is
            $\rrule{\pnstate_\idxi}{\tasingle{\spawnnode}}{\pnstate_{\idxi+1}}{\tasingle{\spawnnode}}$ the
            number of leaf nodes with any label of age $\leq 1$ can only
            decrease.  By definition we have a rule
            $\pnrule{\pnstate_\idxi}{\emptyset}{\pnstate_{\idxi+1}}$ which we
            fire to obtain $\config{\pnstate_{\idxi+1}}{\pnmarking_{\idxi+1}}$
            and maintain the invariant.
    \end{enumerate}
    Thus we are done.
\end{proof}

\subsubsection{Reachability}

Finally, using a slight extension of the reduction used for coverability, we can
show that the reachability problem reduces to the regular reachability problem
for GTRS.

The proof is by a minor extension of the coverability reduction.  Naively, since
the reduction uses leaf nodes to store the value of the counters, we could
simply test reachability with respect to a tree automaton $\ta$ that accepts
trees where the number of leaf nodes labelled by each counter matches the target
marking of the counter.  However, this does not work since the reset actions are
encoded by forcing leaf nodes to become fossilised.  Hence, the number of leaf
nodes labelled by a counter will not match the actual marking of the counter.

To overcome this problem we make two modifications to the encoding.  First, when
a counter is being reset, we give all leaves labelled by that counter the
opportunity to rewrite themselves to $\deadnode$.  Furthermore, when we have
reached the target control state, we have the possibility to make a
non-deterministic guess that the target marking has also been reached.  At this
point we let all active leaf nodes labelled by a counter $\ctr$ to rewrite
themselves to be labelled by $\targetctr{\ctr}$.  We then define the target tree
automaton $\ta$ to accept trees where the number of leaves labelled
$\targetctr{\ctr}$ matches the target marking value of the counter, and,
moreover, there are no leaves labelled by a counter $\ctr$.  The second
condition ensures that no node labelled $\ctr$ allowed itself to become
fossilised while labelled $\ctr$ (in particular, during a reset, all nodes
rewrote themselves to $\deadnode$).  Similarly, after guessing that the target
configuration had been reached, all nodes labelled $\ctr$ rewrote themselves to
$\targetctr{\ctr}$ thus ensuring that the tree accurately represents the true
counter values.

\begin{nameddefinition}{def:gtrsreach}{$\gtrsreach{\pn}$}
    Given a reset Petri-net $\pn = \tup{\pnstates, \counters, \pnrules}$, let
    $\gtrscover{\pn} = \tup{\controls, \alphabet, \rules}$ be the senescent GTRS
    from \refdefinition{def:gtrscover}.  We define $\gtrsreach{\pn} =
    \tup{\controls', \alphabet', \rules'}$ with lifespan $1$ where
    \[
        \begin{array}{rcl}
            \controls' & = & \controls \cup \set{\varcontroldest} \\
            \alphabet' & = & \alphabet \cup \setcomp{\targetctr{\ctr}}{\ctr \in
            \counters} \\
            \rules' & = & \rules \ \cup \\
            & & \set{\rrule{\controldest}
                           {\tasingle{\spawnnode}}
                           {\varcontroldest}
                           {\tasingle{\spawnnode}}}\ \cup \\
            & & \setcomp{\rrule{\varcontroldest}
                               {\tasingle{\ctr}}
                               {\varcontroldest}
                               {\tasingle{\targetctr{\ctr}}}}
                        {\ctr \in \counters} \ \cup \\
            & & \setcomp{\rrule{\killst{\ctr}{\pnstate}}
                               {\tasingle{\ctr}}
                               {\killst{\ctr}{\pnstate}}
                               {\tasingle{\deadnode}}}
                        {\ctr \in \counters \land \pnstate \in \controls} \\
        \end{array}
    \]
\end{nameddefinition}

\begin{namedtheorem}{thm:gtrsreach}{Reachability to Regular Reachability}
    The reachability problem for reset Petri-nets can be reduced to the regular
    reachability problem for senescent GTRS and is thus undecidable.
\end{namedtheorem}
\begin{proof}
    Given a reset Petri-net $\pn = \tup{\pnstates, \counters, \pnrules}$ we
    obtain the senescent GTRS $\gtrsreach{\pn} = \tup{\controls, \alphabet,
    \rules}$ with lifespan $1$ from \refdefinition{def:gtrsreach}.  Let
    $\controlinit$ be the initial control state, and (without loss of
    generality) let $\pnzeromarking$ be the initial marking.  Then, let
    $\controldest$ be the target control state and $\pnmarking$ be the target
    marking.  

    First, define $\ta$ to be the tree automaton accepting all trees of the
    following form:
    \begin{enumerate}
        \item
            all internal nodes are labelled $\splitnode$, and

        \item
            there is one leaf node labelled $\spawnnode$, and

        \item
            for each $\ctr \in \counters$ there are exactly
            $\ap{\pnmarking}{\ctr}$ leaf nodes labelled $\targetctr{\ctr}$, and

        \item
            all other leaf nodes are labelled $\deadnode$.
    \end{enumerate}
    It is straightforward to construct such a $\ta$.

    We first show that the existence of a run
    \[
        \config{\controlinit}{\pnzeromarking} \pnrun 
        \config{\controldest}{\pnmarking}
    \]
    of $\pn$ implies there is a run
    \[
        \config{\controlinit}{\tree_1} \run \config{\varcontroldest}{\tree_2}
    \]
    of $\gtrscover{\pn}$ where $\tree_1 = \spawnnode$ and $\tree_2 \in
    \ap{\lang}{\ta}$.

    We can use almost the same proof as the same direction in the proof of
    \reftheorem{thm:gtrscover}, with a minor modification to the induction
    hypothesis.  That is, instead of maintaining for each $\ctr$ that there are
    exactly $\ap{\pnmarking_\idxi}{\ctr}$ leaf nodes \emph{of age $0$} labelled
    by $\ctr$, we maintain the stronger property that there are exactly
    $\ap{\pnmarking_\idxi}{\ctr}$ leaf nodes labelled by $\ctr$.  To do so, we
    need only update the handling of the reset transition.  That is, when
    handling $\reset{\ctr}$ we fire, for each leaf node labelled $\ctr$, the
    rule $\rrule{\killst{\ctr}{\pnstate}} {\tasingle{\ctr}}
    {\killst{\ctr}{\pnstate}} {\tasingle{\deadnode}}$.  Thus, there are no leaf
    nodes labelled $\ctr$.

    We therefore obtain a run to $\config{\controldest}{\tree}$ such that
    $\tree$ has all internal nodes labelled by $\splitnode$, one node labelled
    by $\spawnnode$, exactly $\ap{\pnmarking}{\ctr}$ leaf nodes labelled $\ctr$
    for each $\ctr$, and all other leaf nodes labelled $\deadnode$.  To complete
    the direction, we fire the rule
    $\rrule{\controldest}{\tasingle{\spawnnode}}{\varcontroldest}{\tasingle{\spawnnode}}$
    followed by
    $\rrule{\varcontroldest}{\tasingle{\ctr}}{\varcontroldest}{\tasingle{\targetctr{\ctr}}}$
    for each leaf node labelled by some $\ctr$.  Whence, we reach the
    configuration $\config{\varcontroldest}{\tree_2}$ with $\tree_2 \in
    \ap{\lang}{\ta}$ as required.

    The other direction is also similar to the coverability case.  We take a run 
    \[
        \config{\pnstate_1}{\tree_1} \run \cdots \run
        \config{\pnstate_\runlen}{\tree_\runlen} \pntran
        \config{\varcontroldest}{\tree_\runlen} \run
        \config{\varcontroldest}{\tree}
    \]
    where $\pnstate_1 = \controlinit$, $\tree_1 = \spawnnode$, $\pnstate_\runlen
    = \controldest$, $\tree \in \ap{\lang}{\ta}$, and for all configurations
    between $\config{\pnstate_\idxi}{\tree_\idxi}$ and
    $\config{\pnstate_{\idxi+1}}{\tree_{\idxi+1}}$ (if there are any) there is
    some $\ctr$ such that such that the control state of the configurations is
    of the form $\killst{\ctr}{\pnstate_{\idxi+1}}$.  We then build a run
    \[
        \config{\pnstate_1}{\pnmarking_1} \pntran \cdots \pntran
        \config{\pnstate_\runlen}{\pnmarking_\runlen}
    \]
    of $\pn$ where $\pnmarking_1 = \pnzeromarking$.  Again, we assume that $\pn$
    can perform ``no-op'' transitions.  We induct from $\idxi = 1$ to $\idxi =
    \runlen$.  We maintain the induction invariant $\ap{\pnmarking}{\ctr}$ equal
    to the number of leaf nodes labelled by $\ctr$ for each counter $\ctr \in
    \counters$.  In the base case the invariant is immediate.

    Now consider the first transition on the run
    \[
        \config{\pnstate_\idxi}{\tree_\idxi} \run
        \config{\pnstate_{\idxi+1}}{\tree_{\idxi+1}} \ .
    \]
    There are several cases depending on the rule used by the transition.  In
    all cases, it is key to observe that no leaf node labelled by $\ctr$ can
    become fossilised in the run of $\gtrsreach{\pn}$: if such a leaf were to
    become fossilised it would still be present in $\tree$ and thus we could not
    have $\tree \in \ap{\lang}{\ta}$.
    \begin{enumerate}
        \item
            When the rule is
            $\rrule{\pnstate_\idxi}{\tasingle{\ctr}}{\pnstate_{\idxi+1}}{\tasingle{\ctr}}$
            or $\rrule{\pnstate_\idxi}{\tasingle{\spawnnode}}{\pnstate_{\idxi+1}}{\tasingle{\spawnnode}}$
            we have $\pnstate_\idxi = \pnstate_{\idxi+1}$.  Since this
            transition can only set the age of some leaf node of age $\leq 1$ to
            $0$ we extend the run of $\pn$ with a no-op transition which
            maintains the invariant.
            
        \item
            When the rule is
            $\rrule{\pnstate_\idxi}{\tasingle{\spawnnode}}{\pnstate_{\idxi+1}}{\tasingle{\tap{\splitnode}{\ctr,
            \spawnnode}}}$ we have a rule
            $\pnrule{\pnstate_\idxi}{\set{\incr{\ctr}}}{\pnstate_{\idxi+1}}$ by
            definition.  We extend the run of $\pn$ by firing this rule to
            obtain $\pnmarking_{\idxi+1}$ satisfying the invariant. 

        \item
            When the rule is
            $\rrule{\pnstate_\idxi}{\tasingle{\ctr}}{\pnstate_{\idxi+1}}{\tasingle{\deadnode}}$
            we have a rule
            $\pnrule{\pnstate_\idxi}{\set{\decr{\ctr}}}{\pnstate_{\idxi+1}}$ by
            definition.  We extend the run of $\pn$ by firing this rule to
            obtain $\pnmarking_{\idxi+1}$.  This is possible since we know there
            is at least one leaf node $\ctr$ which cannot be fossilised (as
            remarked above) and hence $\ap{\pnmarking_\idxi}{\ctr} > 0$ by
            induction.
            
        \item
            When the rule is
            $\rrule{\pnstate_\idxi}{\tasingle{\spawnnode}}{\killst{\ctr}{\pnstate_{\idxi+1}}}{\tasingle{\spawnnode}}$
            we know by definition that the transition reaching
            $\config{\pnstate_{\idxi+1}}{\tree_{\idxi+1}}$ is via the rule
            $\rrule{\killst{\ctr}{\pnstate_{\idxi+1}}}{\tasingle{\spawnnode}}{\pnstate_{\idxi+1}}{\tasingle{\spawnnode}}$
            and all transitions in between are of the form
            $\rrule{\killst{\ctr}{\pnstate_{\idxi+1}}}{\tasingle{\ctr'}}{\killst{\ctr}{\pnstate_{\idxi+1}}}{\tasingle{\ctr'}}$
            where $\ctr' \neq \ctr$,
            $\rrule{\killst{\ctr}{\pnstate_{\idxi+1}}}{\tasingle{\spawnnode}}{\killst{\ctr}{\pnstate_{\idxi+1}}}{\tasingle{\spawnnode}}$,
            or
            $\rrule{\killst{\ctr}{\pnstate_{\idxi+1}}}{\tasingle{\ctr}}{\killst{\ctr}{\pnstate_{\idxi+1}}}{\tasingle{\deadnode}}$.
            Indeed, we know that a rule of the first form must be fired for all
            leaf nodes labelled by some $\ctr' \neq \ctr$, and a rule of the
            third form must be fired for each leaf labelled $\ctr$.  If this
            were not the case then some leaf labelled by a counter would become
            fossilised, preventing $\tree$ from being accepted by $\ta$.  Hence,
            the number of leaf nodes labelled $\ctr$ is zero and the number
            labelled by some other counter $\ctr'$ is the same as in
            $\tree_\idxi$.  Thus, we have the invariant.

        \item
            When the rule is
            $\rrule{\pnstate_\idxi}{\tasingle{\spawnnode}}{\pnstate_{\idxi+1}}{\tasingle{\spawnnode}}$
            then by definition we have a rule
            $\pnrule{\pnstate_\idxi}{\emptyset}{\pnstate_{\idxi+1}}$ which we
            fire to obtain $\config{\pnstate_{\idxi+1}}{\pnmarking_{\idxi+1}}$
            and maintain the invariant.
    \end{enumerate}
    Note that the rules
    $\rrule{\controldest}{\tasingle{\spawnnode}}{\varcontroldest}{\tasingle{\spawnnode}}$
    and
    $\rrule{\varcontroldest}{\tasingle{\ctr}}{\varcontroldest}{\tasingle{\targetctr{\ctr}}}$
    can only be fired during the final stage of the run reaching
    $\config{\varcontroldest}{\tree}$.  Moreover, observe that in the final
    phase, in order for $\tree$ to be accepted, all leaf nodes labelled $\ctr$
    must have been rewritten to $\targetctr{\ctr}$ and hence their number must
    match $\pnmarking$.  Thus, in $\tree_\runlen$ the number of leaf nodes
    labelled by $\ctr$ must have matched $\pnmarking$.  Thus, upon reaching
    $\config{\control_\runlen}{\pnmarking_\runlen}$ we have a run of the
    Petri-net as required.
\end{proof}

\section{Reachability Analysis of Senescent GTRSs}

For the following section, fix a senescent GTRS $\gtrs = \tup{\controls,
\alphabet, \rules}$ with lifespan $\lifespan$.  Furthermore, fix an ordering
$\rruler_1, \ldots, \rruler_\numrules$ on the rules in $\rules$.  Thus, we will
use each rule $\rruler \in \rules$ as an index (that is, we use $\rruler$
instead of $\idxi$ when $\rruler = \rruler_\idxi$).  Notice that $\numrules$
denotes the number of rules in $\gtrs$.

Without loss of generality, we assume that all tree automata appearing in the
rules $\rules$ of $\gtrs$ accept at least one tree (rules not satisfying this
condition can be discarded since they cannot be applied).

We show, using ideas from~\cite{lTP12}, that the control state reachability
problem for senescent GTRSs is decidable and is $\acker$-complete. 

\begin{namedtheorem}
      {thm:ackermanncomplete}
      {Ackermann-Completeness of Reachability} 
    It is the case that the control state reachability problem for senescent
    GTRS is $\acker$-complete.
\end{namedtheorem}
\begin{proof}
    $\acker$-hardness follows from \reftheorem{thm:gtrscover} and the
    $\acker$-hardness of the coverability problem for reset
    Petri-nets~\cite{S02,S10}.  

    The upper bound is obtained in the following sections.  In outline, given a
    senescent GTRS $\gtrs$ we obtain from Definition~\ref{def:pnreach} a reset
    Petri-net $\pnreach{\gtrs}$ triply-exponential in the size of $\gtrs$.  From
    Lemma~\ref{lem:reachcorrect} we know that we can decide the control state
    reachability problem for $\gtrs$ via a coverability problem over
    $\pnreach{\gtrs}$.  Since $\acker$ is closed under all
    primitive-recursive reductions, we have our upper bound as required.
\end{proof}

\subsection{Independent Sub-Tree Interfaces}

Our algorithm will non-deterministically construct a representation of a run of
$\gtrs$ witnessing the reachability property.  A key idea is that during the
guessed run, certain sub-trees may operate independently of one another.

That is, suppose we have a tree consisting of a root node $\tnode$ with a left
sub-tree $\tree_1$ and a right sub-tree $\tree_2$.  If it is the case that,
during the run, the complete tree rooted at $\tnode$ is never matched by the LHS
of a rewrite rule, then we may say that $\tree_1$ and $\tree_2$ develop
independently: any rewrite rule applied during the run either matches a sub-tree
of $\tree_1$ or a sub-tree of $\tree_2$, but never depends on both the contents
of $\tree_1$ and $\tree_2$.  Thus, the interaction between $\tree_1$ and
$\tree_2$ is only via the changes to the control state.

When a rewrite rule $\rruler$ is applied to the tree, rewriting a sub-tree
$\tree_1$ to $\tree_2$, there are two possibilities: either $\tree_2$ develops
independently of the rest of the tree for the remainder of the run, or $\tree_2$
appears as a strict sub-tree of a later rewrite rule application.  In the former
case, a new independent sub-tree has been generated, while, in the later, no new
independent tree has been created.  When $\tree_2$ is independent we say that an
independent sub-tree has been generated via rule $\rruler$.

Adapting the \emph{thread interfaces} introduced by La Torre and Parlato in
their analysis of scope-bounded multi-stack pushdown systems~\cite{lTP12}, we
define a notion of \emph{independent sub-tree interfaces}, which we will refer
to simply as \emph{interfaces}.

\begin{definition}[Independent Sub-Tree Interfaces]
    For a senescent GTRS with lifespan $\lifespan$ and rewrite rule-set
    $\set{\rruler_1, \ldots, \rruler_\numrules}$, an \emph{independent sub-tree
    interface} $\istiseq$ is defined to be a sequence $\tup{\control_1, \bit_1,
    \istitrees_1}\ldots\tup{\control_\num, \bit_\num, \istitrees_\num}$ of
    triples in $\controls \times \set{0,1} \times \nats^\numrules$ with $\num
    \leq \lifespan$.
\end{definition}

An interface $\istiseq$ describes the external effect of the evolution of a
sub-tree over up to $\lifespan$ control state changes.  A sequence $\istiseq =
\tup{\control_1, \bit_1, \istitrees_1}, \ldots, \tup{\control_\num, \bit_\num,
\istitrees_\num}$ describes the sequence of control state changes $\control_1,
\ldots, \control_\num$ witnessed by the sub-tree before it becomes fossilised
(or ceases to change for the remainder of the run).  The component $\bit_\idxi$
indicates whether the subtree effected the control state change (via the
application of a rewrite rule modifying both the tree and the control state) or
whether the control state change is supposed to have been made by an external
independent sub-tree.

The final component $\istitrees_\idxi = \tup{\istigen^{\rruler_1}_\idxi, \ldots,
\istigen^{\rruler_\numrules}_\idxi}$ indicates how many new independent
sub-trees are generated during the lifespan of the sub-tree.  That is, during
the run described by $\istiseq$, we have $\istigen^{\rruler}_\idxi$ independent
sub-trees generated using rule $\rruler$ after the control state has been
changed to $\control_\idxi$ but before the change to control state
$\control_{\idxi+1}$ takes place.  Note if a rule both changes the control state
and generates a new independent sub-tree, the sub-tree is considered to have
been generated \emph{after} the control state has changed.

\subsection{Examples of Independent Sub-Tree Interfaces}

In the following, by abuse of notation, let
$\rrule{\control_1}{\tree_1}{\control_2}{\tree_2}$ for given control states
$\control_1, \control_2$ and trees $\tree_1, \tree_2$ denote a rule
$\rrule{\control_1}{\ta_1}{\control_2}{\ta_2}$ where $\ap{\lang}{\ta_1} =
\set{\tree_1}$ and $\ap{\lang}{\ta_2} = \set{\tree_2}$.  Also, recall $\zerovec$
to denotes the tuple $\tup{0, \ldots, 0}$ and $\onevec{\idxi}$ denotes the tuple
where all components are $0$ except the $\idxi$th, which is $1$.

Consider a senescent GTRS with rules $\set{\rruler_1, \ldots, \rruler_5}$ where
\begin{align*}
    \rruler_1 &= \rrule{\control_1}{\tree_0}{\control_2}{\tap{\tnode}{\tree_1,
    \tree_2}} \\
    \rruler_2 &= \rrule{\control_2}{\tree_1}{\control_2}{\tree^1_1} \\
    \rruler_3 &= \rrule{\control_2}{\tree_2}{\control_3}{\tree^1_2} \\
    \rruler_4 &= \rrule{\control_3}{\tree^1_2}{\control_4}{\tree^2_2} \\
    \rruler_5 &= \rrule{\control_4}{\tree^1_1}{\control_5}{\tree^2_1} \ .
\end{align*}
Now consider the run formed from $\rruler_1, \ldots, \rruler_5$ in sequence,
\begin{align*} 
    \config{\control_1}{\tree_0} 
    &\tran \config{\control_2}{\tap{\tnode}{\tree_1, \tree_2}}
    \tran \config{\control_2}{\tap{\tnode}{\tree^1_1, \tree_2}} \\
    &\tran \config{\control_3}{\tap{\tnode}{\tree^1_1, \tree^1_2}} 
    \tran \config{\control_4}{\tap{\tnode}{\tree^1_1, \tree^2_2}} 
    \tran \config{\control_5}{\tap{\tnode}{\tree^2_1, \tree^2_2}} \ .
\end{align*}

Below we present several alternative decompositions of the above run into
interfaces.  In the first decomposition, we take a lifespan of $5$.  In this
case, we may simply have the decomposition 
\[
    \begin{array}{lllll}
        \tup{\control_1, 0, \zerovec}, 
        &\tup{\control_2, 1, \zerovec}, 
        &\tup{\control_3, 1, \zerovec}, 
        &\tup{\control_4, 1, \zerovec}, 
        &\tup{\control_5, 1, \zerovec}
    \end{array}
\]
indicating that no new independent sub-trees are considered to have been
generated, and thus, all control state changes are effected by the evolution of
the original tree.  Note that $\bit_1 = 0$ since the control state was initially
$\control_1$.

However, the above run can also be decomposed if the lifespan is set to $4$.
One such decomposition can be obtained by considering the application of the
rule $\rruler_1$ to generate $\tap{\tnode}{\tree_1, \tree_2}$, where
$\tap{\tnode}{\tree_1, \tree_2}$ is a new independent sub-tree.  Using
$\exttree$ to denote an independent sub-tree that has been generated, we can
decompose the run into two runs
\begin{align*} 
    \config{\control_1}{\tree_0} 
    &\tran \config{\control_2}{\exttree}
\end{align*}
and the run of the generated independent sub-tree
\begin{align*} 
    \config{\control_2}{\tap{\tnode}{\tree_1, \tree_2}}
    &\tran \config{\control_2}{\tap{\tnode}{\tree^1_1, \tree_2}} \\
    &\tran \config{\control_3}{\tap{\tnode}{\tree^1_1, \tree^1_2}} 
    \tran \config{\control_4}{\tap{\tnode}{\tree^1_1, \tree^2_2}} 
    \tran \config{\control_5}{\tap{\tnode}{\tree^2_1, \tree^2_2}} \ .
\end{align*}
These two runs give rise to two independent sub-tree interfaces that can be
combined to represent the original run.
\[
    \begin{array}{lllll}
        \tup{\control_1, 0, \zerovec}, 
        &\tup{\control_2, 1, \onevec{1}} \\

        &\tup{\control_2, 0, \zerovec}, 
        &\tup{\control_3, 1, \zerovec}, 
        &\tup{\control_4, 1, \zerovec}, 
        &\tup{\control_5, 1, \zerovec}
    \end{array}
\]
The upper interface comes from the first part of the decomposed run, and the
lower interface represents the second part.  Note, the lifespan of $4$
is respected and $\onevec{1}$ indicates that an independent sub-tree has been
generated as the RHS of $\rruler_1$.

Finally, we observe that the evolution of $\tree^1_1$ and $\tree^1_2$ are
independent.  Hence, we could be more eager in our generation of independent
sub-trees.  That is, we can decompose the original run into the following runs.
\[
    \config{\control_1}{\tree_0} 
    \tran \config{\control_2}{\exttree}
\]
and
\[
    \config{\control_2}{\tap{\tnode}{\tree_1, \tree_2}}
    \tran \config{\control_2}{\tap{\tnode}{\exttree, \tree_2}} 
    \tran \config{\control_3}{\tap{\tnode}{\exttree, \exttree}} 
\]
where the evolution of $\tree^1_1$ is given by
\[
    \config{\control_2}{\tree^1_1} 
    \tran \config{\control_3}{\tree^1_1} 
    \tran \config{\control_4}{\tree^1_1} 
    \tran \config{\control_5}{\tree^2_1} 
\]
and the evolution of $\tree^1_2$ by
\[
    \config{\control_3}{\tree^1_2} 
    \tran \config{\control_4}{\tree^2_2} 
    \tran \config{\control_5}{\tree^2_2} \ .
\]
Note, the control state change to $\control_5$ was effected by the evolution of
$\tree^1_1$ and the change to $\control_4$ by the evolution of $\tree^1_2$.  The
respective interfaces for the above runs, aligned to suggest how they combine,
are
\[
    \begin{array}{lllll}
        \tup{\control_1, 0, \zerovec},
        &\tup{\control_2, 1, \onevec{1}} \\
        &\tup{\control_2, 0, \onevec{2}},
        &\tup{\control_3, 1, \onevec{4}} \\
        &\tup{\control_2, 0, \zerovec},
        &\tup{\control_3, 0, \zerovec},
        &\tup{\control_4, 0, \zerovec},
        &\tup{\control_5, 1, \zerovec} \\
        & &\tup{\control_3, 0, \zerovec},
        &\tup{\control_4, 1, \zerovec},
        &\tup{\control_5, 0, \zerovec}  \ .
    \end{array}
\]
Each column represents a single control state change.  It is important that in
each column there is exactly one independent sub-tree for which $\bit_\idxi =
1$.  That is, each control state change is performed by exactly one independent
sub-tree.

\subsection{Representing Interfaces}

In this section we show that interfaces $\istiseq$ can be generated as the
Parikh image of regular automata.  

For each rule $\rruler \in \rules$ and sequence $\tup{\control_1, \bit_1},
\ldots, \tup{\control_\num, \bit_\num}$ with $\num \leq \lifespan$ we will build
a regular automaton $\rega$ over the alphabet 
\[
    \istigtrsoalphabet = \setcomp{\tup{\rruler, \idxi}}{\rruler \in \rules \land
    1 \leq \idxi \leq \num} \ .
\]
By abuse of notation, for a run over a word $\word \in \istigtrsalphabet^\ast$,
we define 
\[
    \parikh{\word} = \tup{\istitrees_1, \ldots, \istitrees_\num}
\]
where for all $1 \leq \idxi \leq \num$ we have $\istitrees_\idxi =
\tup{\istigen^{\rruler_1}_\idxi, \ldots, \istigen^{\rruler_\numrules}}$ and
$\istigen^\rruler_\idxi = \countof{\word}{\tup{\rruler, \idxi}}$.  This
generalises to $\parikh{\rega}$ in the natural way.

In particular, we build $\rega$ such that, if $\tup{\istitrees_1, \ldots,
\istitrees_\num}$ is an element of $\parikh{\rega}$ then there is an independent
sub-tree interface 
\[
    \tup{\control_1, \bit_1, \istitrees_1}, \ldots, \tup{\control_\num,
    \bit_\num, \istitrees_\num}
\]
beginning with a tree $\tree \in \ap{\lang}{\ta_1}$ where $\rruler =
\rrule{\control}{\ta}{\control_1}{\ta_1}$.

We obtain the above regular automaton as follows.  First, from $\gtrs$,
$\rruler$ and $\tup{\control_1, \bit_1}, \ldots, \tup{\control_\num, \bit_\num}$
we build a weakly extended sGTRS $\istigtrs$ that simulates a run of $\gtrs$ from a subtree
appearing on the RHS of $\rruler$, passing precisely the control states
$\control_1, \ldots, \control_\num$ and only effecting a control state change
with a rule in $\gtrs$ if $\bit_\idxi = 1$ (else $\istigtrs$ guesses the control
state change).  The output of this sGTRS gives us information on the independent
sub-trees created during the run.  Then, using \reflemma{lem:wgtrsrega} we
obtain a regular automaton as required.

\begin{definition}[$\istigtrs$] 
    Given a senescent GTRS $\gtrs = \tup{\controls, \alphabet, \rules}$ with
    lifespan $\lifespan$, an $\rruler \in \rules$ and sequence $\tup{\control_1,
    \bit_1}, \ldots, \tup{\control_\num, \bit_\num}$ with $\num \leq \lifespan$
    we construct a weakly extended sGTRS $\istigtrs = \tup{\istigtrscontrols,
    \istigtrsalphabet, \istigtrsoalphabet, \istigtrsrules}$ where, letting $\ta$
    be the tree automaton on the RHS of $\rruler$,
    \begin{align*}        
        \istigtrscontrols =& \set{\tup{\control_1, \bit_1, 1}, \ldots,
        \tup{\control_\num, \bit_\num, \num}} \\
        \istigtrsalphabet =& \alphabet \uplus \set{\exttree, \istigtrsinitlab}
        \\ 
        \istigtrsoalphabet =& \setcomp{\tup{\rruler, \idxi}}{\rruler \in \rules
        \land 1 \leq \idxi \leq \num} \\
    \end{align*}
    and $\istigtrsrules$ is the set 
    \[
    \begin{array}{l}
        \set{\orrule{\tup{\control_1, \bit_1,
        1}}{\tasingle{\istigtrsinitlab}}{\empsym}{\tup{\control_1, \bit_1,
        1}}{\ta}} \ \cup \\
        \setcomp{\orrule{\tup{\control_\idxi, \bit_\idxi,
        \idxi}}{\ta_1}{\empsym}{\tup{\control_\idxi, \bit_\idxi,
        \idxi}}{\ta_2}}{\begin{array}{c}1 \leq \idxi \leq \num\ \land \\
            \rrule{\control_\idxi}{\ta_1}{\control_\idxi}{\ta_2} \in
            \rules\end{array}} \ \cup \\
        \setcomp{\orrule{\tup{\control_\idxi, \bit_\idxi,
        \idxi}}{\ta_1}{\empsym}{\tup{\control_{\idxi+1}, 1,
        \idxi+1}}{\ta_2}}{\begin{array}{c}1 \leq \idxi < \num \land
            \bit_{\idxi+1} = 1 \ \land \\
            \rrule{\control_\idxi}{\ta_1}{\control_{\idxi+1}}{\ta_2} \in
            \rules\end{array}} \ \cup \\
        \setcomp{\orrule{\tup{\control_\idxi, \bit_\idxi,
        \idxi}}{\tasingle{\cha}}{\empsym}{\tup{\control_{\idxi+1}, 0,
        \idxi+1}}{\tasingle{\cha}}}{\begin{array}{c} 1 \leq \idxi < \num \land
        \bit_{\idxi+1} = 0 \ \land \\ \cha \in \istigtrsalphabet \land \text{$\cha$
        has arity $0$} \end{array}} \ \cup \\
        \setcomp{\orrule{\tup{\control_\idxi, \bit_\idxi,
        \idxi}}{\ta_1}{\tup{\rruler, \idxi}}{\tup{\control_\idxi, \bit_\idxi,
        \idxi}}{\tasingle{\exttree}}}{\begin{array}{c}1 \leq \idxi \leq \num\
            \land \\ \rruler =
            \rrule{\control_\idxi}{\ta_1}{\control_\idxi}{\ta_2} \in
            \rules\end{array}} \ \cup \\
        \setcomp{\orrule{\tup{\control_\idxi, \bit_\idxi,
        \idxi}}{\ta_1}{\tup{\rruler, \idxi+1}}{\tup{\control_{\idxi+1}, 1,
        \idxi+1}}{\tasingle{\exttree}}}{\begin{array}{c}1 \leq \idxi < \num
            \land \bit_{\idxi+1} = 1 \ \land \\ \rruler =
            \rrule{\control_\idxi}{\ta_1}{\control_{\idxi+1}}{\ta_2} \\ \land \
            \rruler \in \rules\end{array}} 
    \end{array} 
    \]
    and both $\exttree$ and $\istigtrsinitlab$ have arity $0$.
\end{definition}

In the above definition, we use $\istigtrsinitlab$ to be the starting label of
$\istigtrs$, and the first set of rules contains only the rule generating a
(independent sub-)tree that could have been created by rule $\rruler$.  The
second set of rules simply simulates the rules of $\gtrs$ that do not change the
control state.  The next two sets of rules take care of the cases where either
the control state change is effected by the independent sub-tree under
consideration ($\bit_\idxi = 1$), or whether the control state change is
effected by another (independent) part of the tree ($\bit_\idxi = 0$).  The
final two sets of rules take care of the generation of new independent
sub-trees.  That is, when applying a rule of $\gtrs$, instead of the new tree
appearing in the current tree, a place-holder tree (accepted by
$\tasingle{\exttree}$) is created.  Note, since $\exttree$ is a new label, the
place-holder sub-tree cannot be rewritten during the remainder a run of
$\istigtrs$. 

Using $\istigtrs$ we are able to build a regular representation of the
independent sub-trees generated during a run matching a given interface.

\begin{nameddefinition}{def:istirega}{$\istirega$}
    Given a senescent GTRS $\gtrs = \tup{\controls, \alphabet, \rules}$ with
    lifespan $\lifespan$, an $\rruler \in \rules$ and sequence $\tup{\control_1,
    \bit_1}, \ldots, \tup{\control_\num, \bit_\num}$ with $\num \leq \lifespan$
    we construct $\istigtrs$ as above, and then via \reflemma{lem:wgtrsrega} a
    regular automaton $\istirega$ such that there is a run 
    \[
        \config{\tup{\control_1, \bit_1, 1}}{\tree_1} \orun{\word}
        \config{\tup{\control_\num, \bit_\num, \num}}{\tree_2}
    \]
    where $\tree_1 \in \ap{\lang}{\tasingle{\istigtrsinitlab}}$ and $\tree_2$ is
    any tree iff $\parikh{\word} \in \parikh{\istirega}$.
\end{nameddefinition}

\subsection{Reduction to Reset Petri-Nets}

\subsubsection{Interface Summaries}

We reduce the control state reachability problem for senescent GTRSs to the
coverability problem for reset Petri nets.  To do so, we construct a reset Petri
net whose control states hold a sequence $\tup{\control_1,
\bit_1}\ldots\tup{\control_\num, \bit_\num}$ where $\num \leq \lifespan$.  It
will also have a set of counters
\[
    \pnreachctrs{\gtrs} = \setcomp{\istictr{\rruler}{\idxi}}{1 \leq \idxi \leq
    \num} \ .
\]

We will refer to this sequence as an \emph{interface summary}.  Such a summary
will summarise the combination of a number of interfaces.  Each $\control_\idxi$
indicates that the $\idxi$th next control state is $\control_\idxi$ (with
$\control_1$ being the current control state), and $\bit_\idxi$ will indicate
whether an independent sub-tree has already been generated to account for the
control state change.  The value of each counter $\istictr{\rruler}{\idxi}$
indicates how many independent sub-trees are generated using rule $\rruler$
between the $\idxi$th and $(\idxi+1)$th control state by the combination of the
thread interfaces in the summary.

There are two operations we perform on the interface summary: addition and
resolution.

\paragraph{Addition}

Addition refers to the addition of a thread interface to a given summary.
Suppose we have a summary $\config{\tup{\control_1, \bit_1}, \ldots,
\tup{\control_\num, \bit_\num}}{\pnmarking}$ where $\pnmarking$ gives the
valuation of the counters.  Now suppose we want to add to the summary the effect
of the evolution of an independent sub-tree with interface 
\[
    \tup{\control'_1, \bit'_1, \istitrees_1}\ldots\tup{\control'_{\num'},
    \bit'_{\num'}, \istitrees_{\num'}} \ .
\]
We first require the two sequences $\tup{\control_1,
\bit_1}\ldots\tup{\control_\num, \bit_\num}$ and $\tup{\control'_1,
\bit'_1}\ldots\tup{\control'_{\num'}, \bit'_{\num'}}$ to be \emph{compatible}.
There are two conditions for this.
\begin{enumerate}
    \item
        They must agree on their control states.  That is, for all $1 \leq \idxi
        \leq \minimumof{\num, \num'}$ we have $\control_\idxi =
        \control'_\idxi$.

    \item
        At most one independent sub-tree can effect a control state change. That
        is, for all $1 \leq \idxi \leq \minimumof{\num, \num'}$ we do not have
        $\bit_\idxi = \bit'_\idxi = 1$.
\end{enumerate}
We first define the addition only over the states and bits, that is, we define
\[
    \intadd{\tup{\control_1, \bit_1}, \ldots, \tup{\control_\num,
    \bit_\num}}{\tup{\control'_1, \bit'_1}\ldots\tup{\control'_{\num'},
    \bit'_{\num'}}} 
\]
when the two are compatible to be,
\begin{enumerate}
    \item
        when $\num \leq \num'$, 
        \[
            \tup{\control_1, \bit''_1}\ldots\tup{\control_\num,
            \bit''_\num}\tup{\control'_{\num+1},
            \bit'_{\num+1}}\ldots\tup{\control'_{\num'}, \bit'_{\num'}}
            \]
    \item
        and when $\num > \num'$, 
        \[
            \tup{\control_1, \bit''_1}\ldots\tup{\control_{\num'},
            \bit''_{\num'}}\tup{\control_{\num'+1},
            \bit_{\num'+1}}\ldots\tup{\control_{\num}, \bit_{\num}}
        \]
\end{enumerate}
where in both cases $\bit''_\idxi = 1$ if $\bit_\idxi = 1$ or $\bit'_\idxi = 1$,
and otherwise $\bit''_\idxi = 0$.

Then, the addition, 
\[
    \intadd{\config{\pnsumseq}{\pnmarking}}{\tup{\control'_1, \bit'_1,
    \istitrees_1}\ldots\tup{\control'_{\num'}, \bit'_{\num'},
    \istitrees_{\num'}}}
\]
when the two are compatible is $\config{\pnsumseq'}{\pnmarking'}$ where
\[
    \pnsumseq' = \intadd{\pnsumseq}{\tup{\control'_1,
    \bit'_1}\ldots\tup{\control'_{\num'}, \bit_{\num'}}}
\]
and for all $\rruler$ and $\idxi$
\[
    \ap{\pnmarking'}{\istictr{\rruler}{\idxi}} = 
    \begin{cases}
        \ap{\pnmarking}{\istictr{\rruler}{\idxi}} + \istigen^\rruler_\idxi  &
        \idxi \leq \num' \\
        \ap{\pnmarking}{\istictr{\rruler}{\idxi}} & \idxi > \num'  
    \end{cases}
\]
That is, we add the sub-trees generated to the appropriate counters of the Petri
net.

\paragraph{Resolution}

Addition of thread interfaces to the summary handles the evolution of new
independent sub-trees generated on the run between the current control state
$\control_1$ and the next $\control_2$.  Once all such trees have been accounted
for, we can perform \emph{resolution}.  That is, we remove the completed first
round from the summary.  Note that this can only be done if $\bit_2 = 1$, that
is, some independent sub-tree has taken responsibility for the change to the
next control state $\control_2$.

We thus define 
\[
    \intres{\config{\tup{\control_1, \bit_1}, \tup{\control_2, \bit_2}, \ldots,
    \tup{\control_\num, \bit_\num}}{\pnmarking}} = \config{\tup{\control_2,
    \bit_2}, \ldots, \tup{\control_\num, \bit_\num}}{\pnmarking'}
\]
when $\bit_2 = 1$ and where 
\[
    \ap{\pnmarking'}{\istictr{\rruler}{\idxi}} =
    \ap{\pnmarking}{\istictr{\rruler}{\idxi+1}}
\]
for all $1 \leq \idxi < \num$, and 
\[
    \ap{\pnmarking'}{\istictr{\rruler}{\num}} = 0 \ .
\]

\subsubsection{Reduction to Coverability}

We define a reset Petri-net that has a positive solution to the coverability
problem iff the control state reachability problem for the given senescent GTRS
$\gtrs$ is also positive.

\label{sec:rule1} 
For technical convenience, we assume $\rruler_1 =
\rrule{\controlinit}{\ta_1}{\controlinit}{\ta_2}$ where $\ta_1$ accepts no trees
and $\ta_2$ accepts only the initial tree $\treeinit$.  The assumption of such a
rule does not allow more runs of $\gtrs$ since $\ta_1$ matches no trees.  

\paragraph{Initial Configuration}

The Petri-net begins in a configuration
\[
    \config{\tup{\controlinit, 1}}{\pninitmarking}
\]
where 
\[
    \ap{\pninitmarking}{\istictr{\rruler}{\idxi}} = 
    \begin{cases}
        1 & \text{if $\idxi = 1$ and $\rruler = \rruler_1$} \\ 
        0 & \text{otherwise} \ . 
    \end{cases}
\]
This means that the Petri net is simulating a configuration of the senescent
GTRS where the control state is $\controlinit$ and the only independent sub-tree
that can be generated is the $\treeinit$.

\paragraph{Addition of New Interfaces}

The Petri net can then begin simulating execution as follows.  It will
non-deterministically guess the independent sub-tree interface of the initial
tree during a satisfying run of the reachability problem.  It will do this by
subtracting $1$ from the variable $\istictr{\rruler_1}{1}$ then guessing a
sequence $\tup{\control_1, \bit_1} \ldots \tup{\control_\num, \bit_\num}$.
Since there are only a finite number of possibilities for such a sequence, the
guess can be made in the control state.  To fully guess an interface, however,
the Petri net must also guess the values of $\istitrees_\idxi$ for each $1 \leq
\idxi \leq \num$.  To do this it will simulate (in its control state) the
automaton $\rega$ generated by \refdefinition{def:istirega}, but, instead of
outputting a symbol $\tup{\rruler, \idxi}$, it will increment the counter
$\istictr{\rruler}{\idxi}$.

In the manner described above, the Petri net can update its control state and
counter values to perform an addition
\[
    \intadd{\config{\tup{\control_1, \bit_1}, \ldots, \tup{\control_\num,
    \bit_\num}}{\pnmarking}}{\tup{\control'_1, \bit'_1,
    \istitrees_1}\ldots\tup{\control'_{\num'}, \bit'_{\num'},
    \istitrees_{\num'}}}
\]
for the interface summary it is currently storing in its control state and
counters, and a guessed new interface generated from some available independent
sub-tree.

\paragraph{Resolving The Current Interface Summary}

Given a configuration 
\[ 
    \config{\tup{\control_1, \bit_1}, \ldots, \tup{\control_\num,
    \bit_\num}}{\pnmarking} 
\] 
the Petri net can non-deterministically decide whether to add another interface
to the summary, or whether (if $\bit_2 = 1$) perform a resolution step.

To perform resolution the Petri net first updates the control state to obtain
the sequence $\tup{\control_2, \bit_2}, \ldots, \tup{\control_\num, \bit_\num}$
(that is, deletes the first tuple), and then updates its marking to 
\[
    \ap{\pnmarking'}{\istictr{\rruler}{\idxi}} =
    \ap{\pnmarking}{\istictr{\rruler}{\idxi+1}}
\]
for all $1 \leq \idxi < \num$, and 
\[
    \ap{\pnmarking'}{\istictr{\rruler}{\num}} = 0 \ .
\]
It does this incrementally from $\idxi = 1$ to $\idxi = \num$.  For each given
$\idxi$ it first uses reset transitions to zero each counter
$\istictr{\rruler}{\idxi}$.  Then, when $\idxi < \num$, it performs a loop for
each counter, decrementing $\istictr{\rruler}{\idxi+1}$ and incrementing
$\istictr{\rruler}{\idxi}$.  It repeats this loop a non-deterministic number of
times before moving to the next counter.

Note that this is not a faithful implementation of the resolution operation
since the Petri net cannot ensure that it transfers $\istictr{\rruler}{\idxi+1}$
to $\istictr{\rruler}{\idxi}$ in its entirety, merely that
$\istictr{\rruler}{\idxi} \leq \istictr{\rruler}{\idxi + 1}$.  However,
``forgetting'' the existence of independent sub-trees merely restricts the
number of runs and does not add new behaviours.  Hence such an inaccuracy is
benign (since it is still possible to transfer all sub-trees).  The reset
operation is used to ensure that no leakage occurs between each $\idxi$.

\paragraph{Formal Definition}

We give the formal definition of the reset Petri net $\pnreach{\gtrs}$ that
simulates $\gtrs$ with respect to the control state reachability problem.  For
each $\pnsumseq = \tup{\control_1, \bit_1}\ldots\tup{\control_\num, \bit_\num}$
with $1 \leq \num \leq \lifespan$ and rule $\rruler \in \rules$, let
$\istifulla{\rruler}{\pnsumseq} = \tup{\ististates{\rruler}{\pnsumseq},
\istigtrsoalphabet, \istitrans{\rruler}{\pnsumseq},
\istiinitst{\rruler}{\pnsumseq}, \set{\istifinst{\rruler}{\pnsumseq}}}$ be the
regular automaton obtained via \refdefinition{def:istirega} and without loss of
generality assume $\istifulla{\rruler}{\pnsumseq}$ has the unique initial state
$\istiinitst{\rruler}{\pnsumseq}$ and final state
$\istifinst{\rruler}{\pnsumseq}$.  We assume for all $\rruler$ and $\pnsumseq$
that $\istifulla{\rruler}{\pnsumseq}$ have disjoint state sets.

\begin{nameddefinition}{def:pnreach}{$\pnreach{\gtrs}$}
    Given the senescent GTRS $\gtrs$ (with notation and assumptions as described
    in this section), we define the reset Petri net $\pnreach{\gtrs} =
    \tup{\pnreachstates{\gtrs}, \pnreachctrs{\gtrs}, \pnreachrules{\gtrs}}$
    where $\pnreachctrs{\gtrs}$ is defined above and
    \begin{align*}
        \pnsumseqs =& \setcomp{\tup{\control_1, \bit_1} \ldots
        \tup{\control_\num, \bit_\num} \in \brac{\controls \times
        \set{0,1}}^\num}{1 \leq \num \leq \lifespan}  \\ 
        \pnreachstates{\gtrs} =& \pnsumseqs \cup \setcomp{\tup{\pnsumseq,
        \nfastate} \in \pnsumseqs \times
        \ististates{\rruler}{\pnsumseq'}}{\pnsumseq' \in \pnsumseqs \land
        \rruler \in \rules} \cup
        \setcomp{\pnshiftst{\pnsumseq}{\idxi}}{\pnsumseq \in \pnsumseqs \land 1
        \leq \idxi \leq \lifespan} \\
        \pnreachrules{\gtrs} =& \pnrulesadd \cup \pnrulesres
    \end{align*}
    where
    \begin{align*}
        \pnrulesadd =&
        \setcomp{\pnrule{\pnsumseq}{\decr{\istictr{\rruler}{1}}}{\tup{\pnsumseq_1,
        \istiinitst{\rruler}{\pnsumseq_2}}}}{\rruler \in \rules \land \pnsumseq,
        \pnsumseq_1, \pnsumseq_2 \in \pnsumseqs \land \pnsumseq_1 =
        \intadd{\pnsumseq}{\pnsumseq_2}} \ \cup \\
        & \setcomp{\pnrule{\tup{\pnsumseq,
        \nfastate}}{\incr{\istictr{\rruler}{\idxi}}}{\tup{\pnsumseq,
        \nfastate'}}}{\pnsumseq \in \pnsumseqs \land \exists \rruler',
        \pnsumseq' \text{ s.t. } \nfastate \nfatran{\tup{\rruler,
        \idxi}} \nfastate' \in \istitrans{\rruler'}{\pnsumseq'}} \ \cup \\
        & \setcomp{\pnrule{\tup{\pnsumseq,
        \istifinst{\rruler}{\pnsumseq'}}}{\emptyset}{\pnsumseq}}{\rruler \in
        \rules \land \pnsumseq' \in \pnsumseqs} 
    \end{align*}
    and 
    \begin{align*}
        \pnrulesres =& \setcomp{\pnrule{\pnsumseq}{\pnops}{
        \pnshiftst{\pnsumseq'}{1}}}{\begin{array}{c} \pnsumseq, \pnsumseq' \in
            \pnsumseqs \ \land \\ \pnsumseq' = \tup{\control_2, \bit_2} \ldots
            \tup{\control_\num, \bit_\num}\ \land \\ \pnsumseq =
            \tup{\control_1, \bit_1}\pnsumseq' \land \bit_2 = 1 \ \land \\
            \pnops = \setcomp{\reset{\istictr{\rruler}{1}}}{\rruler \in \rules}
        \end{array}} \ \cup \\
        &
        \setcomp{\pnrule{\pnshiftst{\pnsumseq}{\idxi}}{\set{\decr{\istictr{\rruler}{\idxi+1}},
        \incr{\istictr{\rruler}{\idxi}}}}{\pnshiftst{\pnsumseq}{\idxi}}}{\pnsumseq
        \in \pnsumseqs \land 1 \leq \idxi < \lifespan} \ \cup \\
        & 
        \setcomp{\pnrule{\pnshiftst{\pnsumseq}{\idxi}}{\pnops}{\pnshiftst{\pnsumseq}{\idxi+1}}}{\pnsumseq
        \in \pnsumseqs \land 1 \leq \idxi < \lifespan \land \pnops =
        \setcomp{\reset{\istictr{\rruler}{\idxi+1}}}{\rruler \in \rules}} \ \cup
        \\
        &
        \setcomp{\pnrule{\pnshiftst{\pnsumseq}{\lifespan}}{\emptyset}{\pnsumseq}}{\pnsumseq
        \in \pnsumseqs}
    \end{align*}
\end{nameddefinition}

Note that the size of $\pnreach{\gtrs}$ is dominated by the size of the regular
automata $\istifulla{\rruler}{\pnsumseq}$.  Thus, the size of $\pnreach{\gtrs}$
is triply exponential in the size of $\gtrs$.

\subsection{Correctness of Reduction}

We prove that the control state reachability problem for $\gtrs$ has a positive
solution iff the coverability problem for $\pnreach{\gtrs}$, initial
configuration $\config{\tup{\controlinit, 1}}{\pninitmarking}$, and target
configuration $\config{\tup{\controldest, 1}}{\pnzeromarking}$ is also positive.

We prove each direction in the sections that follow.

\begin{namedlemma}{lem:reachcorrect}{Correctness of Reduction}
    For a given senescent GTRS $\gtrs$ with lifespan $\lifespan$, control states
    $\controlinit$ and $\controldest$, and tree $\treeinit$, there is a lifespan
    restricted run 
    \[
        \config{\controlinit}{\treeinit} \tran \cdots \tran
        \config{\controldest}{\tree}
    \]
    for some $\tree$ of $\gtrs$ iff there is a run
    \[
        \config{\tup{\controlinit, 1}}{\pninitmarking} \pnrun
        \config{\tup{\controldest, 1}}{\pnmarking}
    \]
    of $\pnreach{\gtrs}$ for some $\pnzeromarking \pncovers \pnmarking$.
\end{namedlemma}
\begin{proof}
    From Lemma~\ref{lem:gtrs2petri} and Lemma~\ref{lem:petri2gtrs} below.        
\end{proof}

In the following, let the marking $\submarking{\pnmarking}{\rruler}$ be the
marking 
\[
    \ap{\brac{\submarking{\pnmarking}{\rruler}}}{\istictr{\rruler'}{\idxj}}
    = 
    \begin{cases}
        \ap{\pnmarking}{\istictr{\rruler'}{\idxj}} & \rruler' \neq
        \rruler \lor \idxj > 1 \\
        \ap{\pnmarking}{\istictr{\rruler'}{\idxj}} - 1 & \rruler' =
        \rruler \land \idxj = 1  \ .
    \end{cases}
\]

\subsubsection{From Senescent GTRS to Reset Petri Nets}

\begin{lemma} \label{lem:gtrs2petri}
    For a given senescent GTRS $\gtrs$ with lifespan $\lifespan$, control states
    $\controlinit$ and $\controldest$, and tree $\treeinit$, there is a lifespan
    restricted run 
    \[
        \config{\controlinit}{\treeinit} \tran \cdots \tran
        \config{\controldest}{\tree}
    \]
    for some $\tree$ of $\gtrs$ only if there is a run
    \[
        \config{\tup{\controlinit, 1}}{\pninitmarking} \pnrun
        \config{\tup{\controldest, 1}}{\pnmarking}
    \]
    of $\pnreach{\gtrs}$ for some $\pnzeromarking \pncovers \pnmarking$.
\end{lemma}
\begin{proof}
    The proof proceeds in two steps.  We begin with a run 
    \[
        \config{\control_1}{\tree_1} \tran \cdots \tran
        \config{\control_\runlen}{\tree_\runlen}
    \]
    of $\gtrs$.  First we deconstruct this run into independent sub-trees,
    coupled with their interfaces.  From this deconstruction, we then build a
    run of $\pnreach{\gtrs}$.

    The deconstruction is a sequence
    \[
        \tup{\control_1, \istid{1}, \istrun{1}, \istint{1}} \ldots
        \tup{\control_\runlen, \istid{\runlen}, \istrun{\runlen},
        \istint{\runlen}}
    \]
    where for all $1 \leq \idxi \leq \runlen$ we have that $\istid{\idxi}$
    assigns to each node of $\tree_\idxi$ a natural number indicating the ID of
    the independent sub-tree the node currently belongs to, $\istrun{\idxi}$
    maps each ID (natural number) to a run of some $\istigtrs$ (where the
    appropriate $\rruler, \control_1, \bit_1, \ldots, \control_\num, \bit_\num$
    are defined on-the-fly) which is the evolution of the independent sub-tree
    in the run up to $\idxi$, and finally $\istint{\idxi}$ maps each ID to an
    independent sub-tree interface also representing the run up to $\idxi$.

    We build this sequence by induction (simultaneously arguing its existence).
    We begin with $\istid{1}$ mapping each node in $\tree_1$ to $1$, then
    \begin{align*}
        \ap{\istrun{1}}{1} =& \config{\tup{\control_1, 0, 1}}{\tree_1} \\
        \ap{\istint{1}}{1} =& \tup{\control_1, 0, \zerovec} \ .
    \end{align*}
    Now, inductively take $\tup{\control_\idxi, \istid{\idxi}, \istrun{\idxi},
    \istint{\idxi}}$ and consider the transition
    \[
        \config{\control_\idxi}{\csub{\context}{\tree}} \tran
        \config{\control_{\idxi+1}}{\csub{\context}{\tree'}}
    \]
    where $\tree_\idxi = \csub{\context}{\tree}$ and $\tree_{\idxi+1} =
    \csub{\context}{\tree'}$ and the transition is via rule $\rruler \in \rules$.

    Let $\tnodeparent$ be the parent of $\tree$ in $\tree_\idxi$, should it
    exist (it does not exist if $\tree = \csub{\context}{\tree}$).  There are
    now two cases.
    \begin{enumerate}
        \item 
            There is some $\idxj > \idxi$ such that in the transition
            \[
                \config{\control_\idxj}{\csub{\context'}{\tree^1}} \tran
                \config{\control_{\idxj+1}}{\csub{\context'}{\tree^2}}
            \]
            of the run, the node $\tnodeparent$ appears in $\tree^1$.

            In this case, $\tree'$ cannot form a new independent sub-tree since
            it appears as part of the run over the sub-tree including
            $\tnodeparent$.

            We define, for each $\tnode$ in $\tree_{\idxi+1}$, 
            \[
                \ap{\istid{\idxi+1}}{\tnode} =
                \begin{cases}
                    \ap{\istid{\idxi}}{\tnode} & \text{$\tnode$ is in
                    $\context$} \\
                    \ap{\istid{\idxi}}{\tnodeparent} & \text{$\tnode$ is in
                    $\tree'$} \ .
                \end{cases}
            \]
            
            There are now two further cases, depending on whether the control
            state is changed by the transition.
            \begin{enumerate}
                \item 
                    When $\control_\idxi = \control_{\idxi+1}$ we define, for
                    each $\idxj$ in the image of $\istid{\idxi+1}$,
                    \begin{align*} 
                        \ap{\istrun{\idxi+1}}{\idxj} &=
                        \begin{cases}
                            \ap{\istrun{\idxi}}{\idxj} & \idxj \neq
                            \ap{\istid{\idxi}}{\tnodeparent} \\
                            \runsym & \idxj = \ap{\istid{\idxi}}{\tnodeparent}
                            \\ 
                        \end{cases} \\
                        \ap{\istint{\idxi+1}}{\idxj} &=
                        \ap{\istint{\idxi}}{\idxj}
                    \end{align*}
                    where we define $\runsym$ as follows.  We know
                    $\ap{\istrun{\idxi}}{\ap{\istid{\idxi}}{\tnodeparent}}$ is a
                    run to some configuration $\config{\tup{\control_{\idxi},
                    \bit, \idxi'}}{\csub{\context'}{\tree}}$ where there is some
                    $\context''$ such that $\tree_\idxi =
                    \csub{\context''}{\csub{\context'}{\tree}}$.  We define
                    \[
                        \runsym =
                        \ap{\istrun{\idxi}}{\ap{\istid{\idxi}}{\tnodeparent}}
                        \otran{\empsym} \config{\tup{\control_{\idxi+1}, \bit,
                        \idxi'}}{\csub{\context'}{\tree'}}
                    \]
                    which can be seen to be a transition of $\istigtrs$ due to
                    $\rruler \in \rules$.  
                    
                    Note, in all other cases we keep the same run as in
                    $\istrun{\idxi}$.  Since $\control_\idxi =
                    \control_{\idxi+1}$ we maintain (for all independent
                    sub-trees that have not expired their lifespan) that
                    $\istrun{\idxi}$ tracks the run up to $\idxi$.

                \item 
                    When $\control_\idxi \neq \control_{\idxi+1}$ we define
                    $\ap{\istrun{\idxi+1}}{\idxj}$ for each $\idxj$ in the image
                    of $\istid{\idxi+1}$.  There are two cases.  
                    
                    When $\idxj \neq \ap{\istid{\idxi}}{\tnodeparent}$ we know
                    $\ap{\istrun{\idxi}}{\idxj}$ is a run to some configuration
                    $\config{\tup{\control, \bit, \idxi'}}{\tree''}$.  If
                    $\idxi' = \lifespan$ then the nodes in $\tree''$ can no
                    longer be rewritten after the control state change (they are
                    fossilised).  Hence, we define $\ap{\istrun{\idxi+1}}{\idxj}
                    = \ap{\istrun{\idxi}}{\idxj}$.  When $\idxi' < \lifespan$,
                    we define    
                    \[
                    \ap{\istrun{\idxi+1}}{\idxj} = \ap{\istrun{\idxi}}{\idxj}
                    \otran{\empsym} \config{\tup{\control_{\idxi+1}, 0,
                    \idxi'+1}}{\tree''} \ .
                    \]
                    Such a transition is always possible by the definition of
                    $\istigtrs$.

                    When $\idxj = \ap{\istid{\idxi}}{\tnodeparent}$ we know
                    $\ap{\istrun{\idxi}}{\idxj}$ is a run to some configuration
                    $\config{\tup{\control_{\idxi}, \bit,
                    \idxi'}}{\csub{\context'}{\tree}}$ where there is some
                    $\context''$ such that $\tree_\idxi =
                    \csub{\context''}{\csub{\context'}{\tree}}$.  We define
                    \[
                        \runsym = \ap{\istrun{\idxi}}{\idxj} \otran{\empsym}
                        \config{\tup{\control_{\idxi+1}, 1,
                        \idxi'+1}}{\csub{\context'}{\tree'}}
                    \]
                    which can be seen to be a transition of $\istigtrs$ due to
                    $\rruler \in \rules$.  

                    Finally, all $\idxj$ we define $\ap{\istint{\idxi}}{\idxj}$
                    when $\ap{\istint{\idxi}}{\idxj}$ has $\lifespan$ tuples,
                    and otherwise,
                    \[
                        \ap{\istint{\idxi+1}}{\idxj} = 
                        \begin{cases}
                            \ap{\istint{\idxi}}{\idxj} \tup{\control_{\idxi+1},
                            1, \zerovec} &  \idxj =
                            \ap{\istid{\idxi}}{\tnodeparent} \\ 
                            \ap{\istint{\idxi}}{\idxj} \tup{\control_{\idxi+1},
                            0, \zerovec} & \idxj \neq
                            \ap{\istid{\idxi}}{\tnodeparent} \ .
                        \end{cases}
                    \]
            \end{enumerate}

        \item 
            Either $\tree = \csub{\context}{\tree}$, or there is no $\idxj >
            \idxi$ such that in the transition
            \[
                \config{\control_\idxj}{\csub{\context'}{\tree^1}} \tran
                \config{\control_{\idxj+1}}{\csub{\context'}{\tree^2}}
            \]
            of the run, the node $\tnodeparent$ appears in $\tree^1$.

            In this case, $\tree'$ can form a new independent sub-tree since its
            parent node is not read during the remainder of the run.  
           
            Let $\idxjnew$ be a natural number not in the image of
            $\istid{\idxi}$.  We define, for each $\tnode$ in $\tree_{\idxi+1}$, 
            \[
                \ap{\istid{\idxi+1}}{\tnode} =
                \begin{cases}
                    \ap{\istid{\idxi}}{\tnode} & \text{$\tnode$ is in
                    $\context$} \\
                    \idxjnew & \text{$\tnode$ is in $\tree'$} \ .
                \end{cases}
            \]
            
            There are now two further cases, depending on whether the control
            state is changed by the transition.  In the following,  recall
            $\onevec{\idxe}$ is the vector that is zero in all components except
            the $\idxe$th, which is $1$.
            \begin{enumerate}
                \item 
                    When $\control_\idxi = \control_{\idxi+1}$ we define, for
                    each $\idxj$ in the image of $\istid{\idxi+1}$,
                    \begin{align*} 
                        \ap{\istrun{\idxi+1}}{\idxj} &=
                        \begin{cases}
                            \ap{\istrun{\idxi}}{\idxj} & \idxj \neq
                            \ap{\istid{\idxi}}{\tnodeparent} \land \idxj \neq
                            \idxjnew \\
                            \runsym & \idxj = \ap{\istid{\idxi}}{\tnodeparent}
                            \\ 
                            \config{\tup{\control_{\idxi+1}, 0, 1}}{\tree'} &
                            \idxj = \idxjnew  
                        \end{cases} \\
                        \ap{\istint{\idxi+1}}{\idxj} &= 
                        \begin{cases}
                            \ap{\istint{\idxi}}{\idxj} & \idxj \neq
                            \ap{\istid{\idxi}}{\tnodeparent} \land \idxj \neq
                            \idxjnew \\
                            \istiseq & \idxj = \ap{\istid{\idxi}}{\tnodeparent}
                            \\ 
                            \tup{\control_{\idxi+1}, 0, \zerovec} & \idxj =
                            \idxjnew  
                        \end{cases}
                    \end{align*}
                    where we define $\runsym$ and $\istiseq$ as follows.  
                    
                    First, for $\runsym$, we know
                    $\ap{\istrun{\idxi}}{\ap{\istid{\idxi}}{\tnodeparent}}$ is a
                    run to some configuration $\config{\tup{\control_{\idxi},
                    \bit, \idxi'}}{\csub{\context'}{\tree}}$ where there is some
                    $\context''$ such that $\tree_\idxi =
                    \csub{\context''}{\csub{\context'}{\tree}}$.  We define
                    \[
                        \runsym =
                        \ap{\istrun{\idxi}}{\ap{\istid{\idxi}}{\tnodeparent}}
                        \otran{\tup{\rruler, \idxi'}}
                        \config{\tup{\control_{\idxi+1}, \bit,
                        \idxi'}}{\csub{\context'}{\exttree}}
                    \]
                    which can be seen to be a transition of $\istigtrs$ due to
                    $\rruler \in \rules$.  
                    
                    Note, in all other cases we keep the same run as in
                    $\istrun{\idxi}$.  Since $\control_\idxi =
                    \control_{\idxi+1}$ we maintain (for all independent
                    sub-trees that have not expired their lifespan) that
                    $\istrun{\idxi}$ tracks the run up to $\idxi$.

                    Next, we define $\istiseq$.  Let
                    $\ap{\istint{\idxi}}{\tnodeparent} = \istiseq'
                    \tup{\control_\idxi, \bit, \istitrees}$ for some $\istiseq',
                    \bit$ and $\istitrees$.  We define
                    \[
                        \istiseq = \istiseq' \tup{\control_\idxi, \bit,
                        \istitrees + \onevec{\idxe}}
                    \]
                    when $\rruler = \rruler_\idxe$.

                \item 
                    When $\control_\idxi \neq \control_{\idxi+1}$ we define
                    $\ap{\istrun{\idxi+1}}{\idxj}$ for each $\idxj$ in the image
                    of $\istid{\idxi+1}$.  There are three cases.  

                    When $\idxj = \idxjnew$, we have
                    $\ap{\istrun{\idxi+1}}{\idxj} =
                    \config{\control_{\idxi+1}}{\tree'}$.

                    When $\idxj = \ap{\istid{\idxi}}{\tnodeparent}$ we know
                    $\ap{\istrun{\idxi}}{\idxj}$ is a run to some configuration
                    $\config{\tup{\control_{\idxi}, \bit,
                    \idxi'}}{\csub{\context'}{\tree}}$ where there is some
                    $\context''$ such that $\tree_\idxi =
                    \csub{\context''}{\csub{\context'}{\tree}}$.  We define
                    \[
                        \runsym = \ap{\istrun{\idxi}}{\idxj}
                        \otran{\tup{\rruler, \idxi+1}}
                        \config{\tup{\control_{\idxi+1}, 1,
                        \idxi'+1}}{\csub{\context'}{\exttree}}
                    \]
                    which can be seen to be a transition of $\istigtrs$ due to
                    $\rruler \in \rules$.

                    When $\idxj \neq \ap{\istid{\idxi}}{\tnodeparent}$ and
                    $\idxj \neq \idxjnew$ we know $\ap{\istrun{\idxi}}{\idxj}$
                    is a run to some configuration $\config{\tup{\control, \bit,
                    \idxi'}}{\tree''}$.  If $\idxi' = \lifespan$ then the nodes
                    in $\tree''$ can no longer be rewritten after the control
                    state change (they are fossilised).  Hence, we define
                    $\ap{\istrun{\idxi+1}}{\idxj} = \ap{\istrun{\idxi}}{\idxj}$.
                    When $\idxi' < \lifespan$, we define    
                    \[
                        \ap{\istrun{\idxi+1}}{\idxj} =
                        \ap{\istrun{\idxi}}{\idxj} \otran{\empsym}
                        \config{\tup{\control_{\idxi+1}, 0, \idxi'+1}}{\tree''}
                        \ .
                    \]
                    Such a transition is always possible by the definition of
                    $\istigtrs$.

                    Finally, we define for all $\idxj$ that
                    $\ap{\istint{\idxi}}{\idxj}$ when
                    $\ap{\istint{\idxi}}{\idxj}$ has $\lifespan$ tuples, and
                    otherwise, when $\rruler = \rruler_\idxe$,
                    \[
                        \ap{\istint{\idxi+1}}{\idxj} = 
                        \begin{cases}
                            \tup{\control_{\idxi+1}, 0, \zerovec} & \idxj =
                            \idxjnew \\
                            \ap{\istint{\idxi}}{\idxj} \tup{\control_{\idxi+1},
                            1, \onevec{\idxe}} & \idxj =
                            \ap{\istid{\idxi}}{\tnodeparent} \\ 
                            \ap{\istint{\idxi}}{\idxj} \tup{\control_{\idxi+1},
                            0, \zerovec} & \idxj \neq
                            \ap{\istid{\idxi}}{\tnodeparent} \land \idxj \neq
                            \idxjnew \ .
                        \end{cases}
                    \]
            \end{enumerate}
    \end{enumerate}

    We have now defined our deconstruction
    \[
        \tup{\control_1, \istid{1}, \istrun{1}, \istint{1}} \ldots
        \tup{\control_\runlen, \istid{\runlen}, \istrun{\runlen},
        \istint{\runlen}}  
    \]
    from which we will construct a run
    \[
        \config{\tup{\controlinit, 1}}{\pninitmarking} \pntran
        \config{\pnsumseq_1}{\pnmarking_1} \pntran \cdots \pntran
        \config{\pnsumseq_{\runlen'}}{\pnmarking_{\runlen'}} 
    \]
    of $\pnreach{\gtrs}$, where $\pnsumseq_{\runlen'} = \tup{\controldest, 1}$
    and $\pnzeromarking \pncovers \pnmarking_{\runlen'}$ as required.

    Note that the independent sub-tree interfaces in the decomposition of the
    run are given by $\istint{\runlen}$.  We will now iterate over the run of
    $\gtrs$ from $\idxi = 1$ to $\idxi = \runlen$, building the required run of
    $\pnreach{\gtrs}$ as we go.  First, let $\runsym_1$ be the run
    \[
        \config{\tup{\controlinit, 1}}{\pninitmarking} \pnrun
        \config{\pnsumseq}{\pnmarking} 
    \]
    where, recalling that the ID $1$ denotes the independent sub-tree
    corresponding to the evolution of $\treeinit$ and the independent sub-trees
    it generates, we have 
    \[
        \config{\pnsumseq}{\pnmarking} = \intadd{\config{\tup{\controlinit,
        1}}{\pnzeromarking}}{\ap{\istint{\runlen}}{1}} \ .
    \]
    It remains to prove that such a run exists.  For this, take the $\istigtrs$
    defined by the initial rule $\rruler_1$ (assumed in
    Section~\ref{sec:rule1}), and the sequence $\pnsumseq' = \tup{\control_1,
    \bit_1} \ldots \tup{\control_\num, \bit_\num}$ where
    $\ap{\istint{\runlen}}{1} = \tup{\control_1, \bit_1, \istitrees_1} \ldots
    \tup{\control_\num, \bit_\num, \istitrees_\num}$.  Note that in this case
    $\pnsumseq$ differs from $\pnsumseq'$ only in the value of $\bit_1$ (it is
    $0$ in $\pnsumseq'$ and $1$ in $\pnsumseq$).  Further, observe that by
    construction $\ap{\istrun{\runlen}}{1}$ gives a run of $\istigtrs$
    outputting $\word$ such that $\parikh{\word} = \tup{\istitrees_1, \ldots,
    \istitrees_\num}$.  From this we obtain a run of the accompanying regular
    automaton $\istirega$ with the same Parikh image and thus a run
    \[
        \config{\tup{\controlinit, 1}}{\pninitmarking}
        \pnlabtran{\decr{\istictr{\rruler_1}{1}}} \config{\tup{\pnsumseq,
        \istiinitst{\rruler_1}{\pnsumseq'}}}{\pnzeromarking} \pnrun
        \config{\tup{\pnsumseq, \istifinst{\rruler_1}{\pnsumseq'}}}{\pnmarking}
        \pntran \config{\pnsumseq}{\pnmarking} 
    \]
    of $\pnreach{\gtrs}$ as required.
   
    Now, inductively assume a run $\runsym_\idxi$.  Let
    $\config{\pnsumseq}{\pnmarking}$ be the final configuration of
    $\runsym_\idxi$.  We extend $\runsym_\idxi$ to build $\runsym_{\idxi+1}$ by
    considering the transition
    \[
        \config{\control_\idxi}{\tree_\idxi} \tran
        \config{\control_{\idxi+1}}{\tree_{\idxi+1}}
    \]
    of the run of $\gtrs$.  As in the above deconstruction, there are two cases
    to consider.
    \begin{enumerate}
        \item 
            When the transition does not generate a new independent sub-tree,
            there are two further cases depending on whether a control state
            change occurs.
            \begin{enumerate}
                \item
                    When $\control_\idxi = \control_{\idxi+1}$ we simply define
                    $\runsym_{\idxi+1} = \runsym_\idxi$.

                \item
                    When $\control_\idxi \neq \control_{\idxi+1}$ we perform a
                    resolution step.  That is, we define $\runsym_{\idxi+1}$ to
                    be a run
                    \[
                        \runsym_\idxi \pnrun \config{\pnsumseq'}{\pnmarking'}
                    \]
                    where, recalling $\runsym_\idxi$ ends with the configuration
                    $\config{\pnsumseq}{\pnmarking}$, we have
                    $\config{\pnsumseq'}{\pnmarking'} =
                    \intres{\config{\pnsumseq}{\pnmarking}}$.  That such a run
                    \[
                        \runsym_\idxi \pntran
                        \config{\pnshiftst{\pnsumseq'}{1}}{\pnmarking} \pnrun
                        \config{\pnshiftst{\pnsumseq'}{\lifespan}}{\pnmarking'}
                        \pntran \config{\pnsumseq'}{\pnmarking'}
                    \]
                    of $\pnreach{\gtrs}$ exists follows in a straightforward
                    manner from the transitions in $\pnrulesres$.  Note that to
                    start the resolution, we require $\bit_2 = 1$.  This will
                    always be the case because exactly one independent sub-tree
                    is responsible for firing $\rruler$ and its interface was
                    added when the sub-tree was generated.  Similarly, the
                    compatibility conditions come from the construction of
                    $\istrun{\runlen}$ from a valid run.
            \end{enumerate}

        \item
            When a new independent sub-tree is generated, we again have two
            cases depending on whether the control state changes.  Let
            $\idxjnew$ be the ID of the new sub-tree.
            \begin{enumerate}
                \item 
                    When $\control_\idxi = \control_{\idxi+1}$ we proceed in the
                    same manner as we defined $\runsym_1$.  That is, take the
                    $\istigtrs$ defined by the rule $\rruler$ responsible for
                    the transition, and the sequence $\pnsumseq_I =
                    \tup{\control_1, \bit_1} \ldots \tup{\control_\num,
                    \bit_\num}$ where $\ap{\istint{\runlen}}{\idxjnew} =
                    \tup{\control_1, \bit_1, \istitrees_1} \ldots
                    \tup{\control_\num, \bit_\num, \istitrees_\num}$.  Again,
                    observe that by construction
                    $\ap{\istrun{\runlen}}{\idxjnew}$ gives a run of $\istigtrs$
                    outputting $\word$ such that $\parikh{\word} =
                    \tup{\istitrees_1, \ldots, \istitrees_\num}$.  From this we
                    obtain a run of the accompanying regular automaton
                    $\istirega$ with the same Parikh image and thus a run
                    \[
                        \config{\pnsumseq}{\pnmarking}
                        \pnlabtran{\decr{\istictr{\rruler}{1}}}
                        \config{\tup{\pnsumseq',
                        \istiinitst{\rruler_1}{\pnsumseq_I}}}{\submarking{\pnmarking}{\rruler}}
                        \pnrun \config{\tup{\pnsumseq',
                        \istifinst{\rruler_1}{\pnsumseq_I}}}{\pnmarking'}
                        \pntran \config{\pnsumseq'}{\pnmarking'} 
                    \]
                    of $\pnreach{\gtrs}$ where 
                    \[
                        \config{\pnsumseq'}{\pnmarking'} =
                        \intadd{\config{\pnsumseq}{\submarking{\pnmarking}{\rruler}}}{\ap{\istint{\runlen}}{\idxjnew}}
                        \ .
                    \]
                    That $\submarking{\pnmarking}{\rruler}$ exists follows from
                    the fact that by construction of $\runsym_\idxi$ we have
                    that $\ap{\pnmarking}{\istictr{\rruler}{1}}$ contains the
                    number of independent sub-trees generated by $\rruler$
                    between the change to $\control_\idxi$ and the change to the
                    next control state.  Hence, since we only decrement
                    $\istictr{\rruler}{1}$ when such a sub-tree is generated, we
                    do not fall below $0$.

                \item
                    When $\control_\idxi \neq \control_{\idxi+1}$ we peform a
                    resolution, a subtraction and an addition in the same way as
                    the previous cases.  That is, we build a run
                    \begin{align*}
                        \runsym_\idxi &\pntran
                        \config{\pnshiftst{\pnsumseq^1}{1}}{\pnmarking} \pnrun
                        \config{\pnshiftst{\pnsumseq^1}{\lifespan}}{\pnmarking^1}
                        \pntran \config{\pnsumseq^1}{\pnmarking^1} \\
                        &\pntran \config{\tup{\pnsumseq^2,
                        \istiinitst{\rruler_1}{\pnsumseq_I}}}{\submarking{\pnmarking^1}{\rruler}}
                        \pnrun \config{\tup{\pnsumseq^2,
                        \istifinst{\rruler_1}{\pnsumseq_I}}}{\pnmarking^2}
                        \pntran \config{\pnsumseq^2}{\pnmarking^2} 
                    \end{align*}
                    where $\config{\pnsumseq^1}{\pnmarking^1} =
                    \intres{\config{\pnsumseq}{\pnmarking}}$ and 
                    \[
                        \config{\pnsumseq^2}{\pnmarking^2} =
                        \intadd{\config{\pnsumseq^1}{\submarking{\pnmarking^1}{\rruler}}}{\ap{\istint{\runlen}}{\idxjnew}}
                        \ .
                    \]
                    The existence of such a run follows by the same arguments
                    presented above.
            \end{enumerate}
    \end{enumerate}
    Thus, we construct $\runsym_\runlen$ to some configuration
    $\config{\pnsumseq}{\pnmarking}$ where, necessarily, $\pnzeromarking
    \pncovers \pnmarking$ and $\pnsumseq = \tup{\control_\num, 1}$ where by
    assumption $\control_\num = \controldest$.  This follows since $\pnsumseq$
    was built by additions of interfaces that did not pass beyond
    $\control_\num$ to some $\control_{\num+1}$, and the resolution steps
    occured at the points where the control state changed.  The second component
    of the tuple is $1$ since by construction some sub-tree interface was
    responsible for every control state change.  Thus, we are done.
\end{proof}

\subsubsection{From Reset Petri Nets to Senescent GTRS}

\begin{lemma} \label{lem:petri2gtrs}
    For a given senescent GTRS $\gtrs$ with lifespan $\lifespan$, control states
    $\controlinit$ and $\controldest$, and tree $\treeinit$, there is a lifespan
    restricted run 
    \[
        \config{\controlinit}{\treeinit} \tran \cdots \tran
        \config{\controldest}{\tree}
    \]
    for some $\tree$ of $\gtrs$ if there is a run
    \[
        \config{\tup{\controlinit, 1}}{\pninitmarking} \pnrun
        \config{\tup{\controldest, 1}}{\pnmarking}
    \]
    of $\pnreach{\gtrs}$ for some $\pnzeromarking \pncovers \pnmarking$.
\end{lemma}
\begin{proof}
    Take a run 
    \[
        \config{\pnsumseq_1}{\pnmarking_1} \pnrun
        \config{\pnsumseq_2}{\pnmarking_2} \pnrun \cdots \pnrun
        \config{\pnsumseq_\runlen}{\pnmarking_\runlen}
    \]
    of $\pnreach{\gtrs}$ where $\config{\pnsumseq_1}{\pnmarking_1}, \ldots,
    \config{\pnsumseq_\runlen}{\pnmarking_\runlen}$ are all configurations with
    a control state in $\pnsumseqs$ that occur on the run.  

    By induction from $\idxi = 1$ to $\idxi = \runlen$ we build a run of $\gtrs$
    witnessing the control state reachability property.  We assume by induction
    that for $\pnsumseq_\idxi = \tup{\control_1, \bit_1} \ldots
    \tup{\control_\num, \bit_\num}$ that for all $\num < \idxj \leq \lifespan$
    and $\rruler' \in \rules$, we have
    $\ap{\pnmarking_\idxi}{\istictr{\rruler'}{\idxj}} = 0$.  Moreover, we have
    contexts $\context_1, \ldots, \context_\num$ that are intuitively the trees
    corresponding to the runs of $\gtrs$ built so far (where $\context_\idxi$ is
    paired with the run over control state $\control_\idxi$), but where the
    independent sub-trees (represented by $\exttree$ from the runs of the
    $\istigtrs$ used to build the run) are context variables.  Since the runs
    were built to leave these independent sub-trees uninspected, they can be
    replaced by any tree matching the rule which generated them.  More
    precisely, we have contexts $\context_1, \ldots, \context_\num$ such that
    for all sequences $\treeseq{1}, \ldots, \treeseq{\num}$ where for all $1
    \leq \idxj \leq \num$ we have 
    \[ 
        \treeseq{\idxj} = \treeseq{\rruler_1, \idxj}, \ldots,
        \treeseq{\rruler_\numrules, \idxj} 
    \]
    and for all $\rruler$ 
    \[
        \treeseq{\rruler, \idxj} = \pntree{\rruler}{\idxj}{1}, \ldots,
        \pntree{\rruler}{\idxj}{\ap{\pnmarking}{\istictr{\rruler}{\idxj}}} \in
        \ap{\lang}{\ta}^{\ap{\pnmarking}{\istictr{\rruler}{\idxj}}}
    \]
    where $\ta$ is the tree automaton on the RHS of $\rruler$, we have a run
    $\runsym_\idxi$ of $\gtrs$ reaching a configuration
    \[
        \config{\control_1}{\csub{\context_1}{\treeseq{1}}} \ .
    \]
    and for all $1 < \idxj \leq \num$, if $\bit_\idxj = 1$, we have a run 
    \[
        \config{\control_{\idxj-1}}{\csub{\context_{\idxj-1}}{\treeseq{1},
        \ldots, \treeseq{\idxj-1}}} \run
        \config{\control_\idxj}{\csub{\context_\idxj}{\treeseq{1}, \ldots,
        \treeseq{\idxj}}}
    \]
    of $\gtrs$ and if $\bit_\idxj = 0$, we have a run 
    \[
        \config{\control_\idxj}{\csub{\context_{\idxj-1}}{\treeseq{1}, \ldots,
        \treeseq{\idxj-1}}} \run
        \config{\control_\idxj}{\csub{\context_\idxj}{\treeseq{1}, \ldots,
        \treeseq{\idxj}}} \ .
    \]
    That is, we have a run to
    $\config{\control_\num}{\csub{\context_\num}{\treeseq{1}, \ldots,
    \treeseq{\num}}}$ which may have single transition gaps where control state
    changes occur (in which case $\bit_\idxj = 0$).  
    
    In the base case we have $\pnsumseq_1 = \tup{\controlinit, 1}$ and
    $\pnmarking_1 = \pninitmarking$ and the induction hypothesis is satisfied by 
    \[
        \runsym_1 = \config{\controlinit}{\csub{\context_1}{\tree_1}}
    \]
    where $\context_1$ is the context containing only a root node labelled by a
    variable and $\tree_1$ is necessarily $\treeinit$.

    Now we consider the inductive step.  Take the run of $\pnreach{\gtrs}$
    \[
        \config{\pnsumseq_\idxi}{\pnmarking_\idxi} \pnrun
        \config{\pnsumseq_{\idxi+1}}{\pnmarking_{\idxi+1}} 
    \]
    and fix $\pnsumseq_\idxi = \tup{\control_1, \bit_1} \ldots
    \tup{\control_\num, \bit_\num}$ and $\context_1, \ldots, \context_\num$ by
    induction.  There are two cases: either the run uses transition in
    $\pnrulesadd$ or $\pnrulesres$.
    \begin{enumerate}
        \item 
            If the run uses $\pnrulesadd$ rules then we have a run of the form
            \begin{align*} 
                \config{\pnsumseq_\idxi}{\pnmarking_\idxi}
                \pnlabtran{\decr{\istictr{\rruler}{1}}} &
                \config{\tup{\pnsumseq_{\idxi+1},
                \istiinitst{\rruler}{\pnsumseq}}}{\submarking{\pnmarking_\idxi}{\rruler}}
                \\ 
                \pnrun & \config{\tup{\pnsumseq_{\idxi+1},
                \istifinst{\rruler}{\pnsumseq}}}{\pnmarking_{\idxi+1}} \pntran
                \config{\pnsumseq_{\idxi+1}}{\pnmarking_{\idxi+1}} 
            \end{align*}
            where $\pnsumseq_{\idxi+1} = \intadd{\pnsumseq_\idxi}{\pnsumseq}$.
            Moreover, for the $\istigtrs$ defined from $\rruler$ and $\pnsumseq
            = \tup{\control'_1, \bit'_1} \ldots \tup{\control'_{\num'},
            \bit'_{\num'}}$, we have a run with the Parikh image
            $\tup{\istitrees_1, \ldots, \istitrees_{\num'}}$ from
            $\config{\control'_1}{\tree}$ where $\tree$ is accepted by the RHS
            of $\rruler$ and the run has the interface 
            \[
                \tup{\control'_1, \bit'_1, \istitrees_1} \ldots
                \tup{\control'_{\num'}, \bit'_{\num'}, \istitrees_{\num'}}
            \]
            such that 
            \[
                \config{\pnsumseq_{\idxi+1}}{\pnmarking_{\idxi+1}} =
                \intadd{\config{\pnsumseq_\idxi}{\submarking{\pnmarking_\idxi}{\rruler}}}{\tup{\control'_1,
                \bit'_1, \istitrees_1} \ldots \tup{\control'_{\num'},
                \bit'_{\num'}, \istitrees_{\num'}}} \ .
            \]
            From the run over $\istigtrs$, by creating context variable nodes
            when $\tup{\rruler', \idxj}$ characters are output, we obtain
            contexts $\context'_1, \ldots, \context'_{\num'}$ such that for all
            $\treeseq{1}', \ldots, \treeseq{\num'}'$ where for all $1 \leq \idxj
            \leq \num$ we have 
            \[ 
                \treeseq{\idxj}' = \treeseq{\rruler_1, \idxj}', \ldots,
                \treeseq{\rruler_\numrules, \idxj}'
            \]
            and for all $\rruler'$ 
            \[
                \treeseq{\rruler', \idxj}' = \pntreealt{\rruler'}{\idxj}{1},
                \ldots, \pntreealt{\rruler'}{\idxj}{\istigen^{\rruler'}_\idxj}
                \in \ap{\lang}{\ta'}^{\istigen^{\rruler'}_\idxj}
            \]
            where $\ta'$ is the tree automaton on the RHS of $\rruler'$, we have
            a run 
            \[
                \config{\control'_1}{\tree} \run
                \config{\control'_1}{\csub{\context'_1}{\treeseq{1}'}}
            \]
            and for all $1 < \idxj \leq \num$, if $\bit'_\idxj = 1$, we have a
            run 
            \[
                \config{\control'_{\idxj-1}}{\csub{\context'_{\idxj-1}}{\treeseq{1}',
                \ldots, \treeseq{\idxj-1}'}} \run
                \config{\control'_\idxj}{\csub{\context'_\idxj}{\treeseq{1}',
                \ldots, \treeseq{\idxj}'}}
            \]
            of $\gtrs$ and if $\bit'_\idxj = 0$, we have a run 
            \[
                \config{\control'_\idxj}{\csub{\context'_{\idxj-1}}{\treeseq{1}',
                \ldots, \treeseq{\idxj-1}'}} \run
                \config{\control'_\idxj}{\csub{\context'_\idxj}{\treeseq{1}',
                \ldots, \treeseq{\idxj}'}} \ .
            \]
            Using $\context_1, \ldots, \context_\num$ and $\context'_1, \ldots,
            \context'_{\num'}$ we can establish the induction hypothesis for
            $(\idxi+1)$.  We write 
            \[
                \csub{\context_\idxj}{\ldots, \pntree{\rruler}{1}{1}, \ldots}
            \]
            to single out $\pntree{\rruler}{1}{1}$ in
            $\csub{\context_\idxj}{\treeseq{1}, \ldots, \treeseq{\idxj}}$.  Note
            that since the induction hypothesis holds for all
            $\pntree{\rruler}{1}{1}$ we can select it to match $\tree$ accepted
            by the RHS of $\rruler$ used in the definition of $\context'_1,
            \ldots \context'_{\num'}$ above.  Although there are some details we
            take care of below, the essential idea is to obtain new contexts,
            satisfying the induction hypothesis, by inserting $\context'_\idxj$
            in place of $\pntree{\rruler}{1}{1}$ in $\context_\idxj$.

            There are two (similar) cases, depending on whether $\num \leq
            \num'$.  In both cases we have that the first $\num_{min} =
            \minimumof{\num, \num'}$ tuples of $\pnsumseq_{\idxi+1}$ are 
            \[
                \tup{\control_1, \bit''_1} \ldots \tup{\control_{\num_{min}},
                \bit''_{\num_{min}}},
            \]
            where $\bit''_\idxj = 1$ iff either $\bit_1 = 1$ or $\bit'_1 = 1$.
            Note, since $\pnsumseq_\idxi$ and $\pnsumseq$ are compatible, there
            is agreement on the control states.  Let $\cvar$ be the context
            variable corresponding to the position of $\pntree{\rruler}{1}{1}$
            in each $\csub{\context_\idxj}{\treeseq{1}, \ldots,
            \treeseq{\idxj}}$ (using the same $\cvar$ in each context).  We
            write $\csub{\context_\idxj}{\ldots, \cvar, \ldots}$ to isolate this
            variable, leaving all other variables untouched.  To satisfy the
            induction for all $1 \leq \idxj \leq \num_{min}$ we take the context
            $\csub{\context_\idxj}{\ldots, \context'_\idxj, \ldots}$ (with a
            suitable variable ordering to make the comparison with
            $\pnmarking_{\idxi+1}$).

            Then, we first build $\runsym_{\idxi+1}$ by concatenating to
            $\runsym_\idxi$ the run 
            \[
                \config{\control_1}{\csub{\context_1}{\ldots,
                \pntree{\rruler}{1}{1}, \ldots}} \run
                \config{\control_1}{\csub{\context_1}{\ldots,
                \csub{\context'_1}{\treeseq{1}'}, \ldots}}
            \]
            by simply appending the run to $\context'_1$ above to
            $\runsym_\idxi$.

            Next, for all $1 < \idxj \leq \num_{min}$, if $\bit''_\idxj = 1$, we
            can, if $\bit_\idxj = 1$ (implying $\bit'_\idxj = 0$), build
            \begin{align*}
                \config{\control_{\idxj-1}}{\csub{\context_{\idxj-1}}{\ldots,
                \csub{\context'_{\idxj-1}}{\treeseq{1}', \ldots,
                \treeseq{\idxj-1}'}, \ldots}}  & \run \\
                \config{\control_\idxj}{\csub{\context_\idxj}{\ldots,
                \csub{\context'_{\idxj-1}}{\treeseq{1}', \ldots,
                \treeseq{\idxj-1}'}, \ldots}} & \run \\ 
                \config{\control_\idxj}{\csub{\context_\idxj}{\ldots,
                \csub{\context'_\idxj}{\treeseq{1}', \ldots, \treeseq{\idxj}'},
                \ldots}} &
            \end{align*}
            and if $\bit_\idxj = 0$ (implying $\bit'_\idxj = 1$), build
            \begin{align*}
                \config{\control_{\idxj-1}}{\csub{\context_{\idxj-1}}{\ldots,
                \csub{\context'_{\idxj-1}}{\treeseq{1}', \ldots,
                \treeseq{\idxj-1}'}, \ldots}}  & \run \\
                \config{\control_\idxj}{\csub{\context_{\idxj-1}}{\ldots,
                \csub{\context'_\idxj}{\treeseq{1}', \ldots, \treeseq{\idxj}'},
                \ldots}} & \run \\ 
                \config{\control_\idxj}{\csub{\context_\idxj}{\ldots,
                \csub{\context'_\idxj}{\treeseq{1}', \ldots, \treeseq{\idxj}'},
                \ldots}} & \ .
            \end{align*}
            The remaining case is when $\bit''_\idxj = \bit'_\idxj = \bit_\idxj
            = 0$ and we build
            \begin{align*}
                \config{\control_\idxj}{\csub{\context_{\idxj-1}}{\ldots,
                \csub{\context'_{\idxj-1}}{\treeseq{1}', \ldots,
                \treeseq{\idxj-1}'}, \ldots}}  & \run \\
                \config{\control_\idxj}{\csub{\context_{\idxj-1}}{\ldots,
                \csub{\context'_\idxj}{\treeseq{1}', \ldots, \treeseq{\idxj}'},
                \ldots}} & \run \\ 
                \config{\control_\idxj}{\csub{\context_\idxj}{\ldots,
                \csub{\context'_\idxj}{\treeseq{1}', \ldots, \treeseq{\idxj}'},
                \ldots}} & \ .
            \end{align*}

            The remainder of the cases are below.
            \begin{enumerate}
                \item
                    When $\num \leq \num'$\ldots we have 
                    \[
                        \pnsumseq_{\idxi+1} = \tup{\control_1, \bit''_1} \ldots
                        \tup{\control_\num, \bit''_\num}
                        \tup{\control'_{\num+1}, \bit'_{\num+1}} \ldots
                        \tup{\control'_{\num'}, \bit'_{\num'}}
                    \]
                    and it remains to define contexts for $(\num+1) \leq \idxj
                    \leq \num'$, which we set for each $\idxj$ to be 
                    \[
                        \csub{\context_\num}{\ldots, \context'_{\idxj}, \ldots}
                    \]
                    with the runs, when $\bit'_\idxj = 1$, 
                    \begin{align*}
                        \config{\control'_{\idxj-1}}{\csub{\context_\num}{\ldots,
                        \csub{\context'_{\idxj-1}}{\treeseq{1}', \ldots,
                        \treeseq{\idxj-1}'}, \ldots}} & \run \\
                        \config{\control'_\idxj}{\csub{\context_\num}{\ldots,
                        \csub{\context'_\idxj}{\treeseq{1}', \ldots,
                        \treeseq{\idxj}'}, \ldots}}  &
                    \end{align*} 
                    and if $\bit_\idxj = 0$, we have a run 
                    \begin{align*}
                        \config{\control'_\idxj}{\csub{\context_\num}{\ldots,
                        \csub{\context'_{\idxj-1}}{\treeseq{1}', \ldots,
                        \treeseq{\idxj-1}'}, \ldots}} & \run \\
                        \config{\control'_\idxj}{\csub{\context_\num}{\ldots,
                        \csub{\context'_\idxj}{\treeseq{1}', \ldots,
                        \treeseq{\idxj}'}, \ldots}}  & \ .
                    \end{align*} 

                \item
                    When $\num > \num'$ we have 
                    \[
                        \pnsumseq_{\idxi+1} = \tup{\control_1, \bit''_1} \ldots
                        \tup{\control_{\num'}, \bit''_{\num'}}
                        \tup{\control_{\num'+1}, \bit_{\num'+1}} \ldots
                        \tup{\control_{\num}, \bit'_{\num}}
                    \]
                    and it remains to define contexts for $(\num'+1) \leq \idxj
                    \leq \num'$, which we set for each $\idxj$ to be 
                    \[
                        \csub{\context_\idxj}{\ldots, \context_{\num'}, \ldots}
                    \]
                    with the runs, when $\bit_\idxj = 1$, 
                    \begin{align*}
                        \config{\control_{\idxj-1}}{\csub{\context_{\idxj-1}}{\ldots,
                        \csub{\context'_{\num'}}{\treeseq{1}', \ldots,
                        \treeseq{\num'}'}, \ldots}} & \run \\
                        \config{\control'_\idxj}{\csub{\context_\idxj}{\ldots,
                        \csub{\context'_{\num'}}{\treeseq{1}', \ldots,
                        \treeseq{\num'}'}, \ldots}}  &
                    \end{align*} 
                    and if $\bit_\idxj = 0$, we have a run 
                    \begin{align*}
                        \config{\control'_\idxj}{\csub{\context_{\idxj-1}}{\ldots,
                        \csub{\context'_{\num'}}{\treeseq{1}', \ldots,
                        \treeseq{\num'}'}, \ldots}} & \run \\
                        \config{\control'_\idxj}{\csub{\context_\idxj}{\ldots,
                        \csub{\context'_{\num'}}{\treeseq{1}', \ldots,
                        \treeseq{\num'}'}, \ldots}}  & \ .
                    \end{align*} 
            \end{enumerate}
            To check that the above defined contexts have the right number of
            variables to be in accordance with $\pnmarking_{\idxi+1}$ one only
            need observe that we replaced the tree $\pntree{\rruler}{1}{1}$ with
            the new context (matching $\submarking{\pnmarking_\idxi}{\rruler}$),
            then inserted the new contexts ($\context'_\idxj$) corresponding to
            the addition of $\tup{\istitrees_1, \ldots, \istitrees_{\num'}}$ new
            trees, matching 
            \[
                \config{\pnsumseq_{\idxi+1}}{\pnmarking_{\idxi+1}} =
                \intadd{\config{\pnsumseq_\idxi}{\submarking{\pnmarking_\idxi}{\rruler}}}{\tup{\control'_1,
                \bit'_1, \istitrees_1} \ldots \tup{\control'_{\num'},
                \bit'_{\num'}, \istitrees_{\num'}}} \ .  
            \]

        \item
            If the run uses $\pnrulesres$ then we have a run 
            \[
                \config{\pnsumseq_\idxi}{\pnmarking_\idxi} \pntran
                \config{\pnshiftst{\pnsumseq_{\idxi+1}}{1}}{\pnmarking_\idxi}
                \pnrun
                \config{\pnshiftst{\pnsumseq_{\idxi+1}}{\lifespan}}{\pnmarking_{\idxi+1}}
                \pntran \config{\pnsumseq_{\idxi+1}}{\pnmarking_{\idxi+1}} 
            \]
            where $\pnsumseq_\idxi = \tup{\control_1, \bit_1} \tup{\control_2,
            \bit_2} \ldots \tup{\control_\num, \bit_\num}$ and
            $\pnsumseq_{\idxi+1} = \tup{\control_2, \bit_2} \ldots
            \tup{\control_\num, \bit_\num}$  and $\bit_2 = 1$.  Furthermore, for
            all $\rruler \in \rules$ it is the case that for all $1 \leq \idxj <
            \num$ we have $\ap{\pnmarking_{\idxi+1}}{\istictr{\rruler}{\idxj}}
            \leq \ap{\pnmarking_\idxi}{\istictr{\rruler}{\idxj+1}}$ and for all
            $\num \leq \idxj \leq \lifespan$ we have
            $\ap{\pnmarking_{\idxi+1}}{\istictr{\rruler}{\idxj}} = 0$.

            We obtain $\context'_2, \ldots, \context'_\num$ to satisfy the
            induction from the $\context_1, \ldots, \context_\num$ we have by
            the induction hypothesis.  For each rule $\rruler$, we can fix a
            tree $\tree_\rruler$ that is accepted by the RHS of $\rruler$ (we
            made this benign assumption at the beginning of the section).  Thus,
            intuitively, since $\context'_\idxj$ has fewer holes than
            $\context_\idxj$ we can simply plug each $\context_\idxj$ with the
            $\tree_\rruler$ to obtain $\context'_\idxj$.
            
            Thus we define for all $2 \leq \idxj \leq \num$ the context
            $\context'_\idxj$ such that for all $\treeseq{2}', \ldots,
            \treeseq{\num}'$ with for all $2 \leq \idxj' \leq \num$
            \[ 
                \treeseq{\idxj'}' = \treeseq{\rruler_1, \idxj'}', \ldots,
                \treeseq{\rruler_\numrules, \idxj'}'
            \]
            and for all $\rruler$ 
            \[
                \treeseq{\rruler, \idxj'}' = \pntreealt{\rruler}{\idxj'}{1},
                \ldots,
                \pntreealt{\rruler}{\idxj'}{\ap{\pnmarking_{\idxi+1}}{\istictr{\rruler}{\idxj'-1}}}
                \in
                \ap{\lang}{\ta}^{\ap{\pnmarking_{\idxi+1}}{\istictr{\rruler}{\idxj'-1}}}
            \]
            where $\ta$ is the tree automaton on the RHS of $\rruler$ we have
            \[
                \csub{\context'_\idxj}{\treeseq{2}', \ldots, \treeseq{\num}'} =
                \csub{\context_\idxj}{\treeseq{1}, \ldots, \treeseq{\num}}
            \] 
            where $\treeseq{1}, \ldots, \treeseq{\num}$ is given by
            \[ 
                \treeseq{1} = \treeseq{\rruler_1, 1}, \ldots,
                \treeseq{\rruler_\numrules, 1}
            \]
            and for all $\rruler$ 
            \[
                \treeseq{\rruler, 1} = \tree_\rruler, \ldots, \tree_\rruler \in
                \ap{\lang}{\ta}^{\ap{\pnmarking_\idxi}{\istictr{\rruler}{1}}}
            \]
            and for all $2 \leq \idxj' \leq \num$
            \[ 
                \treeseq{\idxj'} = \treeseq{\rruler_1, \idxj'}, \ldots,
                \treeseq{\rruler_\numrules, \idxj'}
            \]
            and for all $\rruler$ 
            \[
                \treeseq{\rruler, \idxj'} = \pntreealt{\rruler}{\idxj'}{1},
                \ldots,
                \pntreealt{\rruler}{\idxj'}{\ap{\pnmarking_\idxi}{\istictr{\rruler}{\idxj'}}},
                \tree_\rruler, \ldots, \tree_\rruler \in
                \ap{\lang}{\ta}^{\ap{\pnmarking_\idxi}{\istictr{\rruler}{\idxj'}}}
            \]
            where $\ta$ is the tree automaton on the RHS of $\rruler$.

            To check that $\context'_2, \ldots, \context'_\num$ satisfy the
            induction hypothesis we first construct $\runsym_{\idxi+1}$
            combining $\runsym_\idxi$ with (since $\bit_2 = 1$) 
            \[
                \config{\control_1}{\csub{\context_1}{\treeseq{1}}} \run
                \config{\control_2}{\csub{\context'_2}{\treeseq{2}'}}
            \]
            observing that the trees in $\treeseq{1}$ are out of the scope of
            the quantification and thus fixed.  The existence of partial runs
            for all $1 < \idxj \leq \num$ follows directly from the definition
            of $\context'_2, \ldots, \context'_\num$ and the induction
            hypothesis (substituting $\tree_\rruler$ where appropriate as
            above).
    \end{enumerate}
    Thus we are able to maintain the induction hypothesis.  When we consider
    $\config{\pnsumseq_\runlen}{\pnmarking_\runlen}$ we thus get from
    $\runsym_\runlen$ and $\pnsumseq_\runlen = \tup{\controldest, 1}$ a run
    \[
        \config{\controlinit}{\treeinit} \run
        \config{\controldest}{\csub{\context_1}{\treeseq{1}}}
    \]
    for some $\context_1$ and any appropriate sequence $\treeseq{1}$ (of which
    one necessarily exists by assumption).  This witnesses the reachability
    property as required.  
\end{proof}

\section{Conclusion}

We have introduced a sub-class of ground tree rewrite systems with state that
has a decidable reachability problem.  Our sub-class, \emph{senescent ground
tree rewrite systems}, takes scope-bounded pushdown systems as inspiration.  In
this setting, a node of the tree ``ages'' whenever the control state changes.  A
node that reaches a fixed age without being rewritten becomes fossilised and
thus may no longer be changed.  This model generalises weakly extended ground
tree rewrite systems by allowing an arbitrary number of control state changes.

We showed that the control state reachability problem is inter-reducible to
coverability of reset Petri-nets, and is thus $\acker$-complete.  This is a
surprising increase in complexity compared to scope-bounded multi-pushdown
systems for which the analogous problem is PSPACE-complete.  

Thus, we obtain a natural model that captures a rich class of behaviours while
maintaining decidability of reachability properties.  Moreover, since extending
the control state reachability problem to the regular reachability problem
results in undecidability, we know we are close to the limits of decidability.

For future work, we would like to investigate the encoding of additional classes
of multi-stack pushdown systems (e.g. ordered, phased-bounded, and relaxed
notions of scope-bounding, as well as with added features such as dynamic thread
creation) into senescent GTRS.  This may lead to further generalisations of our
model.

It has been shown by tools such as FAST~\cite{BFLP08} and TREX~\cite{ABS01} that
high (even undecidable) complexities do not preclude the construction of
successful model checkers.  Hence, we would like to study practical verification
algorithms for our model and their implementation, which may use the
aforementioned tools as components.

\paragraph{Acknowledgements}

We are grateful for helpful and informative discussions with Anthony Lin,
Sylvain Schmitz, Christoph Haase, and Arnaud Carayol.  This work was supported
by the Engineering and Physical Sciences Research Council [EP/K009907/1].

\bibliographystyle{plain} 
\bibliography{references}

\begin{thebibliography}{10}

\bibitem{AJMdO02}
P.~A. Abdulla, B.~Jonsson, P.~Mahata, and J.~d'Orso.
\newblock Regular tree model checking.
\newblock In {\em CAV}, pages 555--568, 2002.

\bibitem{ABS01}
A.~Annichini, A.~Bouajjani, and M.~Sighireanu.
\newblock Trex: A tool for reachability analysis of complex systems.
\newblock In {\em CAV}, pages 368--372, 2001.

\bibitem{AK77}
T.~Araki and T.~Kasami.
\newblock {Some Decision Problems Related to the Reachability Problem for Petri
  Nets}.
\newblock {\em Theoretical Computer Science}, 3(1):85--104, 1977.

\bibitem{ABH08}
M.~F. Atig, B.~Bollig, and P.~Habermehl.
\newblock Emptiness of multi-pushdown automata is 2etime-complete.
\newblock In {\em Developments in Language Theory}, pages 121--133, 2008.

\bibitem{ABQ11}
M.~F. Atig, A.~Bouajjani, and S.~Qadeer.
\newblock Context-bounded analysis for concurrent programs with dynamic
  creation of threads.
\newblock {\em Logical Methods in Computer Science}, 7(4), 2011.

\bibitem{BLR11}
T.~Ball, V.~Levin, and S.~K. Rajamani.
\newblock A decade of software model checking with slam.
\newblock {\em Commun. ACM}, 54(7):68--76, 2011.

\bibitem{BR00}
T.~Ball and S.~K. Rajamani.
\newblock Bebop: A symbolic model checker for boolean programs.
\newblock In {\em SPIN}, pages 113--130, 2000.

\bibitem{BFLP08}
S.~Bardin, A.~Finkel, J.~Leroux, and L.~Petrucci.
\newblock Fast: acceleration from theory to practice.
\newblock {\em STTT}, 10(5):401--424, 2008.

\bibitem{BEM97}
A.~Bouajjani, J.~Esparza, and O.~Maler.
\newblock Reachability analysis of pushdown automata: Application to
  model-checking.
\newblock In {\em CONCUR}, pages 135--150, 1997.

\bibitem{BKRS09}
L.~Bozzelli, M.~Kret\'{\i}nsk{\'y}, V.~Reh{\'a}k, and J.~Strejcek.
\newblock On decidability of ltl model checking for process rewrite systems.
\newblock {\em Acta Inf.}, 46(1):1--28, 2009.

\bibitem{B69}
W.~S. Brainerd.
\newblock Tree generating regular systems.
\newblock {\em Information and Control}, 14(2):217--231, 1969.

\bibitem{BCCC96}
L.~Breveglieri, A.~Cherubini, C.~Citrini, and S.~Crespi-Reghizzi.
\newblock Multi-push-down languages and grammars.
\newblock {\em Int. J. Found. Comput. Sci.}, 7(3):253--292, 1996.

\bibitem{CGK12}
A.~Cyriac, P.~Gastin, and K.~N. Kumar.
\newblock {MSO} decidability of multi-pushdown systems via split-width.
\newblock In {\em CONCUR}, pages 547--561, 2012.

\bibitem{DHLT90}
M.~Dauchet, T.~Heuillard, P.~Lescanne, and S.~Tison.
\newblock Decidability of the confluence of finite ground term rewrite systems
  and of other related term rewrite systems.
\newblock {\em Inf. Comput.}, 88(2):187--201, 1990.

\bibitem{DT90}
M.~Dauchet and S.~Tison.
\newblock The theory of ground rewrite systems is decidable.
\newblock In {\em LICS}, pages 242--248, 1990.

\bibitem{EKS03}
J.~Esparza, A.~Kucera, and S.~Schwoon.
\newblock Model checking ltl with regular valuations for pushdown systems.
\newblock {\em Inf. Comput.}, 186(2):355--376, 2003.

\bibitem{E95}
C.~P. Estes.
\newblock {\em The Faithful Gardener: A Wise Tale About That Which Can Never
  Die}.
\newblock Tree clause book. HarperCollins, 1995.

\bibitem{FWW97}
A.~Finkel, B.~Willems, and P.~Wolper.
\newblock A direct symbolic approach to model checking pushdown systems.
\newblock In {\em INFINITY}, volume~9, pages 27--37, 1997.

\bibitem{GS66}
S.~Ginsburg and E.~H. Spanier.
\newblock Semigroups, presburger formulas, and languages.
\newblock {\em Pacific Journal of Mathematics}, 16:285 -- 296, 1966.

\bibitem{GL11}
S.~G{\"o}ller and A.~W. Lin.
\newblock Refining the process rewrite systems hierarchy via ground tree
  rewrite systems.
\newblock In {\em CONCUR}, pages 543--558, 2011.

\bibitem{H13}
C.~Haase.
\newblock Subclasses or presburger arithmetic and the weak exponential-time
  hierarchy.
\newblock Under submission.

\bibitem{JM77}
N.~D. Jones and S.~S. Muchnick.
\newblock Even simple programs are hard to analyze.
\newblock {\em J. ACM}, 24(2):338--350, 1977.

\bibitem{KM67}
R.~M. Karp and R.~E. Miller.
\newblock Parallel program schemata: A mathematical model for parallel
  computation.
\newblock In {\em SWAT (FOCS)}, pages 55--61. IEEE Computer Society, 1967.

\bibitem{lTN11}
S.~{La Torre} and M.~Napoli.
\newblock Reachability of multistack pushdown systems with scope-bounded
  matching relations.
\newblock In {\em CONCUR}, pages 203--218, 2011.

\bibitem{lTN12}
S.~{La Torre} and M.~Napoli.
\newblock A temporal logic for multi-threaded programs.
\newblock In {\em IFIP TCS}, pages 225--239, 2012.

\bibitem{lTP12}
S.~{La Torre} and G.~Parlato.
\newblock Scope-bounded multistack pushdown systems: Fixed-point,
  sequentialization, and tree-width.
\newblock In {\em FSTTCS}, pages 173--184, 2012.

\bibitem{L12}
A.~W. Lin.
\newblock Weakly-synchronized ground tree rewriting - (with applications to
  verifying multithreaded programs).
\newblock In {\em MFCS}, pages 630--642, 2012.

\bibitem{L03}
C.~L\"oding.
\newblock {\em Infinite Graphs Generated by Tree Rewriting}.
\newblock PhD thesis, RWTH Aachen, 2003.

\bibitem{MP11}
P.~Madhusudan and G.~Parlato.
\newblock The tree width of auxiliary storage.
\newblock In {\em POPL}, pages 283--294, 2011.

\bibitem{M00}
M.~Maidl.
\newblock The common fragment of ctl and ltl.
\newblock In {\em FOCS}, pages 643--652, 2000.

\bibitem{M98}
R.~Mayr.
\newblock {\em Decidability and Complexity of Model Checking Problems for
  Infinite-State Systems}.
\newblock PhD thesis, TU-M\"unchen, 1998.

\bibitem{M84}
K.~McAloon.
\newblock Petri nets and large finite sets.
\newblock {\em Theoretical Computer Science}, 32:173--183, 1984.

\bibitem{P11}
R.~Piskac.
\newblock {\em Decision Procedures for Program Synthesis and Verification}.
\newblock PhD thesis, Laboratoire d'Analyse et de Raisonnement Authomatis\'es,
  \'Ecole Polytechnique F\'ed\'erale de Lausanne, 2011.

\bibitem{P91}
L.~Pottier.
\newblock Minimal solutions of linear diophantine systems: Bounds and
  algorithms.
\newblock In {\em RTA}, pages 162--173, 1991.

\bibitem{Q08}
S.~Qadeer.
\newblock The case for context-bounded verification of concurrent programs.
\newblock In {\em Proceedings of the 15th international workshop on Model
  Checking Software}, SPIN '08, pages 3--6, Berlin, Heidelberg, 2008.
  Springer-Verlag.

\bibitem{B64}
{R. B\"uchi}.
\newblock Regular canonical systems.
\newblock {\em Archiv fur Math. Logik und Grundlagenforschung 6}, pages
  91--111, 1964.

\bibitem{RSJM05}
T.~W. Reps, S.~Schwoon, S.~Jha, and D.~Melski.
\newblock Weighted pushdown systems and their application to interprocedural
  dataflow analysis.
\newblock {\em Sci. Comput. Program.}, 58(1-2):206--263, 2005.

\bibitem{QR05}
{S. Qadeer} and {J. Rehof}.
\newblock Context-bounded model checking of concurrent software.
\newblock In {\em TACAS}, pages 93--107, 2005.

\bibitem{SS12}
S.~Schmitz and P.~Schnoebelen.
\newblock {Algorithmic Aspects of WQO Theory}.
\newblock http://cel.archives-ouvertes.fr/cel-00727025, August 2012.

\bibitem{S02}
P.~Schnoebelen.
\newblock Verifying lossy channel systems has nonprimitive recursive
  complexity.
\newblock {\em Inf. Process. Lett.}, 83(5):251--261, 2002.

\bibitem{S10}
P.~Schnoebelen.
\newblock Revisiting ackermann-hardness for lossy counter machines and reset
  petri nets.
\newblock In {\em MFCS}, pages 616--628, 2010.

\bibitem{S02b}
S.~Schwoon.
\newblock {\em Model-checking Pushdown Systems}.
\newblock PhD thesis, Technical University of Munich, 2002.

\bibitem{T10}
A.~W. To.
\newblock {\em Model Checking Infinite-State Systems: Generic and Specific
  Approaches}.
\newblock PhD thesis, LFCS, School of Informatics, University of Edinburgh,
  2010.

\bibitem{lTMP07}
S.~La Torre, P.~Madhusudan, and G.~Parlato.
\newblock A robust class of context-sensitive languages.
\newblock In {\em LICS}, pages 161--170, 2007.

\end{thebibliography}

\end{document}